\documentclass{article}
\usepackage[margin=1in]{geometry}    

\usepackage[T1]{fontenc}
\usepackage[utf8]{inputenc}
\usepackage{amsmath,amssymb,amsthm,mathtools}
\usepackage[cal=euler]{mathalpha}
\usepackage{cochineal}
\usepackage{microtype}
\usepackage{bm}
\usepackage{booktabs}
\usepackage{array}
\usepackage{multirow}
\usepackage{longtable}
\usepackage{graphicx}
\usepackage[shortlabels]{enumitem}
\usepackage[numbers,sort&compress]{natbib}
\let\cite\citep
\usepackage[colorlinks=true,linkcolor=blue,citecolor=blue,urlcolor=blue]{hyperref}

\allowdisplaybreaks

\newtheorem{theorem}{Theorem}[section]
\newtheorem{proposition}[theorem]{Proposition}

\newtheorem{corollary}[theorem]{Corollary}
\theoremstyle{definition}

\DeclareMathOperator{\Poi}{Poi}
\DeclareMathOperator{\Bin}{Bin}
\DeclareMathOperator{\Mult}{Mult}
\DeclareMathOperator{\Dir}{Dir}
\DeclareMathOperator{\DirMult}{DirMult}
\DeclareMathOperator{\DTM}{DTM}

\DeclareMathOperator{\DNM}{DNM}
\DeclareMathOperator{\NegBin}{NB}

\DeclareMathOperator{\BetaBin}{BetaBin}

\DeclareMathOperator{\sgn}{sgn}
\DeclareMathOperator{\KL}{KL}
\DeclareMathOperator*{\argmin}{arg\,min}

\newcommand{\R}{\mathbb{R}}
\newcommand{\N}{\mathbb{N}}
\newcommand{\E}{\mathbb{E}}
\newcommand{\Var}{\mathrm{Var}}
\newcommand{\Cov}{\mathrm{Cov}}

\newcommand{\abs}[1]{\left|#1\right|}
\newcommand{\paren}[1]{\left(#1\right)}
\newcommand{\braces}[1]{\left\{#1\right\}}
\newcommand{\bracks}[1]{\left[#1\right]}
\newcommand{\ind}{\mathbf{1}}
\providecommand{\authcmt}[2]{\textcolor{#1}{#2}}
\providecommand{\vac}[1]{\authcmt{blue}{[VAC: #1]}}
\providecommand{\akshay}[1]{\authcmt{red}{[AKSHAY: #1]}}
\providecommand{\nnz}{\mathrm{nnz}}

\newcolumntype{L}[1]{>{\raggedright\arraybackslash}p{#1}}

\numberwithin{equation}{section}
\title{Deviance-style normalization for jointly overdispersed counts}
\author{Akshay Balsubramani \\ {\small \texttt{akshay@vac.bio}}}
\date{}

\renewcommand{\vac}[1]{}
\renewcommand{\akshay}[1]{}
\begin{document}

\maketitle

\begin{abstract}
We introduce a Dirichlet--multinomial (DM) deviance residualization for sparse, jointly overdispersed count matrices, the regime that dominates sequencing-based biochemical assays.
The DM null treats each sample's count vector as a fixed-total composition with a single scalar concentration $\alpha_0$ governing overdispersion, and arises exactly by conditioning independent negative-binomial feature counts on the observed sample total --- making the DM the joint conditional analogue of standard feature-wise overdispersed count models.
The resulting transform preserves exact sparsity, evaluates in constant time per nonzero entry, agrees with multinomial residuals on singleton counts, shrinks repeated-count residuals according to the overdispersion the null tolerates, and recovers the multinomial residual as $\alpha_0\to\infty$.
The same fixed-dispersion comparison principle extends to ordered and tree-structured features via the generalized DM and the Dirichlet-tree multinomial, giving a single residual family that subsumes joint and feature-wise count nulls under a common compositional logic and is computationally lightweight enough to drop into existing sparse pipelines.
\end{abstract}

\tableofcontents

\section{Modeling multivariate count data}
\label{sec:introduction}

Count data are ubiquitous in modern biological measurement, especially in sequencing-based assays.
A sample is usually a $K$-dimensional nonnegative integer vector: the $j$th entry records how many observed events were assigned to feature $j$.
Stacking $n$ samples gives a count matrix $X\in\N^{n\times K}$, with each sample being a row of $X$.

The row total and the feature allocation play different roles.
The total $n_i=\sum_{j=1}^K X_{ij}$ often reflects sequencing depth, exposure, or sampling effort.
Conditional on that total, the scientific question is how the observed events were distributed across features.
A useful normalization should respect this fixed-total, discrete, and sparse structure rather than immediately replacing the matrix by a dense continuous surrogate.

A principled way to do this is to normalize by residuals from a simple count null model.
The familiar Gaussian $z$-score is the simplest example: fit a mean and variance, then keep the signed discrepancy from that fitted null.
Deviance residuals extend the same idea to non-Gaussian exponential-family models by measuring likelihood discrepancy rather than squared error \cite{MacCullaghNelder1989,efron_2022}.
For sparse counts, this likelihood-based route is attractive because the transform can be derived directly from the count model.

Existing count normalizations tend to choose one of two nulls.
Feature-wise Poisson or negative-binomial models allow each feature to be overdispersed, but they do not model the row as a joint allocation of a fixed total.
The multinomial model does model the row jointly, but it is too rigid when repeated counts are more common than multinomial sampling predicts \cite{townes2019feature,AhlmannEltzeHuber2023,lause2021analytic}.
The Dirichlet--multinomial (DM) sits between these choices: it is an overdispersed multinomial, it preserves the compositional interpretation of the row sum, and it arises exactly by conditioning independent negative-binomial feature counts on the observed sample total.

That conditional perspective is the main point of the paper.
Technical variation in biochemical count data often looks like overdispersed reallocation around a shared baseline composition.
A multinomial null can misread that tolerated burstiness as signal, while feature-wise negative-binomial residuals can ignore the competition among features induced by the fixed row total.
The DM is the smallest joint model that addresses both issues at once.
It separates composition from concentration, keeps the row-sum conditioning explicit, and leaves zeros as zeros.

This paper develops a practical residual normalization based on that model.
The default null is
\[
X_i\mid n_i \sim \DirMult(n_i,\alpha_0\bm\pi)
\]
where $\bm\pi$ is the empirical global composition and $\alpha_0$ is a single fitted concentration parameter.
For each sample, the comparison model keeps $\alpha_0$ fixed and replaces the null composition by the sample's empirical composition.
This gives a closed-form sparse transform that can be evaluated on the nonzero entries alone.

The main contributions are threefold.
First, we derive the DM transform, its fitting equations, and its basic properties: exact sparsity preservation, agreement with multinomial residuals on singleton counts, shrinkage relative to multinomial residuals on repeated counts, and convergence to the multinomial transform as $\alpha_0\to\infty$.
Second, we show that conditioning independent negative-binomial feature counts on the row sum gives exactly the same DM null, which explains why the method is the joint conditional analogue of common overdispersed count residualization.
Third, we extend the same fixed-dispersion comparison principle to structured count spaces, including ordered features through the generalized Dirichlet--multinomial and known hierarchies through the Dirichlet-tree multinomial.

Our evaluation compares the new transform to multinomial residuals, practical feature-wise negative-binomial residuals, independent negative-binomial baselines before and after row-sum conditioning, dense compositional baselines, and structured alternatives on biochemical count matrices.

\subsection{The proposed normalization}
\label{sec:ataglance}

The normalization is a single pass over a sparse count matrix.
Here $\DirMult(n,\bm\alpha)$ denotes the Dirichlet--multinomial distribution with total $n$ and concentration vector $\bm\alpha$; its probability mass function is given in \eqref{eq:dmpmf}.

Given a sparse count matrix $X$:
\begin{itemize}[noitemsep]
\item Compute the row sums $n_i$, the column sums $a_j$, and the global null composition $\pi_j=a_j/N$
\item Fit the single concentration $\hat\alpha_0$ of the conditional null $X_i\mid n_i \sim \DirMult(n_i,\hat\alpha_0\bm\pi)$ by the fixed-point iteration
\[
\alpha_0^{\mathrm{new}}
=
\alpha_0\,
\frac{
\sum_{i,j}\pi_j\bracks{\psi(X_{ij}+\alpha_0\pi_j)-\psi(\alpha_0\pi_j)}
}{
\sum_i\bracks{\psi(n_i+\alpha_0)-\psi(\alpha_0)}
}
\]
with $\psi$ the digamma function; this keeps $\alpha_0>0$, increases the likelihood at every step, and converges in tens of iterations from any positive start (Section~\ref{sec:alpha0} gives the derivation and three further fitting variants)
\item For each nonzero entry $x=X_{ij}$, compute
$
 c_{ij}^{\mathrm{DM}}
 =
 \sum_{m=0}^{x-1} \log \left( \frac{ \hat\alpha_0 x/n_i + m}{\hat\alpha_0 \pi_j + m} \right)
$ 
\item Output the signed residuals
$\displaystyle
 d_{ij}^{\mathrm{DM}}
 =
 \sgn(x-n_i\pi_j)\sqrt{\abs{c_{ij}^{\mathrm{DM}}}}
$
leaving every zero entry implicit.
\end{itemize}

Equivalently, $d_{ij}^{\mathrm{DM}}=0$ if $X_{ij}=0$, and for every nonzero entry $X_{ij}$,
\begin{equation}
\label{eq:oneliner}
\boxed{
d_{ij}^{\mathrm{DM}}
=
\sgn(X_{ij} - n_i\pi_j)\;\sqrt{\abs{
\sum_{m=0}^{X_{ij}-1}\log\left(\frac{\hat\alpha_0 X_{ij}/n_i + m}{\hat\alpha_0\pi_j+m}\right)
}}
}
\end{equation}

The transformed matrix can then be used where one would otherwise use a normalized surrogate of $X$: linear embeddings, nearest-neighbor graphs, clustering, regression, or other downstream algorithms that search for structure in the sample-by-feature matrix.
The coordinates also retain an explicit likelihood-contrast interpretation under a joint overdispersed fixed-total count null.

Thus the transform needs only one fitted scalar beyond the empirical composition, preserves sparsity exactly, and can be dropped into the same downstream workflows used for multinomial or negative-binomial residual normalization.
The comparison is deliberately fixed-dispersion: for each sample we let the composition match the observed row, but we keep the fitted concentration parameter $\alpha_0$ fixed.
Later we show that the same DM distribution appears by conditioning an independent negative-binomial model on the observed row sum (Theorem~\ref{thm:nbtodm}), and that the same fixed-dispersion principle extends naturally to ordered features and tree-structured features (Sections~\ref{sec:gdm} and~\ref{sec:dtm}).

\subsection{Performance and calibration across benchmarks}
\label{sec:glance-results}

The DM is not only principled but effective. Across five real count regimes --- single-cell UMI, bulk tissue, chromatin accessibility, and antibody-capture counts --- it improves on the multinomial null everywhere (Figure~\ref{fig:real-data-benchmarks}), at the per-nonzero cost of a single log-gamma evaluation. It is also, among the count-model normalizations, the \emph{only} one that calibrates a parametric differential-expression test: cell-type-wide, the Dirichlet--multinomial null holds the false-positive rate at its nominal level while the Poisson-family nulls are anti-conservative, and under injected differential expression it controls the false-discovery proportion where the multinomial inflates it (Figure~\ref{fig:E08971}). Section~\ref{sec:experiments} develops both comparisons in full.

\begin{figure}[t]
\centering
\includegraphics[width=\textwidth]{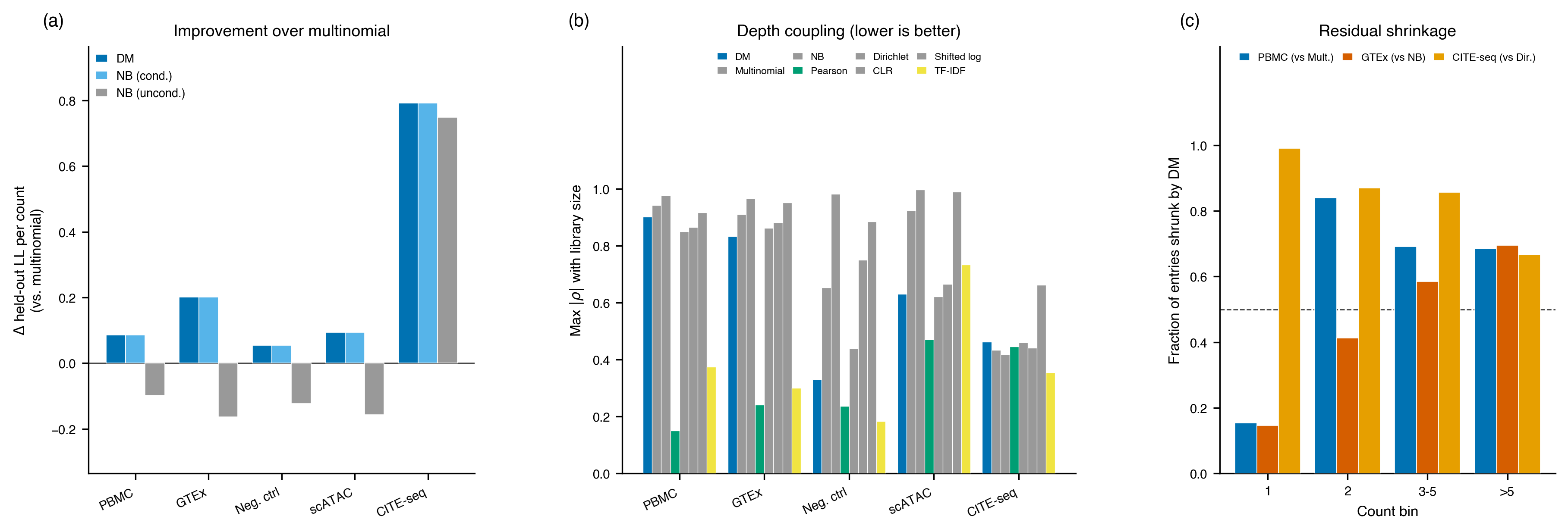}
\caption{\textbf{Real-data benchmark summary across five real count regimes.}
All five datasets are real: PBMC~3k (10x raw UMI), GTEx
\cite{Eraslan2022} (a 20{,}000-cell subsample of the GTEx
cross-tissue atlas), negative controls (PBMC~3k restricted to 20
housekeeping genes), scATAC (10x atac\_pbmc\_500\_nextgem peak counts),
and CITE-seq (10x pbmc\_1k\_protein\_v3 antibody-capture counts).
\textbf{(a)}~Improvement in held-out per-observed-count conditional
log-likelihood over the multinomial baseline, for the DM (highlighted)
and the unconditional and conditional negative-binomial models.
The DM improves on the multinomial everywhere, and the gain is largest
on overdispersed data (CITE-seq, gain $0.79$ from $-0.81$ to $-0.02$;
GTEx, gain $0.20$ from $-2.92$ to $-2.72$); the conditional NB tracks the
DM exactly, while the unconditional NB falls below the multinomial on
several datasets.
\textbf{(b)}~Depth coupling, lower is better: maximum absolute Pearson
correlation between the leading 10 PCA axes and library size, across
eight transforms.
TF-IDF achieves the lowest coupling by construction;
among the count-based methods the DM and multinomial are comparable, and
the analytical Pearson residuals use the per-cell
$\mu_{ij} = n_i \pi_j$ (not the marginal mean) and achieve low coupling
on most datasets.
\textbf{(c)}~Residual shrinkage by count bin: the fraction of nonzero
entries the DM shrinks relative to its paired baseline (multinomial on
PBMC, feature-wise NB on GTEx, dense pseudocount Dirichlet on CITE-seq),
stratified into counts $1$, $2$, $3$--$5$, and $>5$.
Against a shared $\bm\pi$, Proposition~\ref{prop:boundedbymult}
guarantees $|d^{\mathrm{DM}}|\le|d^{\mathrm{Mult}}|$ at every nonzero
entry with strict shrinkage on every repeated count, so against the
multinomial and the NB the shrunk-fraction starts small at singletons
(the gap is exactly zero at count~$1$ by Proposition~\ref{prop:x1}) and
rises with count magnitude; the Dirichlet comparison shrinks a large
majority of entries at every bin, singletons included, because the
Dirichlet baseline's unit pseudocount inflates low-count residuals.
All diagnostics use the training-fit global composition $\bm\pi$ and the
fitted scalar $\hat\alpha_0$.}
\label{fig:real-data-benchmarks}
\end{figure}

\begin{figure}[t]
\centering
\includegraphics[width=\linewidth]{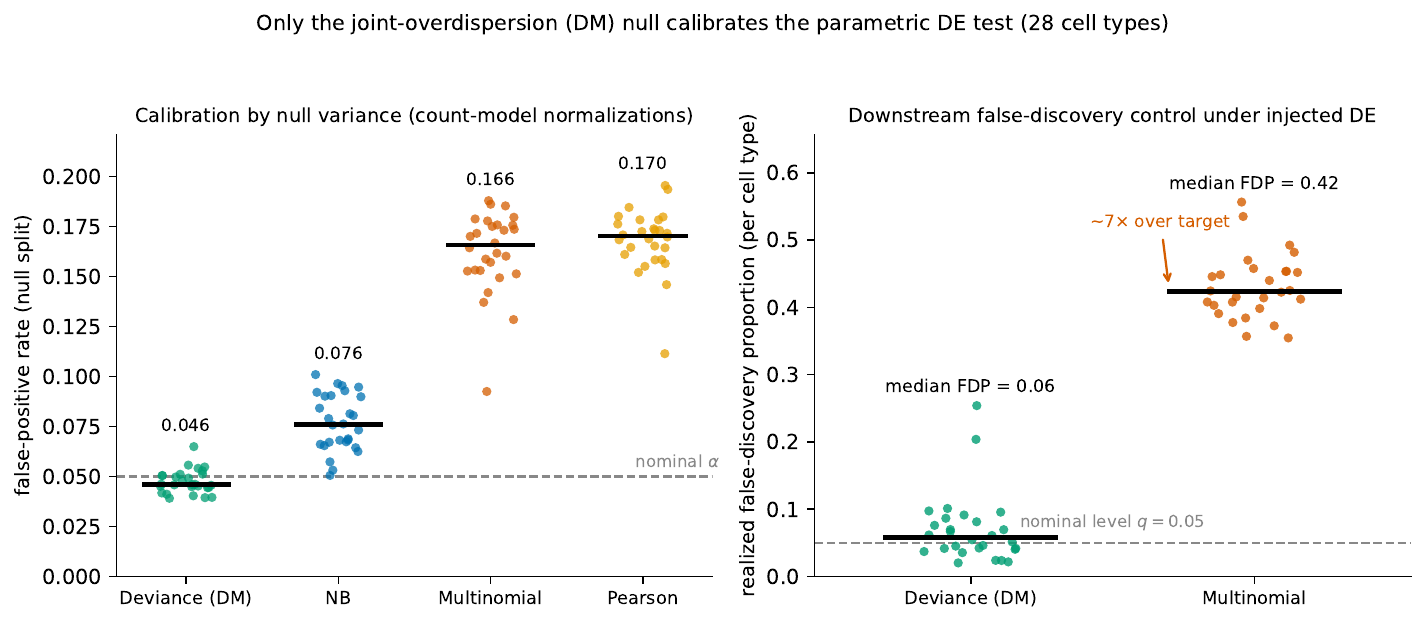}
\caption{The parametric two-group test calibrates only under the joint-overdispersion
(Dirichlet--multinomial) null, across every broad cell type of a cross-tissue atlas; the
Poisson-family (multinomial, Pearson) null is anti-conservative and the per-gene negative-binomial
null intermediate. \emph{Left:} under a null random split, the false-positive rate of the log-mean
test at nominal $\alpha=0.05$ for each count-model normalization's null variance (one point per cell
type, black bar the median across the twenty-eight cell types); the multinomial and analytic Pearson
nulls reject null genes about three to four times too often (median false-positive rates $0.166$ and
$0.170$), the per-gene negative-binomial null halves the inflation but does not remove it ($0.076$),
and the Dirichlet--multinomial null is calibrated ($0.046$, the unique calibrated method, best in
$28/28$). \emph{Right:} under injected differential expression, the realized false-discovery
proportion per cell type at the nominal Benjamini--Hochberg level $q=0.05$, for the two count-model
normalizations run in the downstream test (the Dirichlet--multinomial null and the multinomial null;
the negative-binomial and Pearson nulls are not part of this downstream comparison). The
Dirichlet--multinomial null holds the false-discovery proportion near the level it claims (median
$0.058$) while the multinomial null inflates it roughly sevenfold (median $0.423$), at matched
power.}
\label{fig:E08971}
\end{figure}
\subsection{Prior work and relation to existing normalization methods}
\label{sec:priorwork}

High-dimensional count matrices are standard throughout genomics.
RNA-seq differential-expression workflows are built around overdispersed negative-binomial models with empirical-Bayes shrinkage and related dispersion regularization \cite{RobinsonMcCarthySmyth2010,anders2010differential,LoveHuberAnders2014}.
At single-cell resolution, droplet-based assays produce sparse gene-by-cell matrices of unique molecular identifier counts \cite{Macosko2015}.
These matrices are now the default input for clustering, trajectory analysis, integration, and multimodal representation learning in widely used toolchains \cite{Satija2015,stuart2019comprehensive}.
Count-valued observations also arise at other biological resolutions, including genes, isoforms, peaks, taxa, barcodes, and genomic intervals \cite{Trapnell2010}.

The statistical issues are similar across these settings.
Counts are heterogeneous across samples, sparse across features, and often more variable than a Poisson model would predict.
This has motivated feature-wise negative-binomial regression, residual-based normalization, and variance-stabilizing transforms \cite{HafemeisterSatija2019,ChoudharySatija2022,lause2021analytic,AhlmannEltzeHuber2023}.
A complementary line of work fits a joint multinomial null directly to the sparse count matrix and uses multinomial deviance residuals derived from a single shared composition for feature selection and dimension reduction \cite{townes2019feature}.
The DM transform of this paper is the natural overdispersed extension of that direction: in the multinomial limit $\alpha_0\to\infty$ the DM residual reduces to the multinomial deviance residual \cite{townes2019feature}, while finite $\alpha_0$ accommodates the repeated-count regime where a multinomial null is empirically too rigid.
In parallel, compositional data analysis has emphasized that sample totals induce dependence among features and that log-ratio geometry is often appropriate once one works with dense positive compositions \cite{Aitchison1982,Gloor2017}.
The difficulty is that sparse count matrices contain many exact zeros, so dense log-ratio transforms require pseudocounts or imputation and therefore forfeit the exact sparsity pattern of the data.

The multinomial and Dirichlet--multinomial are especially natural fixed-total models.
Conditional on each sample total $n_i$, the multinomial is the basic joint model of the sample composition, while the DM adds one scalar degree of freedom controlling overdispersion relative to the multinomial \cite{Mosimann1962,ThorsonJohnsonMethotTaylor2017,HarrisonCalderShastryBuerkle2020}.
Bayesian compositional models based on the DM have recently been useful for differential abundance analysis in single-cell and microbiome applications \cite{BuettnerOstnerMuellerTheisSchubert2021,WadsworthArgientoGuindaniEtAl2017,Koslovsky2023}.
What has been missing is a sparse residual transform that uses the same family as a practical normalization null.

\subsection{Notation and setup}
\label{sec:notation}

We write $X_{ij}$ for the count in sample $i$ and feature $j$, and
$
n_i := \sum_{j=1}^K X_{ij},
\qquad
a_j := \sum_{i=1}^n X_{ij},
\qquad
N := \sum_{i=1}^n \sum_{j=1}^K X_{ij}$. 
The empirical global composition is
\[
\pi_j := \frac{a_j}{N},
\qquad j=1,\dots,K
\]
Throughout, we treat the sample totals $n_i$ as observed nuisance scales rather than as targets of the null model. 
The normalization problem is therefore to model how each sample allocates its observed total across features.

Some practical conventions are worth stating explicitly. 
First, all-zero samples and all-zero features carry no likelihood information for the conditional models considered here, so they may be removed prior to fitting. 
Second, because many of the transforms below are intended for sparse matrices, we prefer formulas that can be evaluated only on nonzero entries and then stored again in sparse form whenever possible.

\section{The Dirichlet--multinomial residual transform}
\label{sec:dm}

This section derives the transform stated in Section~\ref{sec:ataglance}.
Multinomial residuals are joint and compositional, but they do not model overdispersion.
Feature-wise negative-binomial residuals model overdispersion, but they do so one feature at a time.
The DM is the simplest fixed-total model that does both at once: it is joint, conditional on the observed row sum, and overdispersed.
\subsection{Where the DM sits among common normalization models}
\label{sec:tour}

There are two broad perspectives on modeling matrices of multivariate counts.

The first is \emph{feature-wise}. 
For each feature $j$, we model $X_{ij}$ separately across samples using a binomial, Poisson, or negative-binomial null. 
These models are straightforward, interpretable, and often effective; they also lead to classical sample-wise deviance residuals that are easy to compute. 
Their main limitation is structural: they treat the features as unrelated once they condition on covariates and sample-specific total $n_i$.

The second is \emph{joint}. 
Here the entire row vector $X_i = (X_{i1},\dots,X_{iK})$ is modeled conditioned on the sample total $n_i$. 
The multinomial is the baseline joint model.  
Its overdispersed analogue is the Dirichlet--multinomial. 
These models respect the compositional structure directly and capture the negative dependence induced by allocating a fixed total count across many features.

For normalization, the null model should be simple enough to estimate cheaply, robustly, and interpretably without destroying any (more complex) signal, while being rich enough to remove the most obvious nuisance structure. 
That is why the DM is an appealing null. 
It can be fit robustly, while still capturing repeated-count behavior that the multinomial misses. 
It also turns out to be the conditional version of an independent negative-binomial model, which makes its relation to standard NB residualization especially clear.

\subsection{Poisson is to multinomial as negative binomial is to DM}

\label{sec:conditionalmotivation}

The conditional viewpoint is classical for the multinomial.

\begin{theorem}[Independent Poisson counts condition to a multinomial]
\label{thm:poitomult}
Fix a sample $i$ and suppose $X_{ij} \sim \Poi(\lambda_{ij})$ independently for $j=1,\dots,K$. 
Then, conditional on $n_i = \sum_{j=1}^K X_{ij}$,
\[
(X_{i1},\dots,X_{iK}) \mid n_i
\sim
\Mult\!\left(n_i,\left\{\frac{\lambda_{ij}}{\sum_{\ell=1}^K \lambda_{i\ell}}\right\}_{j=1}^K\right)
\]
\end{theorem}

The analogous statement for overdispersed marginals is the cleanest entry point to the DM.
Once independent overdispersed feature counts are conditioned on the observed sample total, the joint model is Dirichlet--multinomial.

\begin{theorem}[Independent negative-binomial counts condition to a Dirichlet--multinomial]
\label{thm:nbtodm}
Fix a sample $i$ and suppose for any $j=1,\dots,K$ that 
\[
X_{ij} \sim \NegBin(r_j, p_i)
\qquad \iff \qquad 
\text{Pr}(X_{ij}=x)
=
\frac{\Gamma(x+r_j)}{\Gamma(r_j)\Gamma(x+1)}
(1-p_i)^x p_i^{r_j}
\]
independently for integer $x$. 
Then, conditional on $n_i = \sum_{j=1}^K X_{ij}$,
\[
(X_{i1},\dots,X_{iK}) \mid n_i
\sim
\DirMult(n_i,\bm r),
\qquad
\bm r=(r_1,\dots,r_K)
\]
\end{theorem}

Theorem~\ref{thm:nbtodm} shows that the success probability $p_i$ is a nuisance parameter once one conditions on the observed sample total. 
The conditional distribution depends only on the concentration vector $\bm r$. 
This is the precise sense in which the DM is the joint overdispersed analogue of the multinomial \cite{Zhou2018,GuimaraesLindrooth2007}.

\subsection{The DM is an overdispersed multinomial}
\label{sec:dmbasic}

The Dirichlet--multinomial distribution has been used as an overdispersed multinomial since its introduction by \cite{Mosimann1962}. 
For $x \in \N^K$ with $x_1+\cdots+x_K=n$ and $\bm\alpha \in \mathbb{R}_{+}^{K}$ with $\alpha_0:=\sum_{j=1}^K \alpha_j$, the probability mass function of the DM is
\begin{equation}
\label{eq:dmpmf}
\text{Pr}(X=x \mid n,\bm\alpha)
=
\frac{\Gamma(n+1)\Gamma(\alpha_0)}{\Gamma(n+\alpha_0)}
\prod_{j=1}^K
\frac{\Gamma(x_j+\alpha_j)}{\Gamma(x_j+1)\Gamma(\alpha_j)}
\end{equation}
Its mean is $\E[X_j] = n \frac{\alpha_j}{\alpha_0}$, and its expressions for variance and covariance resemble the multinomial with the same composition $p_j = \alpha_j/\alpha_0$: 
\begin{align*}
\Var(X_j) = n p_j \paren{1- p_j}\frac{n+\alpha_0}{1+\alpha_0} \qquad \qquad 
\Cov(X_j,X_\ell) = -n p_j p_{\ell} \frac{n+\alpha_0}{1+\alpha_0}, 
\qquad j\neq \ell
\end{align*}
Relative to the multinomial with parameters $p_j$, the covariances are inflated by the factor
$\frac{n+\alpha_0}{1+\alpha_0}$. 
Thus $\alpha_0$ controls how close the DM is to the multinomial. 
As $\alpha_0 \to \infty$, the DM approaches the multinomial; for smaller $\alpha_0$, the model becomes more overdispersed. 
In this precise limiting sense, the DM contains the multinomial as the no-overdispersion boundary case. 

\subsubsection{Overdispersed marginals}
\label{sec:overdispersedmarginals}

It is useful to distinguish two different ``marginal'' statements about the DM.

\begin{enumerate}[(i)]
\item 
For a \emph{Dirichlet--multinomial} random vector with fixed total $n$, the univariate marginals are \emph{beta-binomial}, not negative-binomial. 
Only under additional rare-event or large-total scalings do these marginals resemble negative-binomial behavior. 
\item 
The direct connection to the negative binomial comes from the \emph{conditional construction} in Theorem~\ref{thm:nbtodm}: independent negative-binomial feature counts become Dirichlet--multinomial after conditioning on the observed total.
\end{enumerate}

This distinction matters conceptually. 
When we condition on $n_i$, we are explicitly discarding a source of sample-to-sample variation in the total count and retaining only the joint allocation of that total across features.

\subsubsection{Conditioning on sample totals}
\label{sec:conditioning}

In sequencing and other biochemical assays, the row sum $n_i$ usually reflects sampling depth or exposure rather than the feature signal of primary interest.
Conditioning on $n_i$ therefore removes a nuisance scale and asks only how that observed total was allocated across features.
This has several consequences.
As we have shown, it makes the multinomial and Dirichlet--multinomial the natural joint models of composition, because they describe exactly how a fixed total is distributed across features. 
In addition, it cleanly separates \emph{composition} from \emph{depth}. 
In the DM parameterization, $\bm\alpha$ controls both the mean composition and the amount of overdispersion. 
For normalization, it is therefore natural to write
\[
\bm\alpha = \alpha_0 \bm\pi
\qquad
\bm\pi \in \Delta^{K-1}
\]
so that $\bm\pi$ determines the null composition and $\alpha_0$ determines the dispersion around it.

The full DM family certainly allows arbitrary feature-specific $\alpha_j$. 
In this paper we intentionally adopt the tied form $\alpha_j=\alpha_0\pi_j$ for the null model. 
This separates composition from concentration, yields a transparent interpretation, and reduces fitting to a one-dimensional problem in $\alpha_0$ once $\bm\pi$ has been fixed.

\subsubsection{Urn models and burstiness}
\label{sec:urn}

Overdispersion in discrete data is often associated with burstiness or reinforcement: once a feature is observed, it becomes more likely to be observed again. 
P\'olya urn schemes make this intuition precise.

If an urn contains balls of $K$ colors in proportions proportional to $\alpha_1,\dots,\alpha_K$, and each draw is replaced without reinforcement, then $n$ draws produce a multinomial count vector. 
If instead each observed ball is replaced together with an additional ball of the same color, then $n$ reinforced draws produce a Dirichlet--multinomial count vector with parameter $\bm\alpha$ \cite{Mosimann1962}. 
This ``rich classes get richer'' viewpoint is a convenient generative picture for the heavier tails of the DM.

The Dirichlet-negative-multinomial admits an analogous urn construction with one additional ``stopping'' color; we return to it in Section~\ref{sec:dnm}.

\subsection{The proposed DM null model}
\label{sec:framework}

We now model each sample vector by
\begin{equation}
\label{eq:dmnullmodel}
X_i \mid n_i,\bm\alpha \sim \DirMult(n_i,\bm\alpha),
\qquad
\bm\alpha = \alpha_0 \bm\pi
\end{equation}
where the null composition is set to the empirical global composition
$\pi_j = \frac{a_j}{N}$. 

Under \eqref{eq:dmnullmodel}, the log-likelihood is
\begin{equation}
\label{eq:dmll}
\ell(\bm\alpha)
=
\sum_{i=1}^n
\left[
\log\frac{\Gamma(n_i+1)\Gamma(\alpha_0)}{\Gamma(n_i+\alpha_0)}
+
\sum_{j=1}^K
\log\frac{\Gamma(X_{ij}+\alpha_j)}{\Gamma(X_{ij}+1)\Gamma(\alpha_j)}
\right]
\end{equation}

A key question is what comparison model should define the residuals. 
In classical deviance calculations one compares the fitted null to the most flexible model that still lives inside the chosen family. 
For the unrestricted DM family that idea is unhelpful: once both composition and concentration are free, the maximizing sequence drives the model toward the empirical multinomial.

\begin{proposition}[Unrestricted DM saturation collapses to multinomial saturation]
\label{prop:dmsaturation}
Fix a sample count vector $x\in\N^K$ with total $n$.
Then
\[
\sup_{\bm\alpha>0}\text{Pr}_{\DirMult(n,\bm\alpha)}(X=x)
=
\text{Pr}_{\Mult(n,x/n)}(X=x)
\]
where $x/n$ denotes the empirical composition, interpreted on its support if some coordinates of $x$ are zero. 
Equivalently, the unrestricted saturated likelihood over the full DM family is attained only in the limit $\bm\alpha=t\,x/n$ as $t\to\infty$.
\end{proposition}

So a literal saturated deviance for the unrestricted DM family is not the object we want: the comparison would immediately collapse back to the multinomial and erase the DM's overdispersion correction.
Instead we work with the fixed-concentration family
\[
\{\DirMult(n_i,\alpha_0 q): q\in\Delta^{K-1}\}
\]
Our default comparison uses the empirical row composition $q_i=X_i/n_i$ inside this family.
This is the composition-matching model: it gives the observed row exactly the observed mean composition while retaining the fitted concentration $\alpha_0$ as the common dispersion level.

For finite $\alpha_0$, this empirical-composition comparison should not be confused with the exact maximum-likelihood fit of $q$ inside the fixed-concentration DM family.
The exact rowwise optimizer solves a nonlinear equation and does not generally equal $X_i/n_i$.
We use $q_i=X_i/n_i$ deliberately because it is closed form, has the multinomial comparison as its large-$\alpha_0$ limit, preserves exact sparsity, and aligns with the conditional independent-negative-binomial construction in Section~\ref{sec:matchednb}.
Thus, when we use a ``DM deviance-style residual'' in this paper, we mean the fixed-concentration empirical-composition likelihood contrast defined in Section~\ref{sec:dmresids} next.

\subsection{The proposed residual transform}
\label{sec:dmresids}

For sample $i$, define the empirical composition
$q_{ij}:=X_{ij}/n_i$ for $j=1,\dots,K$.
We compare the fitted null $\DirMult(n_i,\alpha_0\bm\pi)$ to the sample-specific empirical-composition model $\DirMult(n_i,\alpha_0 q_i)$.
In other words, sample $i$ gets its own observed composition, but it keeps the same concentration parameter $\alpha_0$ as the null.
This is the conditional joint analogue of an overdispersed count comparison in which the mean allocation is allowed to match the observation while the dispersion scale remains fixed.

The sample-wise log-likelihood contrast of this comparison model against the null is
\begin{align*}
\Delta_i^{\mathrm{DM}}
&=
\sum_{j=1}^K
\left[
\log\frac{\Gamma(X_{ij}+\alpha^{\mathrm{cmp}}_{ij})}{\Gamma(\alpha^{\mathrm{cmp}}_{ij})}
-
\log\frac{\Gamma(X_{ij}+\alpha^{\mathrm{null}}_{j})}{\Gamma(\alpha^{\mathrm{null}}_{j})}
\right] \\
\alpha^{\mathrm{cmp}}_{ij} &= \alpha_0 q_{ij},
\qquad
\alpha^{\mathrm{null}}_j = \alpha_0\pi_j
\end{align*}
Coordinates with $X_{ij}=0$ are interpreted by continuity; their contribution is the empty sum below.
Using the identity
$\log\{\Gamma(x+a)/\Gamma(a)\}=\sum_{m=0}^{x-1}\log(a+m)$ for integer $x$, we obtain the entry-wise additive term
\begin{equation}
\label{eq:dmcellterm}
c_{ij}^{\mathrm{DM}}
:=
\sum_{m=0}^{X_{ij}-1}
\log\paren{\frac{\alpha_0 X_{ij}/n_i + m}{\alpha_0\pi_j + m}}
\end{equation}
These terms satisfy $\Delta_i^{\mathrm{DM}}=\sum_{j=1}^K c_{ij}^{\mathrm{DM}}$.
They need not all be nonnegative, because the comparison is joint and the entry-wise decomposition is an attribution of a row-level likelihood contrast.
Following the standard signed square-root convention used for multinomial residuals, we define
\begin{equation}
\label{eq:dmresid}
d_{ij}^{\mathrm{DM}}
:=
\sgn\paren{X_{ij}-n_i\pi_j}
\sqrt{\abs{c_{ij}^{\mathrm{DM}}}}
\end{equation}

For feature-wise Poisson and negative-binomial models, the squared deviance residual is a nonnegative per-entry contribution to a classical deviance.
For joint models such as the multinomial and the DM, row-level likelihood contrasts decompose into additive coordinate terms that can be signed.
The transform \eqref{eq:dmresid} should therefore be viewed as a signed square root of an additive likelihood term rather than as a literal nonnegative entry-wise deviance decomposition.
When we say ``DM residual'' below, we mean this fixed-concentration empirical-composition transform.

When $X_{ij}=0$, the sum in \eqref{eq:dmcellterm} is empty, so $c_{ij}^{\mathrm{DM}}=0$ and therefore
\[
X_{ij}=0 \implies d_{ij}^{\mathrm{DM}}=0
\]
This is the basic sparsity-preserving property of the transform.

\subsection{Basic properties}

\label{sec:dmproperties}

The DM residuals interpolate between two familiar regimes. 
On ultra-sparse data, where most nonzero entries are singletons, they are essentially multinomial. 
Their distinctive behavior appears on repeated counts, where the DM discounts deviations that a multinomial would treat as too surprising.

\begin{proposition}[DM and multinomial residuals agree on singletons]
\label{prop:x1}
If $X_{ij}=1$, then the DM and multinomial signed residual transforms coincide:
$\displaystyle
d_{ij}^{\mathrm{DM}}
=
\sgn(1-n_i\pi_j)\sqrt{\abs{\log\frac{1}{n_i\pi_j}}}
= 
d_{ij}^{\mathrm{Mult}}
$. 
\end{proposition}

\begin{proposition}[The DM shrinks large-count residuals toward zero]
\label{prop:boundedbymult}
The exact per-cell term is the closed form \eqref{eq:dmcellterm}, $c_{ij}^{\mathrm{DM}}=\sum_{m=0}^{x-1}\log\frac{\alpha_0 x/n_i+m}{\alpha_0\pi_j+m}$, and the multinomial term is its $\alpha_0\to\infty$ limit $c_{ij}^{\mathrm{Mult}}=x\log\frac{x}{n_i\pi_j}$ (Proposition~\ref{prop:multlimit}).
For every nonzero entry $X_{ij}=x$ these two exact quantities are ordered by magnitude,
$\abs{c_{ij}^{\mathrm{DM}}}
\le
\abs{c_{ij}^{\mathrm{Mult}}}
=
x \abs{\log\frac{x}{n_i\pi_j}}$,
with the same sign $\sgn(x-n_i\pi_j)$.
Consequently,
$
\abs{d_{ij}^{\mathrm{DM}}}
\le
\abs{d_{ij}^{\mathrm{Mult}}}
$.
\end{proposition}

\begin{proposition}[Multinomial limit]
\label{prop:multlimit}
For every fixed nonzero $X_{ij}=x$,
$\lim_{\alpha_0\to\infty} c_{ij}^{\mathrm{DM}}
=
x \log\frac{x}{n_i\pi_j}
$. 
So $d_{ij}^{\mathrm{DM}} \to d_{ij}^{\mathrm{Mult}}$ as $\alpha_0 \to \infty$.
\end{proposition}

Together these facts explain why the DM is a conservative overdispersed replacement for multinomial residuals. 
Singletons are unchanged, so the transform behaves almost identically to the multinomial on matrices dominated by zeros and ones. 
The distinction emerges exactly on repeated counts. 
There the DM term is a softened version of the multinomial term, with the amount of softening controlled by $\alpha_0$: large $\alpha_0$ recovers the multinomial, while smaller $\alpha_0$ discounts repeated counts more aggressively.

\subsection{Why keep the dispersion fixed in the comparison?}
\label{sec:nullchoice}

The null model says that all samples share one global composition $\bm\pi$ and one concentration level $\alpha_0$.
The comparison model changes only the composition of sample $i$, replacing $\bm\pi$ with the empirical row composition $q_i=X_i/n_i$, while keeping $\alpha_0$ fixed.
This is the modeling choice that lets the residual ask a focused question: after allowing the row to have its own composition, how much of the discrepancy from the global null remains attributable to the observed allocation rather than to the overdispersion level already tolerated by the null?

Proposition~\ref{prop:dmsaturation} explains why the concentration should not be re-fit freely inside the comparison.
If one re-fits the full DM family sample by sample, the optimum collapses to the empirical multinomial limit and the overdispersion correction disappears.
Fixing the concentration keeps the part of the model that is meant to explain burstiness, while still letting the row-level composition depart from the global baseline.

Other comparison rules are possible.
One could compute the exact fixed-concentration rowwise MLE of the composition, or choose a row-specific comparison concentration.%

Those variants are legitimate likelihood comparisons, but they lose the simple closed-form residual \eqref{eq:dmresid} and change the shrinkage behavior on repeated counts.
The empirical-composition comparison is the cleanest default: it is explicit, sparse, agrees with the multinomial transform in the relevant limits, and facilitates the connection to conditioned independent negative-binomial counts in Section~\ref{sec:matchednb}.

\subsection{Independent negative-binomial counts}
\label{sec:matchednb}

The preceding sections have established two extremes in the normalization landscape.
The multinomial models all features jointly but does not accommodate overdispersion; indeed, Proposition~\ref{prop:dmsaturation} shows that unrestricted DM saturation collapses back to the multinomial.
At the other extreme, independent negative-binomial models handle overdispersion naturally but treat features as unrelated.
The DM sits squarely between these two poles, and its connection to the independent negative-binomial model is especially clean: conditioning independent NB counts on the observed row sum yields exactly the DM.

Suppose that before conditioning on the row sum we model the feature counts as independent negative-binomial variables
\[
Y_{ij} \sim \NegBin(r_j,p_i),
\qquad j=1,\dots,K
\]
with a common sample-specific success probability $p_i$ and shape parameters $r_j=\alpha_0\pi_j$. 
Let $N_i:=\sum_{j=1}^K Y_{ij}$.
Then $N_i\sim\NegBin(\alpha_0,p_i)$ and, by Theorem~\ref{thm:nbtodm},
\[
Y_i\mid N_i=n_i \sim \DirMult(n_i,\alpha_0\bm\pi)
\]
So the DM is exactly what remains after conditioning this independent negative-binomial model on the observed total count.

\begin{proposition}[Independent negative-binomial likelihood factorization]
\label{prop:nbfactorization}
For any count vector $x_i$ with total $n_i=\sum_j x_{ij}$,
\[
\prod_{j=1}^K \text{Pr}_{\NegBin(r_j,p_i)}(Y_{ij}=x_{ij})
=
\text{Pr}_{\NegBin(\alpha_0,p_i)}(N_i=n_i)\;
\text{Pr}_{\DirMult(n_i,\alpha_0\bm\pi)}(X_i=x_i\mid n_i)
\]
Equivalently,
$\ell_i^{\mathrm{indNB}}(x_i;r,p_i)
=
\ell_i^{\mathrm{NB\ total}}(n_i;\alpha_0,p_i)
+
\ell_i^{\mathrm{DM}}(x_i\mid n_i;\alpha_0 \bm \pi)$. 
\end{proposition}

A convenient choice is
\[
p_i = \frac{\alpha_0}{n_i+\alpha_0}
\]
which makes $\E[N_i]=n_i$ under the null independent negative-binomial model. 
Now compare that null to the sample-specific independent negative-binomial model obtained by replacing
\[
r_j=\alpha_0\pi_j
\qquad\text{by}\qquad
r_{ij}^{\mathrm{cmp}} = \alpha_0\frac{X_{ij}}{n_i}
\]
while keeping the same $p_i$. 
Interpreting zero-count coordinates by continuity, the factor for the total count $N_i$ is then identical under the null and comparison models and therefore cancels.

\begin{corollary}[Independent negative-binomial and DM give the same per-sample likelihood contrast]
\label{cor:nbdmdeviance}
Under the choice $p_i=\alpha_0/(n_i+\alpha_0)$ and the comparison model $r_{ij}^{\mathrm{cmp}}=\alpha_0 X_{ij}/n_i$, the per-sample independent negative-binomial log-likelihood contrast equals the DM log-likelihood contrast:
\[
\Delta_i^{\mathrm{indNB}} = \Delta_i^{\mathrm{DM}}
\]
Equivalently, twice these contrasts agree sample by sample. 
\end{corollary}

The equality in Corollary~\ref{cor:nbdmdeviance} is \emph{sample-wise}, not entry-wise. 
Standard feature-wise NB residual normalization fits each feature separately, usually with feature-specific dispersions, and compares each entry against its own saturated mean. 
So it is not algebraically identical to the DM residual matrix. 
The exact comparison is at the level of per-sample likelihoods; the comparison at the level of residual matrices or downstream analyses is practical rather than algebraic.

\subsection{The Dirichlet-negative-multinomial}
\label{sec:dnm}

The negative multinomial is to the multinomial what the negative binomial is to the Poisson. 
Its overdispersed analogue is the Dirichlet-negative-multinomial (DNM) \cite{Mosimann1963,HausmanHallGriliches1984,FarewellFarewell2013}. 
One can think of the DNM as introducing an additional ``failure'' category, indexed by $0$, and sampling repeatedly until a prescribed number $x_0$ of failures have been observed.

The DNM is useful because it explains how variability in the total sample count can coexist with a DM-like allocation model for the observed categories. 
Once the sample total is conditioned on, however, the extra stopping parameters disappear.

\begin{proposition}[Conditioning a DNM on the observed total yields a DM]
\label{prop:dnm}
Suppose
$X_i \sim \DNM(x_{0,i},\alpha_{0,i},\bm\alpha)$. 
Then, conditional on $\sum_{j=1}^K X_{ij}=n_i$,
\[
X_i \mid \sum_{j=1}^K X_{ij}=n_i \sim \DirMult(n_i,\bm\alpha)
\]
\end{proposition}

From the viewpoint of normalization, Proposition~\ref{prop:dnm} is reassuring. 
If sample-to-sample depth variation is modeled through a DNM-type stopping mechanism, conditioning on the observed depth still reduces the problem to the same DM residual transform. 
The DNM and DM are therefore alternate parameterizations of the same conditional normalization problem. 

\subsection{Statistical testing interpretations}
\label{sec:testing}

One reason to work with deviance-style residuals rather than an ad hoc transform is that the underlying quantities remain likelihood based.
For each sample $i$,
\[
\Delta_i^{\mathrm{DM}}=\sum_{j=1}^K c_{ij}^{\mathrm{DM}}
\]
is the log-likelihood contrast between the empirical-composition comparison model and the fitted DM null, with $\alpha_0$ held fixed.
Summing these contrasts over rows or feature groups gives a well-defined decomposition of how the fitted null differs from the composition-matching comparison.

At finer resolution, the entry-wise terms $c_{ij}^{\mathrm{DM}}$ and magnitudes $(d_{ij}^{\mathrm{DM}})^2=\abs{c_{ij}^{\mathrm{DM}}}$ provide an attribution map of where that contrast comes from.
The signed terms tell us whether a coordinate pushes the row-level contrast up or down, while the squared residuals give the corresponding unsigned contribution size.
Because every entry is scored against the same fitted null, one can compare aggregated contributions across groups of rows, columns, or annotations.

For arbitrary feature subsets, these aggregates should be interpreted as decomposed pieces of a joint likelihood contrast rather than as standalone saturated deviances.
Even so, they are useful diagnostically: they localize which coordinates dominate the global departure from the null, quantify how much particular subsets contribute, and provide a principled scale on which to compare signal strength across samples, feature groups, or branches of a known hierarchy.

Concretely, for any index set $S\subseteq\{1,\dots,n\}\times\{1,\dots,K\}$, define the signed aggregate
\[
D_S:=2\sum_{(i,j)\in S}c_{ij}^{\mathrm{DM}}
\]
When $S$ is a single row, $D_{\{i\}\times[K]}$ is the row-level DM contrast.
When $S$ is a single column, $D_{[n]\times\{j\}}$ is the signed contribution of feature $j$ to the global contrast across samples.
When $S$ is a rectangular block of rows and columns, $D_S$ is the signed contribution of that block.
By contrast, $\sum_{(i,j)\in S}(d_{ij}^{\mathrm{DM}})^2$ is an unsigned attribution score that ignores sign cancellations.

For formal hypothesis testing, the natural objects are row-level or block-level likelihood ratios with a clearly specified alternative, preferably fit on held-out data or calibrated by resampling.
If one replaces the empirical-composition comparison by the exact maximum-likelihood alternative for a chosen block model, the usual large-sample Wilks-type calibration can be used as a guide under standard regularity conditions \cite{MacCullaghNelder1989,efron_2022}.
The individual entry-wise DM residuals themselves should be treated as descriptive coordinates: because the terms may be signed and are coupled by the row-sum constraint, they do not have a universal exact null law.

Finally, the difference $|d_{ij}^{\mathrm{Mult}}|-|d_{ij}^{\mathrm{DM}}|$ at each nonzero entry quantifies the residual magnitude that is ``explained away'' by overdispersion.
Entries where this gap is large are precisely the repeated-count coordinates where conditioning on the row sum and allowing joint overdispersion most changes the residual geometry relative to a multinomial baseline.

\subsection{DM contrasts as e-values and e-processes}
\label{sec:eprocess}

\paragraph{Background.}
This paragraph fixes the testing terminology used below.
An \emph{e-variable} for a null hypothesis $H_0$ is a non-negative random variable $E$ with $\mathbb{E}_{H_0}[E]\le 1$; observing the realization $E\ge 1/\alpha$ is finite-sample evidence against $H_0$ at level $\alpha$.
An \emph{e-process} is a non-negative process $\{M_T\}_{T\ge 1}$ adapted to a filtration that is bounded above in expectation by a martingale starting at $1$ under $H_0$; \emph{Ville's inequality} \cite{ramdas2023statistical} bounds $\mathbb{P}_{H_0}\!\left(\sup_{T\ge 1} M_T\ge 1/\alpha\right)\le\alpha$ uniformly over $T$, including data-dependent stopping.
A test is \emph{anytime-valid} when type-I error is controlled simultaneously over every stopping time $T$, equivalently when its running test statistic is an e-process.
\emph{Split conformal prediction} produces calibrated prediction sets by computing nonconformity scores on a held-out calibration split and inverting their quantiles.

The row-level contrast $\Delta_i^{\mathrm{DM}}=\sum_{j=1}^K c_{ij}^{\mathrm{DM}}$ is, by construction, a log-likelihood ratio between two specific count models in the fixed-concentration DM family.
Two distinct constructions arise, and only one of them is an e-variable.
If the comparison composition $q_i$ is fit from the same $X_i$ used to evaluate the likelihood ratio --- as in the empirical-composition contrast $q_i=X_i/n_i$ of Section~\ref{sec:dmresids} that drives the residual transform --- then the resulting statistic is a generalized likelihood ratio rather than a fixed-alternative likelihood ratio, and the e-variable bound $\mathbb{E}_{H_0}[\exp\Delta_i^{\mathrm{DM}}]\le 1$ does not hold in general.
The construction that does yield an e-variable splits the row: fit the comparison composition on one part of the row's counts and evaluate the likelihood ratio on the held-out remainder, in the manner of universal inference \cite{wasserman2020universal}.
Running products of split-LR e-variables over independent rows then form an e-process.
The next corollary states this precisely; combined with the preceding sections it gives a finite-sample, anytime-valid testing framework that uses the same DM family as the residual transform of Section~\ref{sec:dmresids}, with the row-split discipline added on top for the testing application.

\begin{corollary}[Split-LR DM-e-variable and DM-e-process]
\label{cor:dm-e-process}
Fix the global composition $\bm\pi$ and concentration $\alpha_0$.
Given a row $X_i$ with total $n_i$, split it by independent binomial thinning at rate $\tfrac12$: assign each of the $n_i$ unit counts independently to either the \emph{train} or the \emph{eval} side, yielding $X_i^{\mathrm{tr}}+X_i^{\mathrm{ev}}=X_i$ with totals $n_i^{\mathrm{tr}}+n_i^{\mathrm{ev}}=n_i$.
Define the train-side empirical composition $q_i^{\mathrm{tr}}\in\Delta^{K-1}$ by $q_{ij}^{\mathrm{tr}}=X_{ij}^{\mathrm{tr}}/n_i^{\mathrm{tr}}$ (any fixed tie-breaking convention when $n_i^{\mathrm{tr}}=0$, e.g.\ $q_i^{\mathrm{tr}}=\bm\pi$).
The eval side is then scored against the null law that the train side leaves behind.
Because the DM is a Dirichlet mixture of multinomials, conditioning on the train counts updates the null composition by Dirichlet--multinomial conjugacy, so the conditional null law of the eval counts is the posterior-predictive Dirichlet--multinomial $\DirMult(n_i^{\mathrm{ev}},\alpha_0\bm\pi+X_i^{\mathrm{tr}})$ rather than the marginal $\DirMult(n_i^{\mathrm{ev}},\alpha_0\bm\pi)$.
Define the \emph{split-LR DM e-variable}
\begin{equation}
\label{eq:dmesplit}
E_i^{\mathrm{DM,split}}
:=
\frac{p_{\DirMult(n_i^{\mathrm{ev}},\,\alpha_0 q_i^{\mathrm{tr}})}\bigl(X_i^{\mathrm{ev}}\mid n_i^{\mathrm{ev}}\bigr)}
     {p_{\DirMult(n_i^{\mathrm{ev}},\,\alpha_0\bm\pi+X_i^{\mathrm{tr}})}\bigl(X_i^{\mathrm{ev}}\mid n_i^{\mathrm{ev}}\bigr)}
\end{equation}
For any sample $X_i$ generated under the global null $X_i\mid n_i\sim\DirMult(n_i,\alpha_0\bm\pi)$, the statistic in \eqref{eq:dmesplit} satisfies the exact identity
\[
\mathbb{E}_{H_0}\!\left[E_i^{\mathrm{DM,split}}\right]=1,
\]
since it is the likelihood ratio of two proper probability mass functions on the eval simplex and the denominator (the posterior-predictive law) has full support there, so the numerator's mass integrates to one against it.
In particular $\mathbb{E}_{H_0}[E_i^{\mathrm{DM,split}}]\le 1$, which is all the e-variable interface requires, so $E_i^{\mathrm{DM,split}}$ is an e-variable for the global null.
For an \emph{independent} sequence of rows under the same null, each carrying its own thinning split, the running product
\[
M_T := \prod_{i=1}^T E_i^{\mathrm{DM,split}}
\qquad\text{satisfies}\qquad
\mathbb{P}\!\left(\sup_{T\ge 1} M_T \ge 1/\alpha\right)\le \alpha
\quad\text{for every }\alpha\in(0,1)
\]
i.e.\ $\{M_T\}_{T\ge 1}$ is an e-process.
The proof, which instantiates the split-likelihood-ratio construction of universal inference \cite{wasserman2020universal} on the binomially-thinned row, is given in Appendix~\ref{sec:apdxproofs} (Proof of Corollary~\ref{cor:dm-e-process}).
\end{corollary}

The row split is necessary for the following reason.
The in-sample plug-in $\exp(\Delta_i^{\mathrm{DM}})$ with $q_i=X_i/n_i$ of Section~\ref{sec:dmresids} is a generalized likelihood ratio rather than a fixed-alternative likelihood ratio, since the same $X_i$ supplies both the alternative composition and the likelihood evaluation; the e-variable bound $\mathbb{E}_{H_0}[\exp\Delta_i^{\mathrm{DM}}]\le 1$ then need not hold, and in regimes where the multinomial limit is approached it fails by orders of magnitude.
The binomial-thinning split above makes $X_i^{\mathrm{tr}}$ and $X_i^{\mathrm{ev}}$ conditionally independent given the row's underlying DM composition draw.
The two halves are not independent unconditionally, however: they share that latent composition, so conditioning on $X_i^{\mathrm{tr}}$ shifts the eval-side null toward the train-side allocation.
Scoring $X_i^{\mathrm{ev}}$ against this posterior-predictive null $\DirMult(n_i^{\mathrm{ev}},\alpha_0\bm\pi+X_i^{\mathrm{tr}})$---rather than the marginal $\DirMult(n_i^{\mathrm{ev}},\alpha_0\bm\pi)$---is exactly what restores the fixed-alternative likelihood-ratio bound, since against the correct conditional law $q_i^{\mathrm{tr}}$ acts as a deterministic alternative.
The cost is a power loss from evaluating the likelihood ratio on half the row's counts; the standard remedies of multiple-split aggregation or stratified splitting recover most of this loss.

The split-LR corollary makes three things available without re-deriving theory.
First, an \emph{anytime-valid} sequential test for the global null: applied to the running product $M_T=\prod_{i=1}^T E_i^{\mathrm{DM,split}}$ over a stream of \emph{independent} rows, stopping at the first $T$ with $M_T\ge 1/\alpha$ controls type-I error at most $\alpha$, even when $T$ is itself a stopping time.
Independence across rows is what makes $M_T$ a supermartingale; a sequence that is exchangeable but dependent (rows coupled through a shared latent composition) is not covered, because there the per-row conditional mean given the past can exceed one.
This is the count-data analogue of the post-hoc e-process algebra of \cite{howard2020b,ramdas2023statistical}, with the row split playing the role usually played by an external held-out fold.

Second, finite-sample valid testing without a Wilks-type asymptotic regularity assumption.
The split-LR row e-variable $E_i^{\mathrm{DM,split}}$ produces a universally valid p-value $p_i=\min(1, 1/E_i^{\mathrm{DM,split}})$ for ``row $i$ is anomalous,'' valid by the same Markov-bound construction as in \cite{wasserman2020universal}.
This applies at every $n_i$ for which the binomial split yields $n_i^{\mathrm{tr}},n_i^{\mathrm{ev}}\ge 1$, including the small-library-size regime where the Wilks chi-square calibration of $D_{\{i\}\times[K]}$ breaks down.

Third, the split-LR construction composes with downstream conformal scoring.
Per-row $E_i^{\mathrm{DM,split}}$ is a valid nonconformity statistic in split conformal prediction because the row-split discipline already supplies the calibration/evaluation separation that split conformal requires; concatenating an outer conformal calibration split with the inner per-row train/eval split preserves the finite-sample marginal coverage guarantee.
Running products of $E_i^{\mathrm{DM,split}}$ then form a sequential e-process suitable for chained multi-stage pipelines.

A practical caveat: $E_i^{\mathrm{DM,split}}$ is well-defined only when the binomial split assigns at least one count to each side, i.e.\ $n_i\ge 2$.
For singleton rows ($n_i=1$) one falls back to the standard fixed-alternative LR with $q_i^{\mathrm{tr}}=\bm\pi$, which contributes a degenerate factor of $1$ to the running product and so does not enter the test.
The same row-split discipline also explains why the in-sample plug-in $\exp(\Delta_i^{\mathrm{DM}})$ of Section~\ref{sec:dmresids} should not be used as an e-variable: it conflates the alternative-fitting and likelihood-evaluation roles of the row and is therefore a generalized likelihood ratio whose expectation is not bounded by one.
The exact identity $\mathbb{E}_{H_0}[E_i^{\mathrm{DM,split}}]=1$, the unbounded inflation of the in-sample plug-in, and the anti-conservatism of the asymptotic $\chi^2$ calibration on sparse counts are all borne out numerically, confirming that the split-LR construction---not a naive deviance test---is the route to anytime-valid DM evidence on the sparse data this paper targets (Section~\ref{sec:eval-E01250}).

This e-process bridge places the DM family inside the broader post-hoc and anytime-valid testing landscape \cite{howard2020b,ramdas2023statistical,wasserman2020universal}: the same fixed-concentration DM family that gives the closed-form residual transform of Section~\ref{sec:dmresids} also supplies, via the split-LR construction above, the row-level building block for sequential and finite-sample valid count testing.

\section{Extensions}
\label{sec:extensions}

\subsection{The generalized DM for ordered features}
\label{sec:gdm}

The ordinary DM treats the feature labels symmetrically.
That is the right default when the columns are an unordered feature set, but it is not always the right model.
Sometimes the categories are produced by a genuine sequence of decisions: allocate counts to the first category versus everything remaining, then allocate part of the remainder to the second category, and so on.
Genomic intervals along a fixed order, ordered cell states, or pre-specified sequential gates are examples where such an order may be scientifically meaningful.

The Connor--Mosimann generalized Dirichlet distribution and its multinomial mixture, the generalized Dirichlet--multinomial (GDM), model exactly this sequential allocation process \cite{ConnorMosimann1969,NgTianTang2011,ZhouLange2010}.
Each stage is a beta-binomial split, equivalently a two-category DM, with its own concentration.
The benefit is local flexibility: early coarse splits and later refinements can have different overdispersion levels.
The cost is that the model depends on the chosen order, so it should be used only when that order is part of the scientific design rather than an arbitrary column ordering.
This qualifier is testable, and it holds: with the parameter count held fixed, the generalized model improves held-out likelihood over a random permutation only when the feature order carries genuine stage-varying dispersion, and the gain collapses on an exchangeable control, so the improvement is attributable to the order rather than to the extra dispersion parameters (Section~\ref{sec:eval-E06033}).

For $x\in\N^K$ with $x_1+\cdots+x_K=n$, define the remainder counts
$R_j:=\sum_{\ell=j+1}^K x_\ell$, with $R_0:=n$ and $R_{K-1}=x_K$.
Under the generalized Dirichlet--multinomial with parameters $(\alpha_j,\beta_j)_{j=1}^{K-1}$,
\begin{equation}
\label{eq:gdm_pmf}
\text{Pr}(X=x)
=
\prod_{j=1}^{K-1}
\binom{R_{j-1}}{x_j}
\frac{B(x_j+\alpha_j,R_j+\beta_j)}{B(\alpha_j,\beta_j)}
\end{equation}
Equivalently,
\[
X_j\mid R_{j-1}\sim\BetaBin(R_{j-1},\alpha_j,\beta_j),
\qquad j=1,\dots,K-1
\]
with $X_K=R_{K-1}$.
Thus the GDM is a cascade of fixed-total two-category overdispersed splits.

\begin{proposition}[Fixed-total generalized DM is again a curved exponential family]
\label{prop:gdmexpfam}
For fixed total $n$, the generalized Dirichlet--multinomial can be written as
\begin{align*}
\log \text{Pr}(X=x)
&=
 h_n(x)
 +
 \sum_{j=1}^{K-1}\sum_{m=0}^{n-1}
 \Big[
 \ind(m<x_j)\log(\alpha_j+m)
 +
 \ind(m<R_j)\log(\beta_j+m) \\
&\hspace{3cm}
 -
 \ind(m<R_{j-1})\log(\alpha_j+\beta_j+m)
 \Big]
\end{align*}
where $h_n(x)$ depends on $x$ but not on the parameters.
Therefore, one convenient sufficient-statistic representation is
\[
T(x)=\big\{\ind(m<x_j),\ \ind(m<R_j),\ \ind(m<R_{j-1})\big\}_{\substack{j=1,\dots,K-1\\m=0,\dots,n-1}}
\]
\end{proposition}

This sufficient-statistic representation records both the count assigned at each stage and the amount of count mass left after that stage.
The indicators $\{\ind(m<x_j)\}$ are per-feature threshold histograms, while $\{\ind(m<R_j)\}$ and $\{\ind(m<R_{j-1})\}$ track the sequential depletion of the remaining count budget.
This is exactly the information an ordered cascade can use and the permutation-invariant DM cannot.

The stage-wise factorization gives a direct analogue of the DM comparison.
For sample $i$, write
$R_{ij}:=\sum_{\ell=j+1}^K X_{i\ell}$, set $R_{i0}=n_i$, and define
$q_{ij}:=X_{ij}/R_{i,j-1}$ for stages with $R_{i,j-1}>0$.
Let $\nu_j:=\alpha_j+\beta_j$ be the stage concentration.
We keep $\nu_j$ fixed and replace the null split probability by the empirical split probability:
\[
\alpha^{\mathrm{cmp}}_{ij}:=\nu_j q_{ij},
\qquad
\beta^{\mathrm{cmp}}_{ij}:=\nu_j(1-q_{ij})
\]
The resulting stage-wise additive term is
\begin{equation}
\label{eq:gdm_stagecell}
c_{ij}^{\mathrm{GDM}}
=
\sum_{m=0}^{X_{ij}-1}
\log\frac{\nu_j q_{ij}+m}{\alpha_j+m}
+
\sum_{m=0}^{R_{ij}-1}
\log\frac{\nu_j(1-q_{ij})+m}{\beta_j+m}
\end{equation}
and each sample's contrast decomposes as
\[
\Delta_i^{\mathrm{GDM}}
=
\sum_{j=1}^{K-1} c_{ij}^{\mathrm{GDM}}
\]

The analogy with the DM is now transparent: keep the dispersion of each local split fixed, and let only the local split probability adapt to the sample.
The tradeoff is also transparent.
The GDM gives stage-specific dispersion, but it is order-dependent and it no longer yields a nonzero-only leafwise residual matrix.
A zero in feature $j$ need not give zero contribution, because the second sum in \eqref{eq:gdm_stagecell} depends on the remaining count $R_{ij}$.
Thus the GDM is useful for ordered count data, but it is not the default sparse transform for unordered feature matrices.

The tree extension below places this ordered model in a broader context: a GDM is the special case of a Dirichlet-tree multinomial in which the tree is a comb that peels off one feature at a time and passes the remainder forward.

\subsection{The Dirichlet-tree multinomial and tree-structured features}
\label{sec:dtm}

Many structured count tables are not ordered along a line, but instead come with a hierarchy.
Microbiome data are a clear example: leaves are taxa and internal nodes are higher taxonomic groups.
The same logic applies to lineage trees, ontologies, and other curated feature hierarchies.
The natural extension of the DM in this setting is the Dirichlet-tree multinomial (DTM) \cite{Dennis1991,Minka2004,WangZhao2017}.

The DTM contains both structured models already discussed.
If the tree has a single root whose children are the leaves, the DTM is exactly the ordinary flat DM.
If the tree is a comb that repeatedly splits off one ordered feature and passes the remainder down the chain, the DTM recovers the generalized DM.
Thus the DTM is the multiscale family that includes the flat unordered model and the ordered sequential model as special cases.

A flat DM uses one concentration parameter for the whole leaf composition.
A DTM instead models each internal node separately, so broad splits between major branches can have different dispersion from fine splits within a branch.
When the tree is scientifically meaningful, this extra structure can improve fit or representation; when it is not, the flat DM remains the simpler and safer null.

Let $T$ be a rooted tree with leaf set $\mathcal L=\{1,\dots,K\}$ and internal-node set $\mathcal I$.
For an internal node $\nu\in\mathcal I$ with children $\mathrm{ch}(\nu)$, let $\mathcal L(\nu,c)$ denote the leaves under child $c\in\mathrm{ch}(\nu)$.
Given a leaf-count vector $x$, define the aggregated child-subtree counts
\[
x_{\nu c}:=\sum_{\ell\in\mathcal L(\nu,c)}x_\ell,
\qquad
N_\nu:=\sum_{c\in\mathrm{ch}(\nu)}x_{\nu c}
\]
The DTM assigns an independent Dirichlet distribution to the branch-probability vector at each internal node.
After integrating those branch probabilities out, the leaf-count likelihood factorizes node by node.

\begin{proposition}[Nodewise DM factorization of the DTM]
\label{prop:dtmfactor}
For fixed tree $T$ and leaf-count vector $x$ with total $n$,
\[
\text{Pr}_{\DTM(T,\{\bm\alpha_\nu\})}(X=x\mid n)
=
\prod_{\nu\in\mathcal I}
\text{Pr}_{\DirMult(N_\nu,\bm\alpha_\nu)}(X_{\nu(\cdot)}=x_{\nu(\cdot)}\mid N_\nu)
\]
where $x_{\nu(\cdot)}=(x_{\nu c})_{c\in\mathrm{ch}(\nu)}$.
In particular, the ordinary DM is recovered when the root is the only internal node.
\end{proposition}

The intuition is that conditioning on the branch probabilities makes the leaf multinomial likelihood factorize into a product of multinomial terms for the child-subtree counts at each internal node; integrating each node against its local Dirichlet prior yields a local DM factor, and the multinomial coefficients telescope down the tree.
The full derivation is in Section~\ref{sec:dtm_aux}.

This factorization gives an immediate structured version of the DM residual transform.
Define global branch compositions by
\[
\pi_{\nu c}:=\frac{A_{\nu c}}{A_\nu},
\qquad
A_{\nu c}:=\sum_{i=1}^n X_{i,\nu c},
\qquad
A_\nu:=\sum_{c\in\mathrm{ch}(\nu)}A_{\nu c}
\]
and parameterize each internal node by $\alpha_{\nu c}=\alpha_{\nu0}\pi_{\nu c}$.
Because the DTM log-likelihood is a sum of nodewise DM log-likelihoods, each $\alpha_{\nu0}$ can be fit by the same one-dimensional MLE used in Section~\ref{sec:alpha0}; one may also pool concentrations by depth, shrink them toward a global value, or tie them through a parametric form $\alpha_{\nu0}=w_\nu\rho$.

The resulting branch-level residual coordinates are
\begin{align}
c_{i,\nu c}^{\mathrm{DTM}}
&:=
\sum_{m=0}^{X_{i,\nu c}-1}
\log\paren{\frac{\alpha_{\nu0}X_{i,\nu c}/N_{i\nu}+m}{\alpha_{\nu0}\pi_{\nu c}+m}},
\nonumber\\
d_{i,\nu c}^{\mathrm{DTM}}
&:=
\sgn\paren{X_{i,\nu c}-N_{i\nu}\pi_{\nu c}}
\sqrt{\abs{c_{i,\nu c}^{\mathrm{DTM}}}} \label{eq:dtm_branchresid}
\end{align}
These coordinates are indexed by internal-node/child pairs rather than by leaves, so the representation is multiscale rather than flat.
A sample activates only branches on the root-to-leaf paths of its nonzero leaves, so the transform is naturally sparse over branches even though it is no longer sparse over the original leaf coordinates.

The DTM is most appropriate when the hierarchy is known before normalization.
Learning a hierarchy from the same data is a larger model-selection problem: it can be useful, but it changes the null itself and requires validation or sample splitting to avoid overfitting.
Learning the tree from data is therefore a separate undertaking that the present paper does not pursue; here the DTM is used only with a hierarchy fixed in advance.

\subsection{Interpretation through sufficient statistics}
\label{sec:suffstats}

The main count distributions in this paper all admit useful exponential-family representations once their natural support is fixed. 
For the multinomial, the sufficient statistic is the count vector itself. 
For the Dirichlet, the sufficient statistic is $\log x$. 
For the DM with fixed total $n$, a convenient sufficient statistic is the collection of indicators
\[
\braces{\ind(m < x_k)}_{\substack{m=0,\dots,n-1\\k=1,\dots,K}}
\]
equivalently the count histogram or survival function across features. 
This is the right viewpoint from which to understand why the DM still behaves like an exponential-family model even though its heavy tails make it look very different from the multinomial \cite{Elkan2006,Morris1982}.

The sufficient-statistic viewpoint also clarifies what the proposed residual transform is doing. 
The multinomial reacts only to the sample composition. 
The DM augments the observable statistics from raw counts to count-threshold indicators, so repeated counts are less surprising than under a pure multinomial null. 
That is exactly the shrinkage formalized in Proposition~\ref{prop:boundedbymult}. 
The generalized DM of Section~\ref{sec:gdm} enlarges these statistics again by adding stage-wise remainder counts; the cost is that the resulting contrast is sequential and order dependent rather than a permutation-invariant sparse entry-wise transform. 
The DTM of Section~\ref{sec:dtm} performs the analogous enlargement on a tree: the sufficient statistics become nodewise threshold indicators of aggregated subtree counts, and the residual representation becomes multiscale rather than flat.

\subsection{How complete is this count-model spectrum?}
\label{sec:spectrum}

These are not the only multivariate count models one could write down.
One can enlarge the DM to ordered generalized Dirichlet families, tree-structured Dirichlet-tree models, latent-Gaussian models, zero-inflated variants, and more.
The narrower claim is about \emph{closed-form sparse likelihood-based normalization}. 
For that goal, the list of appropriate exact models is much shorter.

If we insist on conditioning on the observed row sum, keeping the comparison rule transparent, evaluating everything exactly from the nonzero entries, and preserving permutation symmetry when the features are unstructured, then the multinomial and the DM are close to canonical. 
The multinomial is the simplest joint fixed-total model. 
The DM is the simplest overdispersed extension that keeps the same compositional meaning and still has an exact sparse transform.

The generalized DM from Section~\ref{sec:gdm} is the clearest near-miss. 
It really does enlarge the model class by introducing stage-specific dispersion, but its sufficient statistics involve ordered remainder counts and its residual contributions depend on feature order. 
The DTM from Section~\ref{sec:dtm} is the more natural structured extension when a genuine hierarchy is available. 
It preserves exact factorization, allows nodewise overdispersion, and yields a multiscale branch-level residualization, but it gives up flat leafwise symmetry and exact leafwise sparsity.

The main remaining omission is richer cross-feature dependence. 
The DM ties all extra-multinomial variation to one concentration parameter and one P\'olya-reinforcement mechanism; the DTM replaces that by several local concentration parameters tied to a known hierarchy.
Neither model can represent arbitrary covariance on log-ratios.
What the tied-concentration assumption costs when the data is generated outside the DM family is itself measurable: under six non-DM mechanisms the flat DM stays at or above the multinomial in held-out likelihood, degrades with the magnitude of latent dependence but is blind to its rank, and introduces no spurious downstream false positives, so a richer model is preferred exactly when its extra flexibility shows up as a held-out gap (Section~\ref{sec:eval-E01253}).
Latent-Gaussian models, especially the logistic-normal multinomial and its unconditional analogue the multivariate Poisson-lognormal, were introduced precisely for that purpose \cite{AitchisonShen1980,AitchisonHo1989}.
Modern microbiome and sequencing models continue to use this idea when the dependence structure itself is scientifically important \cite{XiaChenFungLi2013}.

A natural additional model to consider, in the spirit of the present paper, is therefore the isotropic logistic-normal multinomial (LNM) null:
\begin{align*}
\eta_i &\sim \mathcal{N}\!\left(0,\tau^2\paren{I_K-\frac{1}{K}\mathbf{1}\mathbf{1}^\top}\right) \\
p_{ij}(\eta_i) &:= \frac{\pi_j e^{\eta_{ij}}}{\sum_{\ell=1}^K \pi_\ell e^{\eta_{i\ell}}} \\
X_i \mid n_i,\eta_i &\sim \Mult(n_i,p_i(\eta_i))
\end{align*}
Here $\bm\pi$ again encodes the null composition, while $\tau^2$ plays the role of a scalar dispersion parameter; as $\tau^2\to 0$, the model reduces continuously to the multinomial. 
Replacing the isotropic covariance by a full or low-rank matrix $\Sigma$ yields a much richer family that can represent dependence patterns inaccessible to the DM.

The reason not to make the LNM the main normalization model is practical rather than philosophical. 
Its marginal likelihood is an integral over a latent Gaussian field and has no closed form, so fitting typically requires Laplace approximation, variational inference, or MCMC.%

The resulting deviance is no longer an exact sum of sparse entry-wise terms, and an exact transform analogous to the log-gamma form $c_{ij}^{\mathrm{DM}}=\log[\Gamma(x+\alpha_0 x/n_i)/\Gamma(\alpha_0 x/n_i)]-\log[\Gamma(x+\alpha_0\pi_j)/\Gamma(\alpha_0\pi_j)]$ (precise form in \eqref{eq:lgammaresid}, App.~\ref{sec:sparseimplementation}) is lost.
Thus the spectrum is essentially complete for the narrower goal of \emph{exact likelihood-based normalization with transparent structure}: multinomial for no overdispersion, DM for symmetric one-parameter overdispersion on an unstructured feature set, generalized DM when a sequential order is scientifically meaningful, DTM when a fixed tree is scientifically meaningful, and latent-Gaussian models once one is willing to abandon closed form.

\section{Experiments}
\label{sec:experiments}

Across all benchmarks below the recommended fitting protocol is the simplest one: estimate the scalar concentration $\alpha_0$ from the training-split DM likelihood with $\bm\pi$ held at the empirical global composition, and apply the closed-form residual transform of Section~\ref{sec:dmresids}.
The iterative fixed-point fit of Section~\ref{sec:alpha0} converges in tens of steps across the full overdispersion range and tracks $\alpha_0$ smoothly from the near-multinomial regime ($\alpha_0\to\infty$) to the negative-binomial regime ($\alpha_0\to 0$); the single concentration parameter therefore acts as one adaptive knob that interpolates between the two endpoint nulls.
Jointly refitting $\bm\pi$ and $\alpha_0$ rarely changes downstream geometry on the datasets we examine and is exercised only as a sensitivity check (Appendix~\ref{sec:impldetails}).
The dispersion sweeps in the same appendix probe the full $\alpha_0$ range and verify that the transform behaves as expected at both endpoints.

The empirical evaluation addresses five questions.
\begin{itemize}
\item When do DM residuals improve on multinomial residuals?
\item How close are they in practice to feature-wise negative-binomial residual methods?
\item When do richer models---such as a jointly fitted composition, a known tree model, or a latent-Gaussian model---pay for their extra flexibility?
\item Does the sparse DM transform preserve downstream biological structure while reducing depth and nuisance coupling?
\item Is the method computationally competitive with common pointwise sparse transforms such as shifted-log or \texttt{log1p} normalization?
\end{itemize}

The core practical question is whether DM normalization removes technical sampling variation without washing out real biological structure.
For that reason, the main text emphasizes real-data and negative-control benchmarks, together with structured-data and runtime ablations.
Synthetic calibration and model-recovery studies are still useful for checking algebra and implementation, but they are secondary to that real-data question and are therefore deferred to Appendix~\ref{sec:synthetic_appendix}.

To keep the comparison systematic rather than anecdotal, every benchmark follows the same sequence: preprocess, fit, transform, evaluate held-out fit when a generative model is available, measure downstream geometry, and report resource use.
The runtime tables include both model-fitting time and pure transform time, with pointwise \texttt{log1p} included as the reference for the fastest sparse nonzero-only normalization.

\subsection{Questions, baselines, and outputs}

All methods are compared on the same preprocessed matrix within each dataset. 
All-zero rows and all-zero columns are removed before fitting. 
Any optional feature filtering or high-variance-feature restriction is defined once per dataset on the raw training counts and then frozen across normalization methods. 
Dense baselines are run on exactly the same filtered matrix as the sparse methods; if a dense method exceeds the stated memory or wall-clock budget, we report that failure rather than shrinking the problem for that method alone.

The comparison set has three layers. 
The first layer contains practical baselines: pointwise shifted logarithms or \texttt{log1p}, CLR, TF--IDF when appropriate \cite{RamosEtAl2003}, and the dense pseudocount Dirichlet residual transform of Section~\ref{app:dirichlet-dense} \cite{Aitchison1982,Gloor2017,AhlmannEltzeHuber2023}. 
The Dirichlet baseline is included for a very specific reason: it is the closest continuous compositional analogue of the DM comparison, but it necessarily densifies the matrix and therefore cleanly separates ``Dirichlet-type residualization'' from ``exact sparse count residualization.'' 
Unless otherwise stated, the CLR and Dirichlet baselines use the same unit pseudocount so that any difference between them is attributable to the model family rather than to different zero-handling conventions.

The second layer contains raw-count models with exact likelihoods on the observed discrete sample space: multinomial residuals, analytic Pearson or feature-wise NB residuals, practical regularized NB regression, and the independent negative-binomial model of Section~\ref{sec:matchednb} \cite{lause2021analytic,HafemeisterSatija2019,townes2019feature,ChoudharySatija2022}. 
The third layer contains structured or richer alternatives when they are scientifically appropriate and computationally feasible: the generalized DM when a genuine sequential order exists, the DTM when a meaningful tree exists, and logistic-normal or Poisson-lognormal baselines when one wants a dense latent-Gaussian comparator \cite{AitchisonShen1980,AitchisonHo1989,XiaChenFungLi2013}.

For completeness, Appendix~\ref{app:otherresiduals} places the classical binomial, Poisson, negative-binomial, multinomial, Dirichlet, and DM residual formulas side by side even when not all of them are run on every real-data benchmark. 
This is useful because the formulas differ much less than their names suggest: they are all signed square-root likelihood-contrast transforms relative to closely
related nulls.

Within each benchmark we report several families of outputs:
\begin{enumerate}
\item exact held-out conditional log-likelihood on the raw count space when the model provides one;
\item unconditional held-out log-likelihood when independent NB models are included;
\item transformation density, peak memory, fit time, transform time, and nonzero throughput relative to pointwise \texttt{log1p};
\item depth-coupling diagnostics, including correlations between latent axes and $n_i$ or size factors;
\item downstream neighborhood, clustering, retrieval, and annotation metrics under a common pipeline; and
\item paired residual diagnostics on the nonzero entries, especially DM-versus-multinomial,
DM-versus-feature-wise-NB, and DM-versus-Dirichlet magnitude comparisons stratified by
$X_{ij}=1$, $2$, $3$--$5$, and $>5$.
\end{enumerate}

The last item is particularly important. 
Singleton counts should look nearly identical under multinomial, DM, and the independent negative-binomial comparison after conditioning on row sums.
Repeated counts are exactly where conditioning on the row sum and allowing joint overdispersion can change the geometry.

When a method does not define an exact likelihood on the same raw count space---as with shifted
logs, CLR, TF--IDF, or pseudocount Dirichlet residuals---we compare it on downstream and
computational criteria rather than mixing its objective values into the exact count-likelihood
tables.
Likewise, when a richer baseline is fit by approximation (for example a logistic-normal
model fit by Laplace or variational methods), we report the approximation that was actually
optimized and do not place it on the same axis as exact closed-form count likelihoods without an
explicit caveat.

Whenever a single split would be noisy, we repeat the entire benchmark over several matched
random splits and report paired averages together with uncertainty intervals across splits.
All
methods within a split share the same training, validation, and test sets, the same filtered
feature set, and the same downstream hyperparameters unless a dedicated sensitivity analysis is
being reported.

Concretely, we run three matched random $80/20$ train/test splits on the
two real scRNA-seq datasets (PBMC~3k at 2700 cells; GTEx subsampled to%

5000 cells for the repeated-splits section) and report mean $\pm$
standard deviation plus one-sided paired Wilcoxon p-values ("DM $>$
baseline") for each metric.
Table~\ref{tab:wilcoxon_summary} summarizes the result.
Held-out LL deltas
are paired and consistent across splits; with only three splits the
Wilcoxon test saturates at $p = 2^{-3} = 0.125$ (the smallest p-value
achievable with three all-positive signs), so significance is bounded by
the test's minimum, not by the underlying effect size.
Running more
splits would reduce this bound mechanically; we keep $n_{\text{seeds}}=3$
for compute tractability and report the delta magnitudes directly.

\begin{table}[t]
\centering
\small
\setlength{\tabcolsep}{4pt}
\renewcommand{\arraystretch}{1.15}
\begin{tabular}{@{}llrr@{}}
\toprule
Dataset & DM advantage over & $\Delta$ (mean $\pm$ std) & Wilcoxon p \\
\midrule
PBMC 3k (real)   & Multinomial           & $+0.087 \pm 0.003$ & $0.125$ \\
PBMC 3k (real)   & NB (uncond.)          & $+0.169 \pm 0.013$ & $0.125$ \\
PBMC 3k (real)   & NB (cond.) & $+0.000 \pm 0.000$ & $1.000$ \\
GTEx 5k\textsuperscript{*}  & Multinomial           & $+0.201 \pm 0.007$ & $0.125$ \\
GTEx 5k\textsuperscript{*}  & NB (uncond.)          & $+0.418 \pm 0.019$ & $0.125$ \\
GTEx 5k\textsuperscript{*}  & NB (cond.) & $+0.000 \pm 0.000$ & $1.000$ \\
\bottomrule
\end{tabular}
\caption{Paired Wilcoxon (one-sided, alternative ``DM $>$ other'') for
held-out per-count log-likelihood across three matched random train/test
splits.
$\Delta$ is the mean signed DM$-$other difference, standard
deviation across splits reported.
$^{*}$The GTEx cross-tissue atlas is subsampled to 5{,}000
cells for this repeated-splits section; the single-split benchmark
tables elsewhere use a 20{,}000-cell subsample of the same atlas.
Wilcoxon p-values saturate at $2^{-n_{\text{seeds}}}$
when all deltas have the same sign; with $n_{\text{seeds}}=3$ this gives a
floor of $0.125$.
The DM vs.\ NB-conditional rows collapse to
$\Delta = 0$ exactly, reconfirming Corollary~\ref{cor:nbdmdeviance} at the
level of whole held-out datasets (not just the numerical check of
Table~\ref{tab:regime2}).}
\label{tab:wilcoxon_summary}
\end{table}

\subsection{Real-data benchmark suite}

The main real-data question is whether DM normalization separates biological structure from technical sampling effects while preserving the sparse geometry of the original matrix. 
For that reason, the benchmark should include not only heterogeneous biological datasets, but also negative-control or low-heterogeneity datasets of the kind emphasized in earlier normalization studies \cite{Svensson2020,lause2021analytic,AhlmannEltzeHuber2023}. 
On those controls the desired outcome is not dramatic separation; it is that one fitted null accounts for most of the observed variation, with little residual depth dependence and little spurious low-dimensional structure.

The benchmark should also span several biochemical count regimes rather than only one assay family. 
A minimal convincing suite is summarized in Table~\ref{tab:benchmarkregimes}: at least one singleton-dominated matrix, one moderate repeated-count matrix, one heavy-dominance or enrichment matrix, and one genuinely tree-structured matrix.

\begin{table}[t]
\centering
\small
\setlength{\tabcolsep}{5pt}
\renewcommand{\arraystretch}{1.25}
\begin{tabular}{@{}L{0.20\textwidth}L{0.28\textwidth}L{0.45\textwidth}@{}}
\toprule
Regime & Representative assays & Why it matters \\
\midrule
Singleton-dominated & scATAC-seq peak counts & Tests the near-multinomial limit where repeated counts are rare and the DM should behave almost identically to the multinomial. \\

Moderate repeated counts & scRNA-seq UMI counts; spatial transcriptomics & Tests whether overdispersion-aware joint residuals improve on multinomial residuals without destroying neighborhood geometry. \\

Heavy dominance / enrichment & CITE-seq ADT panels; microbiome taxon tables; CRISPR or MPRA barcode counts & Tests whether the DM appropriately damps repeated high-count events and dominant features. \\

Tree-structured & Microbiome phylogenies or taxonomies; lineage- or ontology-organized count panels & Tests whether nodewise dispersion in the DTM buys real predictive or geometric gains over a flat DM. \\
\bottomrule
\end{tabular}
\caption{Benchmark regimes that together probe the main statistical claims of the paper.}
\label{tab:benchmarkregimes}
\end{table}

For each real dataset, preprocessing is frozen before normalization. 
We remove all-zero rows and features, optionally restrict to a fixed feature panel or high-variance-feature set defined from the raw training counts, and freeze that feature set across methods. 
Where known technical covariates are available---batch, donor, chemistry, capture lane, or treatment arm---we record them once and use them only in dedicated conditional-null ablations rather than mixing them into the default comparison.

Sample-level splits are used for any held-out likelihood or supervised downstream task. 
When the dataset is large enough, we use a stratified $60/20/20$ training/validation/test split; when it is smaller, we use repeated five-fold cross-validation with matched folds across methods. 
All model fitting, parameter tuning, and any choice between fixed-$\pi$ and jointly fitted variants are made without access to the final test split.

The primary real-data fit metric is held-out conditional log-likelihood per sample and per observed count for methods defined on the raw count space. 
Secondary metrics are transformation density, wall-clock cost, peak memory, residual-depth coupling, and downstream neighborhood quality. 
Dense baselines remain part of the downstream and computational comparison even when they do not live on the same exact count-likelihood axis.

This protocol is identical across transforms. 
All filtering, feature selection, and train/validation/test splitting are defined \emph{before} normalization; any null compositions $\pi$ or $\pi_{\nu c}$ are fit on the training split only to avoid leakage; and the same transformed dimensionality, neighborhood size, clustering settings, and downstream classifiers are then applied to every method.

Figure~\ref{fig:real-data-benchmarks} summarizes the comparison across
all five real regimes: the DM improves on the multinomial everywhere, its
depth coupling stays comparable to the multinomial, and its residual
shrinkage concentrates on the repeated counts that carry the
overdispersion.
The held-out-likelihood comparison on PBMC~3k is reported with paired
uncertainty in Table~\ref{tab:pbmc-heldout}: across eight matched train/test
splits the DM improves on the multinomial by $0.070$ nats per observed count
(bootstrap interval $[0.069,0.072]$, all eight splits, paired
$p=0.0039$), it coincides with the conditional independent-negative-binomial
model to floating-point precision as Corollary~\ref{cor:nbdmdeviance} predicts,
and the only transform with a higher held-out likelihood is the dense Dirichlet
pseudocount, which forfeits the exact sparsity the DM preserves.

\begin{table}[t]
\centering
\small
\setlength{\tabcolsep}{6pt}
\renewcommand{\arraystretch}{1.15}
\begin{tabular}{@{}lrrrrr@{}}
\toprule
& PBMC 3k & GTEx 20k & Neg.\ ctrl & scATAC & CITE-seq \\
& (real) & (real) & (real) & (real) & (real) \\
\midrule
DM                     & $\mathbf{-1.27}$ & $\mathbf{-2.72}$ & $\mathbf{-0.20}$ & $\mathbf{-0.79}$ & $\mathbf{-0.02}$ \\
Multinomial            & $-1.35$ & $-2.92$ & $-0.25$ & $-0.88$ & $-0.81$ \\
NB (unconditional)     & $-1.45$ & $-3.14$ & $-0.38$ & $-1.04$ & $-0.06$ \\
NB (conditional) & $\mathbf{-1.27}$ & $\mathbf{-2.72}$ & $\mathbf{-0.20}$ & $\mathbf{-0.79}$ & $\mathbf{-0.02}$ \\
\midrule
$\hat\alpha_0$ (DM)    & $4.7\!\times\!10^{3}$ & $2.9\!\times\!10^{3}$ & $6.6\!\times\!10^{1}$ & $7.0\!\times\!10^{3}$ & $8.0$ \\
\bottomrule
\end{tabular}
\caption{Held-out conditional log-likelihood per observed count on the five
real benchmarks, plus the fitted DM scalar $\hat\alpha_0$.
Bold = best
(the NB-conditional row equals DM exactly by
Corollary~\ref{cor:nbdmdeviance}, as verified to machine precision ---
a maximum residual of $4.3\times10^{-11}$ along an independent
negative-binomial log-probability route --- in
Figure~\ref{fig:identity-verification}).
DM matches
multinomial in the multinomial-limit direction (large
$\hat\alpha_0$) and improves on it everywhere $\hat\alpha_0$ is moderate.
The PBMC $\hat\alpha_0 \approx 4.7\times 10^{3}$ and the housekeeping-gene
negative-control $\hat\alpha_0 \approx 66$ both reflect real residual
overdispersion after conditioning on library size; the CITE-seq ADT
$\hat\alpha_0 \approx 8$ reflects the heavy-dominance regime.}
\label{tab:heldout_ll}
\end{table}

\begin{table}[t]
\centering
\small
\setlength{\tabcolsep}{5pt}
\renewcommand{\arraystretch}{1.15}
\begin{tabular}{@{}l|rrrr|rrrr@{}}
\toprule
& \multicolumn{4}{c|}{Label transfer accuracy} & \multicolumn{4}{c}{Adjusted Rand Index} \\
\cmidrule(lr){2-5}\cmidrule(lr){6-9}
& PBMC & GTEx & scATAC & CITE-seq & PBMC & GTEx & scATAC & CITE-seq \\
\midrule
DM          & $0.90$ & $\mathbf{0.84}$ & $0.11$ & $0.22$ & $0.47$ & $0.33$ & $0.00$ & $-$ \\
Multinomial & $0.90$ & $0.83$ & $0.11$ & $0.18$ & $0.47$ & $0.34$ & $0.00$ & $-$ \\
NB          & $0.89$ & $0.83$ & $0.13$ & $0.22$ & $0.38$ & $0.32$ & $0.00$ & $-$ \\
Pearson     & $0.84$ & $0.60$ & $0.12$ & $0.22$ & $0.37$ & $0.13$ & $-$    & $-$ \\
Dirichlet   & $0.87$ & $0.81$ & $0.15$ & $0.23$ & $0.34$ & $0.30$ & $-$    & $-$ \\
CLR         & $0.90$ & $0.83$ & $0.12$ & $0.22$ & $0.46$ & $0.32$ & $-$    & $-$ \\
Shifted log & $0.89$ & $0.83$ & $0.13$ & $0.19$ & $0.39$ & $0.32$ & $-$    & $-$ \\
TF-IDF      & $\mathbf{0.91}$ & $0.79$ & $0.11$ & $0.22$ & $\mathbf{0.54}$ & $\mathbf{0.34}$ & $-$ & $-$ \\
\bottomrule
\end{tabular}
\caption{Downstream metrics: $k$-NN label transfer accuracy and adjusted
Rand index (spectral clustering with $k$-NN affinity, $k=30$, PCA to 10
components).
ARI entries marked $-$ are near zero, reflecting the
difficulty of unsupervised clustering on the sparse scATAC and CITE-seq
count matrices.}
\label{tab:downstream_summary}
\end{table}

\subsection{Structured-data benchmarks}

Whenever a genuine hierarchy is available, we run a structured comparison in three layers: a flat
DM on the leaves, a minimal ordered ablation when a scientifically meaningful sequential order
exists, and a full DTM on the known tree.
The flat DM serves as the unstructured joint baseline;
the ordered model tests whether any benefit is merely due to breaking symmetry; the DTM tests
whether a real hierarchy organizes the overdispersion well enough to improve prediction or
representation.

In these structured benchmarks the tree itself is fixed in advance.
Aggregated subtree counts are
computed once, nodewise null compositions are estimated on the training split only, and nodewise
concentration parameters are fit either independently, pooled by depth, or tied to a single global
scalar as an ablation.
This makes it possible to separate the benefit of using the hierarchy at all
from the benefit of allowing node-specific dispersion.

The most informative structured-data metrics are held-out leafwise composition likelihood, held-out nodewise split likelihood, branch-level residual sparsity, and whether the highest-variance residual directions localize coherent subtrees rather than scattered leaves. 
In microbiome-style data, a convincing win for the DTM should therefore look like both better held-out split fit and more taxonomically coherent local neighborhoods, in line with prior tree-aware regression results \cite{WangZhao2017}.
If neither happens, the flat DM should be preferred on parsimony grounds.
On the controlled synthetic tree, Figure~\ref{fig:structured-benchmarks}
bears this out: the flat DM decisively beats all three DTM strategies on
held-out likelihood, node-specific dispersion does vary with depth, and
the DM still produces the most tree-aware feature geometry.

\begin{figure}[t]
\centering
\includegraphics[width=0.92\textwidth]{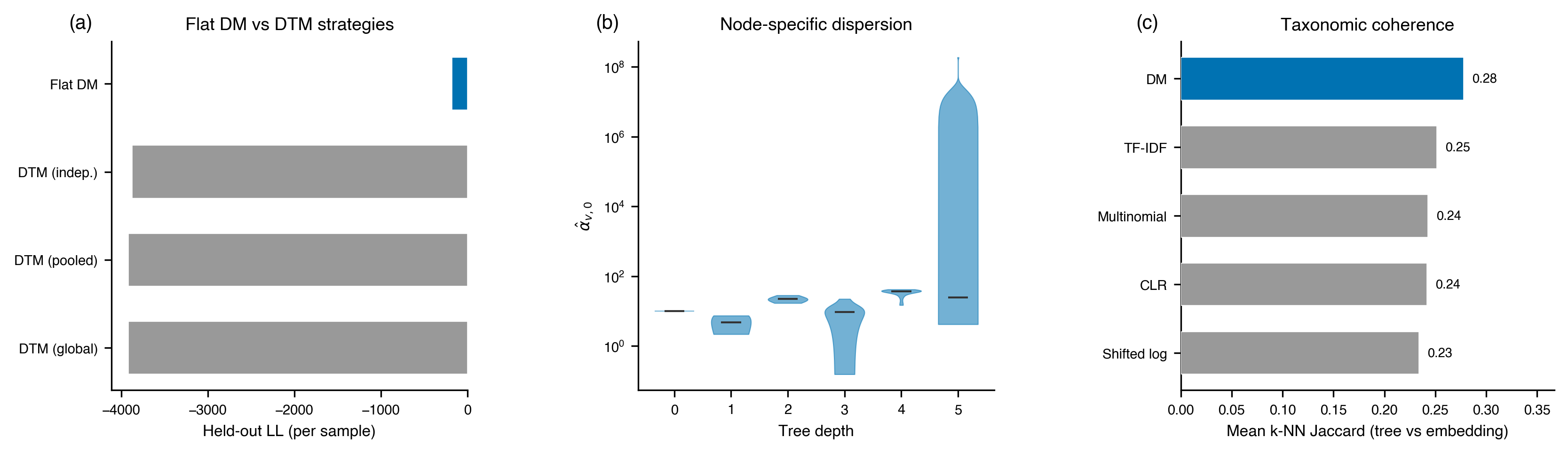}
\caption{\textbf{Structured-data benchmarks.}
\textbf{(a)}~Held-out per-sample LL: flat DM (unstructured joint
baseline, highlighted) vs.\ three DTM strategies (independent nodewise
$\alpha_{v,0}$, pooled-by-depth, single global scalar).
On the synthetic tree data used here the flat DM is dramatically better:
held-out likelihood for every DTM strategy is worse by more than an order
of magnitude, the per-node dispersion parameters overfitting badly while
the single-parameter flat model generalizes.
\textbf{(b)}~Node-specific fitted dispersion $\hat\alpha_{v,0}$ stratified
by tree depth (independent strategy, violin = distribution across nodes at
each depth).
The spread of dispersion across depth motivates nodewise fitting.
\textbf{(c)}~Taxonomic coherence: for each normalization, features are
projected into a low-dimensional feature-level PCA embedding, and the
mean Jaccard overlap is measured between each feature's $k=10$ neighbors
in the embedding and its $k$ closest features by tree distance.
Higher is more taxonomically coherent: the DM (highlighted) attains the
most tree-aware feature geometry, ahead of TF-IDF, multinomial, CLR, and
shifted log.}
\label{fig:structured-benchmarks}
\end{figure}


The synthetic comparison above is deliberately controlled: with the tree
fixed and the overdispersion ratio swept, there are regimes in which a flat
DM holds its own and the parsimony argument favors it.
The real-data analogue tells the complementary story.
On a genuine microbial cohort --- the American Gut Project, carrying its
own Greengenes phylogeny rather than a synthetic tree --- the picture
under heavy biological overdispersion is far less equivocal:
honoring the hierarchy is decisive, the flat null collapses outright, and
the particular taxonomy carries information beyond what any tree affords,
beating a leaf-shuffled control with separated intervals.
The nuance is that the fixed taxonomy is not the last word --- a
hierarchy learned from co-occurrence does better, and a fully factorized
per-feature model better still --- so the tree helps without being the
whole story.
The full benchmark, with all seven estimators and held-out intervals, is
in Section~\ref{sec:eval-E01138}.


\subsection{Real-taxonomy DTM benchmark on a microbiome cohort (E01138)}
\label{sec:eval-E01138}

\paragraph{Claim.}  On a genuine microbial cohort with a real phylogeny, does
honoring the taxonomy through a Dirichlet-tree-multinomial (DTM) prior pay off
against a flat Dirichlet-multinomial, and does the \emph{particular} hierarchy
matter?  This is the real-data analogue of the structured-tree benchmark.

\paragraph{Setup.}  American Gut Project 16S counts (a $5{,}000$-sample,
$1{,}500$-OTU subset of the filtered $17{,}305\times 15{,}736$ table; the
$1{,}500$ most-prevalent OTUs) with the Greengenes $97\%$ closed-reference
phylogeny pruned to the kept leaves ($2{,}999$ tree nodes).  Five matched
$80/20$ splits; held-out per-sample log-likelihood; BCa bootstrap intervals
($n_{\mathrm{boot}}=1000$).  Seven estimators: flat Dirichlet-multinomial; the
fixed-taxonomy DTM and its depth-pooled and single-global-$\alpha_0$ variants; a
random-tree DTM (leaf-shuffled negative control); a data-driven
hierarchical-clustering DTM; and a feature-wise negative binomial.

\paragraph{Result.}  Held-out per-sample log-likelihood (mean, higher is better;
$95\%$ BCa interval):

\begin{center}
\begin{tabular}{lrr}
\toprule
Estimator & Held-out LL & 95\% BCa CI \\
\midrule
Feature-wise NB             & $-1978.5$  & $[-1997.9,\,-1960.0]$ \\
Data-driven (hclust) DTM    & $-2257.4$  & $[-2282.7,\,-2235.1]$ \\
Single-$\alpha_0$ DTM       & $-2357.7$  & $[-2385.2,\,-2334.9]$ \\
Fixed-taxonomy DTM          & $-2357.7$  & $[-2385.0,\,-2334.0]$ \\
Random-tree DTM             & $-2431.8$  & $[-2459.1,\,-2407.2]$ \\
Depth-pooled DTM            & $-2530.3$  & $[-2560.3,\,-2505.3]$ \\
Flat Dirichlet-multinomial  & $-81115.6$ & $[-82981.4,\,-77942.7]$ \\
\bottomrule
\end{tabular}
\end{center}

\paragraph{Reading.}  Three things land.  First, structure is decisive on real
overdispersed counts: every tree-aware estimator beats the flat
Dirichlet-multinomial by more than an order of magnitude --- the flat model's
moment-matched concentration collapses to its floor under the cohort's
$\sim\!10^{3}\times$ overdispersion, so it cannot absorb the depth structure that
the tree priors decompose node by node.  Second, the \emph{particular} hierarchy
carries information: the fixed taxonomy strictly beats the leaf-shuffled
random-tree control with non-overlapping intervals, so the win is the taxonomy's
geometry, not merely the extra parameters a tree affords.  Third, the honest
nuance --- a data-driven hierarchy learned from co-occurrence beats the fixed
taxonomy, and a fully factorized per-feature negative binomial beats every
tree-prior estimator.  Tree structure helps, and when a tree is used a
data-driven one helps most, but a flexible per-feature model that lets each OTU
set its own dispersion does better still on this cohort.

\paragraph{Caveats.}  The cohort is subset to $5{,}000$ samples (uniform-random)
and the $1{,}500$ most-prevalent OTUs.  The flat-model collapse is the
pre-registered concentration-flooring failure mode under heavy overdispersion,
not a numerical artefact (all log-likelihoods are proper).  The data-driven tree
is built by connectivity-constrained agglomeration over a nearest-neighbor
graph (near-linear in the leaf count).

\subsection{Fitting and computational ablations}
\label{sec:mleablation}

The fitting stage is benchmarked separately from the residual transform itself.
The default fixed-point fit is checked against the three further variants of Section~\ref{sec:alpha0}, comparing
\begin{enumerate}[1.]
\item final training log-likelihood;
\item held-out log-likelihood when validation splits are available;
\item wall-clock time to a common likelihood or gradient tolerance;
\item number of iterations;
\item preprocessing time and memory for any histogram compression; and
\item stability under different initializations.
\end{enumerate}

A fair protocol uses the same starting point for all methods: $\pi_j^{(0)}=a_j/N$ and a common $\alpha_0^{(0)}$ obtained from a coarse line search or log-grid scan of the scalar objective. 
Joint methods may then be initialized at $\alpha_j^{(0)}=\alpha_0^{(0)}\pi_j^{(0)}$. 
Stopping should be based on likelihood change or gradient norm rather than on a fixed iteration count.

The expected comparison has two layers. 
First, fixed-point and Newton should agree within each parameterization when run to convergence; any discrepancy is numerical rather than statistical. 
Second, joint fits should dominate fits with $\bm\pi$ held fixed in raw likelihood, and the size of that gap measures how costly the simpler global-composition null is on a given dataset. 
For the DTM, the same logic applies node by node: one can compare fully independent nodewise concentrations, depth-pooled concentrations, and a single shared scalar to quantify how much structured dispersion is present.

Computational scaling is reported as a function of both $\mathrm{nnz}(X)$ and $\max_i n_i$.
For flat DM fits, the relevant comparison is direct sparse digamma/trigamma evaluation versus the compressed histogram representation.
For the DTM, there is an additional structural advantage: once the tree aggregations are built, the nodewise scalar fits are embarrassingly parallel.

The transform-time benchmark separates fitting from application.
After $\hat\alpha_0$ and $\bm\pi$ are fixed, the DM transform uses the log-gamma formula \eqref{eq:lgammaresid} and therefore costs one constant-time special-function calculation per nonzero entry.
The closest computational baseline is pointwise \texttt{log1p}, which also makes one pass over the nonzero entries and preserves the sparsity pattern.
Measuring wall-clock time, peak memory, and processed nonzeros per second for \texttt{log1p} and DM residuals across three decades of nonzero count, the transform's fitted scaling exponent is one, its throughput is flat at roughly six nanoseconds per nonzero independent of matrix size, and its output carries exactly the input's nonzeros (Section~\ref{sec:eval-E01326}).
The cost relative to \texttt{log1p} is a constant factor of about fourteen, a log-gamma evaluation in place of a logarithm, paid once per nonzero and never growing with the matrix, so the transform inherits \texttt{log1p}'s linear scaling and its sparsity while a dense compositional baseline exhausts memory on the large matrices the transform handles inside the sparse nonzero set.
Fitting the concentration is the more expensive stage, a few times the cost of applying the transform, so the intended use of fitting once and transforming many times is the economical one.

\subsection{Downstream analyses and reporting standards}
Normalization is only useful insofar as it supports the tasks for which the counts were collected.
We therefore evaluate each transform in a common downstream pipeline.
After normalization we
compute a fixed-dimensional linear embedding (typically PCA or truncated SVD; by default 50
components, or the largest smaller allowable value), build a $k$-nearest-neighbor graph in that
embedding (by default $k=30$), and report neighborhood-based metrics such as label purity,
adjusted Rand index, normalized mutual information, and label-transfer accuracy whenever
ground-truth annotations are available.
For replicate or perturbation datasets, replicate retrieval
and condition prediction play the same role.

Downstream hyperparameters are not tuned separately for each normalization on the test split.
Instead, either one common set of hyperparameters is fixed per dataset, or a small grid is explored
once and then summarized identically across methods.
This prevents downstream retuning from
washing out genuine differences between the normalizations themselves.

We also report diagnostics targeted to the paper's specific claims.
These include correlations
between leading latent axes and sequencing depth, per-sample row norms after transformation,
paired scatterplots of DM versus multinomial, NB, and Dirichlet residual magnitudes on the
nonzero entries, and count-stratified summaries of the residual gap on singletons versus repeated
counts.
On structured datasets we add flat-DM versus DTM likelihood improvements and
branch-level sparsity summaries.

Low-dimensional visualizations, when included later, are treated as illustrations rather than as
primary evidence.
The quantitative comparison rests on likelihood, calibration, neighborhood
structure, and computation.
This keeps the narrative from depending on a visually persuasive but
ultimately qualitative embedding plot.

\begin{figure}[t]
\centering
\includegraphics[width=\textwidth]{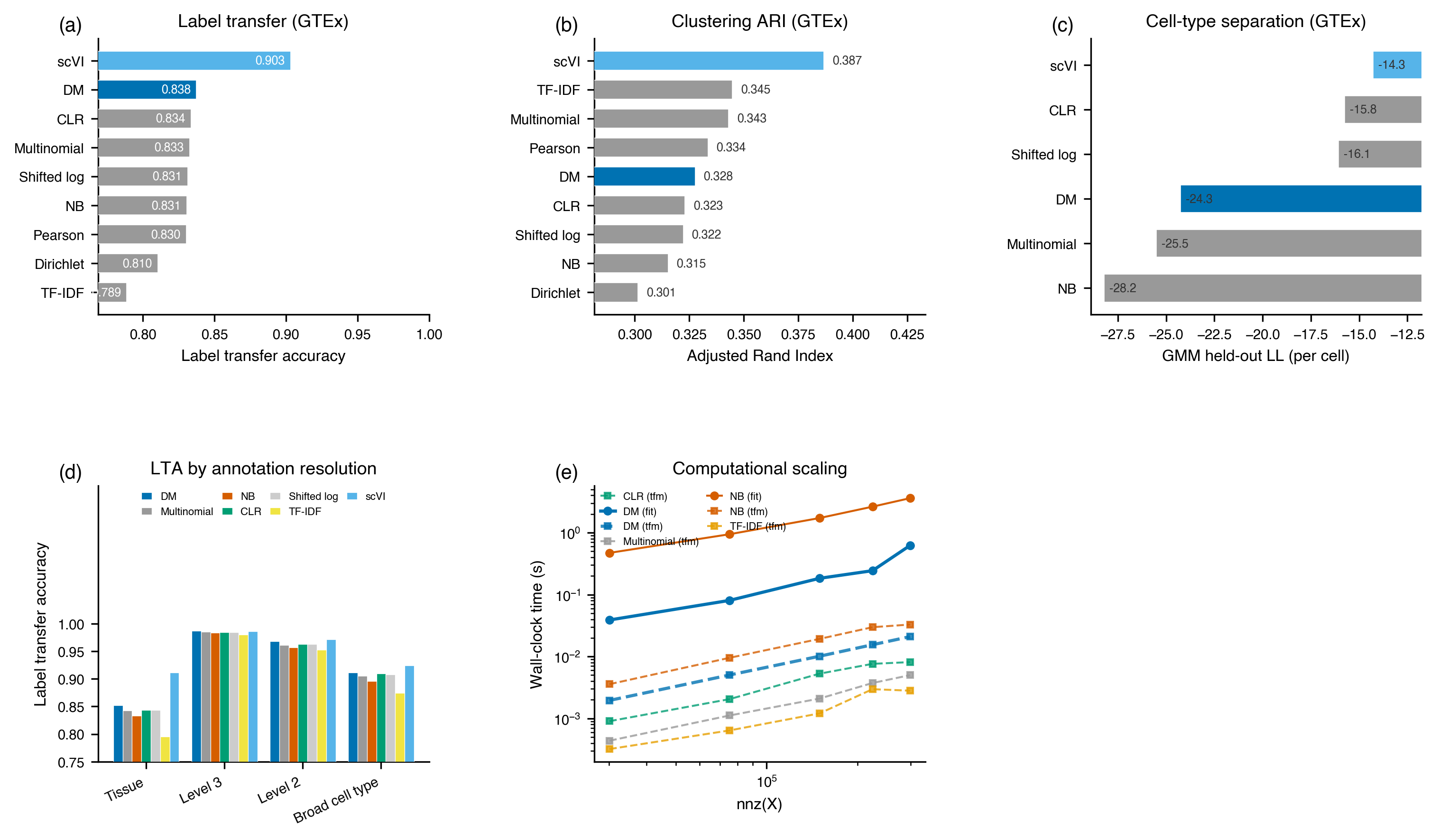}
\caption{\textbf{Downstream analyses on GTEx.}
A deep generative baseline, scVI, is included wherever the comparison is
commensurable.
It attains the strongest embeddings here, but at far higher training cost
and at the cost of differential-expression calibration: scVI inflates the
realized false-discovery proportion several-fold
(Section~\ref{sec:eval-E01246}), whereas the DM stays calibrated.
Among the linear-time normalizations the DM is competitive throughout.
\textbf{(a)}~Label transfer accuracy on GTEx broad cell types (42 classes),
$k$-NN majority vote with $k=30$ on the 10-dim PCA embedding, methods
sorted.
scVI leads; the DM is the strongest of the count-based normalizations while
preserving sparsity.
\textbf{(b)}~Adjusted Rand Index from spectral clustering on the same
$k$-NN graph, methods sorted.
scVI leads; the DM sits mid-pack among the normalizations.
\textbf{(c)}~Cell-type separation: label-initialized Gaussian-mixture
held-out log-likelihood in PCA space (42 broad cell types), methods sorted;
TF-IDF is omitted because its scale is incomparable.
Each method's 20-dim PCA embedding is scored against a GMM whose
components' means and covariances are estimated from the labeled cells;
higher is better separation.
scVI separates the cell types best; the DM improves on the multinomial but
trails the log-ratio normalizations (CLR, shifted log) on this metric.
\textbf{(d)}~Label transfer accuracy across four annotation resolutions
(Tissue, Level~3, Level~2, Broad cell type), grouped by method.
At this finer-grained transfer the methods are close: scVI and the DM are
both competitive at every resolution.
\textbf{(e)}~Computational scaling: wall-clock time vs.\ $\mathrm{nnz}(X)$
for the fit (solid) and transform (dashed) stages across the linear-time
methods.
The DM fit and transform both scale linearly in $\mathrm{nnz}$; the NB fit
is the most expensive because it runs a feature-wise Newton loop at
$O(nK)$ cost.}
\label{fig:downstream-analyses}
\end{figure}


The empirical evaluation is designed to support four claims, and the
downstream summary in Figure~\ref{fig:downstream-analyses} speaks to each:
label transfer and clustering on GTEx, label-transfer accuracy across
annotation resolutions, cell-type separation under a label-initialized
Gaussian mixture, and linear-in-$\mathrm{nnz}(X)$ computational scaling.
First, the DM should improve on the multinomial whenever repeated counts matter, yet reduce to it in singleton-dominated regimes.
Second, it should behave as a drop-in replacement for practical feature-wise NB residuals while keeping the output sparse and the assumptions transparent. 
Third, richer models---jointly fitted $\bm\pi$, DTM, or logistic-normal baselines---should only be preferred when their extra flexibility pays off measurably in held-out fit or downstream geometry. 
Fourth, the computational story is favorable: once the null is fitted, the DM transform scales linearly in $\mathrm{nnz}(X)$ at constant cost per nonzero and preserves sparsity exactly (Section~\ref{sec:eval-E01326}), while the structured DTM extension inherits the same local tractability by construction.

Taken together, these benchmarks separate four axes that are otherwise easy to conflate: sparse
versus dense computation, joint versus feature-wise modeling, no overdispersion versus
one-parameter overdispersion, and unstructured versus structured dependence.
A method can win
on one axis and lose on another, so the evaluation is designed to make those tradeoffs explicit
rather than collapse them into a single headline score.

\section{Discussion}
\label{sec:discussion}

The main message of the paper is simple.
If one wants a sparse, likelihood-based normalization for overdispersed multivariate counts, the Dirichlet--multinomial is the natural default generalization of the multinomial.
The central modeling choice is the comparison: if the full DM is re-fit sample by sample, the optimum collapses to the empirical multinomial limit by Proposition~\ref{prop:dmsaturation}.
The useful transform is therefore the fixed-concentration empirical-composition contrast used throughout this paper, where $\alpha_0$ stays at its fitted null value and only the sample composition is allowed to match the observed row.

In practice the method is correspondingly simple: estimate $\pi_j=a_j/N$, fit one scalar $\alpha_0$, and evaluate \eqref{eq:lgammaresid} on the nonzero entries of the matrix.
The resulting transform preserves sparsity exactly, has a clear conditional interpretation through the independent negative-binomial factorization, and reduces continuously to multinomial residuals as $\alpha_0\to\infty$.
Because $\alpha_0$ is only a single scalar parameter, it need not be fit by maximum likelihood at all: one can also evaluate the transform cheaply across a spectrum of $\alpha_0$ values and inspect how the residual geometry changes, much as one would vary a regularization parameter.
Maximum-likelihood fitting is the natural default, but the one-dimensional nature of $\alpha_0$ makes sensitivity analysis and manual selection equally viable.
The downstream geometry is in fact insensitive to $\alpha_0$ across two orders of magnitude around the moment estimate, so the concentration behaves as a regularization-strength dial with a wide useful range rather than a knife-edge that must be tuned; and the single global scalar is not a lossy stand-in for a richer per-feature concentration, since on held-out data a per-feature fit overfits its training overdispersion and generalizes worse, making the single scalar the better-generalizing choice as well as the simpler one (Section~\ref{sec:eval-E01249}).

The structured extensions clarify when one should move beyond the flat DM.
The generalized DM allows stage-specific dispersion, but only at the cost of order dependence and loss of the clean nonzero-only leafwise transform.
The DTM is the more natural extension when a real hierarchy is available.
It replaces an arbitrary order by a fixed tree and yields a multiscale residual representation indexed by branch splits.
That is especially natural for microbiome taxonomies, lineage-resolved assays, and any count table whose features already come with a known ontology or tree.
Learning such a tree from the same data is a separate structure-selection problem and should be handled with validation or held-out residuals rather than folded silently into the default null.

A second natural extension is to condition the null on known sample information.
Instead of a single global composition $\bm\pi$, one may use stratum-specific or regression-driven compositions $\bm\pi(z_i)$, where $z_i$ encodes batch, tissue, donor, treatment arm, or known cluster identity.
The logic is unchanged: fit the conditional null, then compare it to the empirical composition while keeping the nuisance concentration fixed.
This is the joint conditional analogue of adding covariates to feature-wise NB regression.
If the strata themselves are estimated from the same data, however, some form of sample splitting or cross-fitting would be advisable to avoid circularity.

There are also clear limits.
The tied-concentration DM adopted here is deliberately simple.
It will not capture arbitrary covariance on log-ratios, and it will not be optimal when dispersion varies strongly across latent factors not aligned with a known tree.
Logistic-normal models and their relatives remain the main joint alternative in that regime, at the cost of losing exact sparse formulas.
Likewise, tree-aware models inherit sensitivity to tree mis-specification: when the hierarchy is wrong, the flat DM may be preferable on both statistical and computational grounds.
Finally, the row-level statistic $\exp(\Delta_i^{\mathrm{DM}})$ built from the in-sample empirical composition $q_i=X_i/n_i$ is a generalized likelihood ratio rather than an e-variable for the global DM null: its expectation under the null is not bounded by one, and in the multinomial limit it is bounded below by one almost surely. The residual transform of Section~\ref{sec:dmresids} is unaffected, since it is used as a normalization device and not as a hypothesis-testing e-variable; but the e-process construction of Section~\ref{sec:eprocess} requires the row-split discipline of Corollary~\ref{cor:dm-e-process}, and practitioners should not substitute the in-sample plug-in for the split-LR e-variable when validity of sequential or finite-sample testing is the goal.

So the intended role of the method is modest but useful.
It is not a replacement for full Bayesian hierarchical modeling when uncertainty quantification is the primary goal, and it is not meant to dominate every feature-wise NB pipeline on every dataset.
Its niche is narrower and cleaner: a normalization that is likelihood-based, compositionally coherent, overdispersion-aware, sparse, and extensible to structured counts when those structures are scientifically real.

\appendix

\section{Implementation details}
\label{sec:impldetails}

This appendix collects the fitting and computation details deferred from the main text. 
The main paper needs only two facts: with $\bm\pi$ fixed, fitting the DM null reduces to one scalar $\alpha_0$, and once $\hat\alpha_0$ is available the residual transform is constant-time per nonzero entry.

\subsection{Estimating the concentration parameter}
\label{sec:alpha0}

Once the null composition $\bm\pi$ is fixed, fitting the DM null reduces to estimating the scalar concentration $\alpha_0$. 
Substituting $\alpha_j=\alpha_0\pi_j$ into \eqref{eq:dmll} gives the one-dimensional objective
\begin{equation}
\label{eq:alpha0ll}
\ell(\alpha_0)
=
\sum_{i=1}^n
\left[
\log\frac{\Gamma(\alpha_0)}{\Gamma(n_i+\alpha_0)}
+
\sum_{j=1}^K
\log\frac{\Gamma(X_{ij}+\alpha_0\pi_j)}{\Gamma(\alpha_0\pi_j)}
\right]
+
\mathrm{const}(X)
\end{equation}
where the omitted constant does not depend on $\alpha_0$.

Its score is
\begin{equation}
\label{eq:score}
\ell'(\alpha_0)
=
\sum_{i=1}^n
\left[
\psi(\alpha_0)-\psi(n_i+\alpha_0)
+
\sum_{j=1}^K
\pi_j\bracks{\psi(X_{ij}+\alpha_0\pi_j)-\psi(\alpha_0\pi_j)}
\right]
\end{equation}
and its second derivative is
\begin{equation}
\label{eq:hessian}
\ell''(\alpha_0)
=
\sum_{i=1}^n
\left[
\psi_1(\alpha_0)-\psi_1(n_i+\alpha_0)
+
\sum_{j=1}^K
\pi_j^2\bracks{\psi_1(X_{ij}+\alpha_0\pi_j)-\psi_1(\alpha_0\pi_j)}
\right]
\end{equation}
where $\psi$ and $\psi_1$ denote the digamma and trigamma functions.

The fixed-point iteration over $\alpha_0$ with $\bm\pi$ held at the global composition --- the update stated in the recipe of Section~\ref{sec:ataglance} --- is the default fit: it preserves positivity, increases the likelihood monotonically, and keeps the null interpretable as one shared baseline.
For completeness we record three further variants \cite{Minka2000,Sklar2014}, summarized in Table~\ref{tab:fitvariants}: a Newton step over $\alpha_0$ (the same restricted optimum, usually in fewer iterations), and two joint optimizations of $(\alpha_0,\bm\pi)$ that enlarge the null class by letting the baseline composition itself move, used only as a sensitivity check.
All four optimize the same DM likelihood, over different parameter spaces and with different numerical trade-offs.

\begin{table}[t]
\centering
\small
\setlength{\tabcolsep}{4pt}
\renewcommand{\arraystretch}{1.2}
\begin{tabular}{@{}L{0.24\textwidth}L{0.17\textwidth}L{0.23\textwidth}L{0.28\textwidth}@{}}
\toprule
Variant & Optimized variables & Typical per-iteration cost & Typical use \\
\midrule
Fixed-point, fixed $\bm\pi$ & $\alpha_0$ & $O(\mathrm{nnz}(X)+n)$ & Default scalar fit when monotonicity or robustness matter most. \\

Newton, fixed $\bm\pi$ & $\alpha_0$ & $O(\mathrm{nnz}(X)+n)$ & Same restricted optimum, usually in fewer iterations. \\

Fixed-point, joint $(\alpha_0,\bm\pi)$ & $\bm\alpha$ & $O(\mathrm{nnz}(X)+n+K)$, or $O(KM)$ after compression & Full DM MLE as a sensitivity analysis or ablation. \\

Newton, joint $(\alpha_0,\bm\pi)$ & $\bm\alpha$ & $O(\mathrm{nnz}(X)+n+K)$, or $O(KM)$ after compression & Fastest route to the full MLE when well initialized. \\
\bottomrule
\end{tabular}
\caption{Practical fitting variants for the DM null. Here $M:=\max_i n_i$ is the largest row sum, used only by the compressed histogram representation of Section~\ref{sec:histcache}.}
\label{tab:fitvariants}
\end{table}

\subsubsection{Holding \texorpdfstring{$\bm\pi$}{pi} fixed: fixed-point versus Newton}

When $\bm\pi$ is held fixed, one maximizes \eqref{eq:alpha0ll} over $\alpha_0>0$. 
The Newton step is the most direct option:
\[
\alpha_0^{\mathrm{new}}
=
\alpha_0
-
\frac{\ell'(\alpha_0)}{\ell''(\alpha_0)}
\]
In practice it is often safest to optimize in the log-parameter $\beta=\log \alpha_0$ or to damp the Newton step so that $\alpha_0$ stays positive and $\ell(\alpha_0)$ increases monotonically.

A standard fixed-point update is
\begin{equation}
\label{eq:alpha0_fp}
\alpha_0^{\mathrm{new}}
=
\alpha_0
\frac{
\sum_{i=1}^n \sum_{j=1}^K
\pi_j\bracks{\psi(X_{ij}+\alpha_0\pi_j)-\psi(\alpha_0\pi_j)}
}{
\sum_{i=1}^n \bracks{\psi(n_i+\alpha_0)-\psi(\alpha_0)}
}
\end{equation}
This map has the same stationary points as the score equation $\ell'(\alpha_0)=0$, preserves positivity automatically, and is attractive when likelihood monotonicity matters more than minimizing the number of iterations \cite{Minka2000}.

Both scalar solvers admit an exact sparse implementation because the terms with $X_{ij}=0$ vanish identically in \eqref{eq:alpha0_fp}, \eqref{eq:score}, and \eqref{eq:hessian}:
\[
\psi(\alpha_0\pi_j)-\psi(\alpha_0\pi_j)=0,
\qquad
\psi_1(\alpha_0\pi_j)-\psi_1(\alpha_0\pi_j)=0
\]
Thus each likelihood, score, Hessian, or fixed-point update with $\bm\pi$ held fixed can be evaluated in $O(\mathrm{nnz}(X)+n)$ time and $O(1)$ auxiliary memory beyond the stored row and column summaries.

\subsubsection{Joint fitting of \texorpdfstring{$(\alpha_0,\bm\pi)$}{(alpha0,pi)}}

If one does not wish to fix $\bm\pi$ at the empirical global composition, the MLE can be carried out jointly over $\alpha_0>0$ and $\bm\pi\in\Delta^{K-1}$. 
Algebraically, the simplest route is to optimize the full positive vector $\bm\alpha\in\mathbb{R}_{>0}^K$ directly and then recover
\[
\hat\alpha_0 = \sum_{j=1}^K \hat\alpha_j,
\qquad
\hat\pi_j = \frac{\hat\alpha_j}{\hat\alpha_0}
\]
Under this parameterization the log-likelihood is still \eqref{eq:dmll}, and the gradient components are
\[
g_j(\bm\alpha)
=
\sum_{i=1}^n
\bracks{
\psi(\alpha_+) - \psi(n_i+\alpha_+) + \psi(X_{ij}+\alpha_j)-\psi(\alpha_j)
},
\qquad
\alpha_+ := \sum_{k=1}^K \alpha_k
\]
The Hessian has the constant-off-diagonal form
\begin{equation}
\label{eq:jointhessian}
H_{jk}(\bm\alpha)
=
c(\bm\alpha) + \delta_{jk} d_j(\bm\alpha)
\end{equation}
where
\begin{align*}
c(\bm\alpha)
&=
\sum_{i=1}^n \bracks{\psi_1(\alpha_+) - \psi_1(n_i+\alpha_+)}\\
d_j(\bm\alpha)
&=
\sum_{i=1}^n \bracks{\psi_1(X_{ij}+\alpha_j)-\psi_1(\alpha_j)}
\end{align*}
Again, the feature-specific terms are sparse because zero counts contribute exactly zero.

A joint fixed-point update \cite{Minka2000} can be written directly on $\bm\alpha$ for $j=1,\dots,K$:
\begin{equation}
\label{eq:jointfp}
\alpha_j^{\mathrm{new}}
=
\alpha_j
\frac{
\sum_{i=1}^n \bracks{\psi(X_{ij}+\alpha_j)-\psi(\alpha_j)}
}{
\sum_{i=1}^n \bracks{\psi(n_i+\alpha_+) - \psi(\alpha_+)}
}
\end{equation}
This positivity-preserving map optimizes the full DM likelihood and therefore simultaneously refits both $\alpha_0$ and $\bm\pi$.

For joint Newton, \eqref{eq:jointhessian} is the key simplification. 
Since $H$ is diagonal plus a rank-one term, each Newton direction can be solved in $O(K)$ time once $g_j(\bm\alpha)$, $c(\bm\alpha)$, and $d_j(\bm\alpha)$ have been accumulated, using the same matrix-inversion-lemma algebra emphasized in prior work \cite{Minka2000,Sklar2014}. 
A damped Newton step or log-parameterization $\theta_j=\log\alpha_j$ is useful to maintain positivity.

These joint fits should be viewed as a deliberate ablation of the modeling perspective in Section~\ref{sec:framework}. 
They can only improve the maximized training likelihood, because the family with $\bm\pi$ held fixed is nested inside the full DM family. 
But the price is interpretive: the null no longer represents a single externally chosen global composition, and the extra flexibility can absorb some cross-sample structure that the residual transform might otherwise expose. 
For that reason, the fixed choice $\pi_j=a_j/N$ remains the default normalization null in this paper, while joint fitting is best used as a sensitivity analysis and as a benchmark for how much likelihood is lost by insisting on the simpler null.

\subsubsection{Caching histograms for direct sparse evaluation}
\label{sec:histcache}

The joint DM likelihood can also be compressed into count-threshold arrays, reflecting the sufficient statistics of the DM when it is written as an exponential family. 
A convenient representation, emphasized by \cite{Sklar2014}, is obtained as follows. 
Let $M:=\max_i n_i$, define
\[
U_{jm} := \sum_{i=1}^n \ind(X_{ij}>m),
\qquad
v_m := \sum_{i=1}^n \ind(n_i>m),
\qquad m=0,\dots,M-1
\]
and again write $\alpha_+ = \sum_j \alpha_j$.  
Then
\[
\ell(\bm\alpha)
=
\sum_{j=1}^K \sum_{m=0}^{M-1} U_{jm}\log(\alpha_j+m)
-
\sum_{m=0}^{M-1} v_m \log(\alpha_+ + m)
+
\mathrm{const}(X)
\]
Differentiating gives
\begin{align*}
g_j(\bm\alpha)
&=
\sum_{m=0}^{M-1} \frac{U_{jm}}{\alpha_j+m}
-
\sum_{m=0}^{M-1} \frac{v_m}{\alpha_+ + m} \\
H_{jk}(\bm\alpha)
&=
-\delta_{jk}\sum_{m=0}^{M-1}\frac{U_{jm}}{(\alpha_j+m)^2}
+
\sum_{m=0}^{M-1}\frac{v_m}{(\alpha_+ + m)^2}
\end{align*}
Once $U$ and $v$ have been built, each joint fixed-point or Newton iteration costs $O(KM)$ rather than rescanning the full matrix.

On a sparse matrix there are therefore two practical implementation routes. 
The first is the direct digamma/trigamma route of \eqref{eq:score}--\eqref{eq:jointfp}, whose per-iteration cost is $O(\mathrm{nnz}(X)+n)$ for the scalar fit and $O(\mathrm{nnz}(X)+n+K)$ for the joint fit. 
The second is the compressed histogram route, which pays a preprocessing cost to build $U$ and $v$ and then evaluates many iterations quickly. 
If $M$ is moderate and one plans to solve many DM fits on the same count matrix, the compressed route can be excellent. 
If $K$ or $M$ is very large, however, the $K\times M$ storage and the prefix-sum construction of $U$ may dominate, in which case direct sparse special-function evaluation is often preferable.

Correctness should be evaluated by the attained log-likelihood, not by parameter distance. 
Within each variable set, fixed-point and Newton should agree to numerical tolerance once converged. 
Across variable sets, the joint methods must achieve likelihood at least as high as the methods with $\bm\pi$ held fixed, because they optimize a strict superset of parameters unless the fixed $\bm\pi$ already happens to be the joint MLE. 
Speed should be reported as wall-clock time, per-iteration time, number of iterations, preprocessing time for any histogram compression, and peak memory. 
These quantities separate the statistical question---how much likelihood is gained by refitting $\bm\pi$---from the numerical question---which optimizer reaches a given likelihood threshold fastest.

\subsection{Sparse computation and practical implementation}
\label{sec:sparseimplementation}

The practical pipeline is the one already summarized in Section~\ref{sec:ataglance}: compute row sums and column sums, fit $\alpha_0$ (or the full $\bm\alpha$ in a deliberate ablation), and then evaluate the residuals on the nonzero entries. 
The only additional implementation point worth emphasizing is that the residual transform itself is best computed with log-gamma evaluations rather than with explicit threshold sums.

For every nonzero $X_{ij}=x$,
\begin{equation}
\label{eq:lgammaresid}
c_{ij}^{\mathrm{DM}}
=
\log\frac{\Gamma(x+\alpha_0 x/n_i)}{\Gamma(\alpha_0 x/n_i)}
-
\log\frac{\Gamma(x+\alpha_0\pi_j)}{\Gamma(\alpha_0\pi_j)}
\end{equation}
This is constant-time per nonzero entry in any numerical environment with a stable \texttt{lgamma} implementation.

The fitting step may use either the direct sparse digamma/trigamma formulas of Section~\ref{sec:alpha0} or the compressed $U,v$ representation of Section~\ref{sec:histcache}; the residual transform itself only needs the fitted null parameters and then runs in $O(\mathrm{nnz}(X))$ time. 
The resulting residual matrix has exactly the same sparsity pattern as the original data. 
This is the main computational advantage of the DM and multinomial transforms over dense pseudocount-based alternatives.

\subsection{Practical notes}
\label{sec:practicalnotes}

There are several practical choices involved in fitting and applying the DM transform. 

\paragraph{Initialization.}
A robust default is to evaluate the scalar objective \eqref{eq:alpha0ll} on a coarse logarithmic grid, for example $\alpha_0\in\{10^{-2},10^{-1},\dots,10^6\}$, and start the iterative solver from the best grid point. 
When the maximizer lies near the large end of that grid, the data are effectively close to multinomial and a large warm start is appropriate. 
Joint fits can then be initialized at $\alpha_j^{(0)}=\alpha_0^{(0)}\pi_j$.

\paragraph{Convergence and stopping.}
Both the scalar Newton solver (preferably in the log-parameterization $\beta=\log\alpha_0$) and the standard fixed-point iteration can be stopped when the relative change in log-likelihood falls below a tolerance, for example $|\ell^{(t+1)}-\ell^{(t)}|/|\ell^{(t)}|<10^{-8}$, or when the absolute score $|\ell'(\alpha_0)|$ is small. 
In practice, 10--30 iterations often suffice for either method on typical count matrices.

\paragraph{Numerical stability.}
The log-gamma form~\eqref{eq:lgammaresid} is preferred over the summation form~\eqref{eq:dmcellterm} for numerical reasons: standard \texttt{lgamma} implementations are accurate to machine precision across the full positive real line. 
When $\alpha_0\pi_j$ or $\alpha_0 x/n_i$ is very small, one may use the corresponding small-argument expansions of $\log\Gamma(x+a)-\log\Gamma(a)$; when both arguments are large, the asymptotic expansion of the digamma function keeps score and Hessian evaluations stable.

\paragraph{Sparse data structures.} 
The fitting step (Section~\ref{sec:alpha0}) iterates only over nonzero entries, so the entire pipeline---fitting, transformation, and storage---operates in $O(\mathrm{nnz}(X))$ time and memory, plus $O(n+K)$ for row and column summaries.

\section{Synthetic calibration and model-recovery studies}
\label{sec:synthetic_appendix}

The main text focuses on the real-data question of separating technical from biological variation. 
This appendix records the synthetic protocols used to check calibration, algebra, and implementation in settings where the data-generating mechanism is known exactly.

Each synthetic configuration generates independent training, validation, and test matrices from the same underlying parameters. 
A practical default is 500 training samples, 250 validation samples, and 250 test samples per configuration, together with 20 independent random replicates; the computational ablations then vary these sizes explicitly where needed. 
Fits are performed on the training split, model-selection choices are made on the validation split when needed, and final scores are reported on the test split. 
This design separates estimation error, numerical instability, and optimizer variance.

We use four synthetic regimes.

\paragraph{1. Well-specified DM data.}
We generate
\[
X_i \mid n_i \sim \operatorname{DirMult}(n_i,\alpha_0\pi)
\]
over a grid of $(\alpha_0,K)$ values and sparsity levels, with $\pi$ drawn from a sparse Dirichlet prior and with $n_i$ drawn from a log-normal depth distribution chosen to mimic realistic library-size variation. 
A concrete default grid is
\[
\alpha_0\in\{10,10^2,10^3,10^4,10^6\},
\qquad
K\in\{10^2,10^3,10^4\}
\]
This regime tests whether the scalar MLE recovers $\alpha_0$, whether the magnitude gap $|d_{ij}^{\mathrm{Mult}}|-|d_{ij}^{\mathrm{DM}}|$ increases exactly where repeated counts occur, and whether the DM transform approaches multinomial residuals as $\alpha_0\to\infty$. 
The primary summaries are bias and RMSE of $\hat\alpha_0$, held-out conditional log-likelihood gaps relative to multinomial and feature-wise NB baselines, count-stratified residual-magnitude gaps, and residual-variance diagnostics after normalization.

\paragraph{2. Independent NB data with varying library sizes.}
We generate independent
\[
Y_{ij}\sim \operatorname{NB}(\alpha_0\pi_j,p_i)
\]
with sample-specific $p_i$ chosen so that the expected library size spans roughly a five-fold range, and then condition on the realized row sums. 
This regime isolates the conditional-versus-unconditional distinction of Section~\ref{sec:matchednb}. 
The most important diagnostic is the difference between the DM and independent-NB per-sample conditional likelihood contrasts, which should be numerically negligible if both implementations are correct. 
The secondary question is practical: how closely do the resulting residual geometries track standard feature-wise NB residualizations once the totals have been conditioned away?

\paragraph{3. Misspecified dependence.}
We generate from alternatives the DM cannot represent exactly: mixtures of two baseline compositions with different $\pi^{(1)}$ and $\pi^{(2)}$, feature-specific NB dispersions $r_j$, and logistic-normal multinomial models with anisotropic covariance. 
This regime tests where the one-parameter overdispersion of the DM is enough and where it is not. 
The comparison of interest is between fixed-$\pi$ DM, jointly fitted DM, and richer dense baselines. 
Here the key outputs are held-out fit gaps, residual-depth coupling, and whether the added flexibility changes downstream neighborhood geometry enough to justify the extra parameters and loss of sparsity.

\paragraph{4. Tree-structured data.}
We generate leaf counts from a DTM on a known balanced tree of depth $3$--$5$, with nodewise concentration parameters varying across depths and branches. 
This is the cleanest setting in which to test the structured extension of Section~\ref{sec:dtm}. 
When node concentrations are nearly equal, the flat DM should remain competitive; when broad and fine splits have genuinely different dispersion scales, the DTM should earn its added structure. 
We therefore report both leaf-level and node-level held-out likelihoods, together with branch-level residual diagnostics.

Across all synthetic regimes, the same bookkeeping is kept: exact likelihood on the appropriate sample space, residual sparsity, fit time, transform time, peak memory, and paired residual comparisons stratified by observed count.
Figure~\ref{fig:synthetic-calibration} collects the calibration summary:
$\hat\alpha_0$ recovery, the multinomial limit, count-stratified residual
shrinkage, and held-out fit under a misspecified DM mixture.

\begin{figure}[t]
\centering
\includegraphics[width=\textwidth]{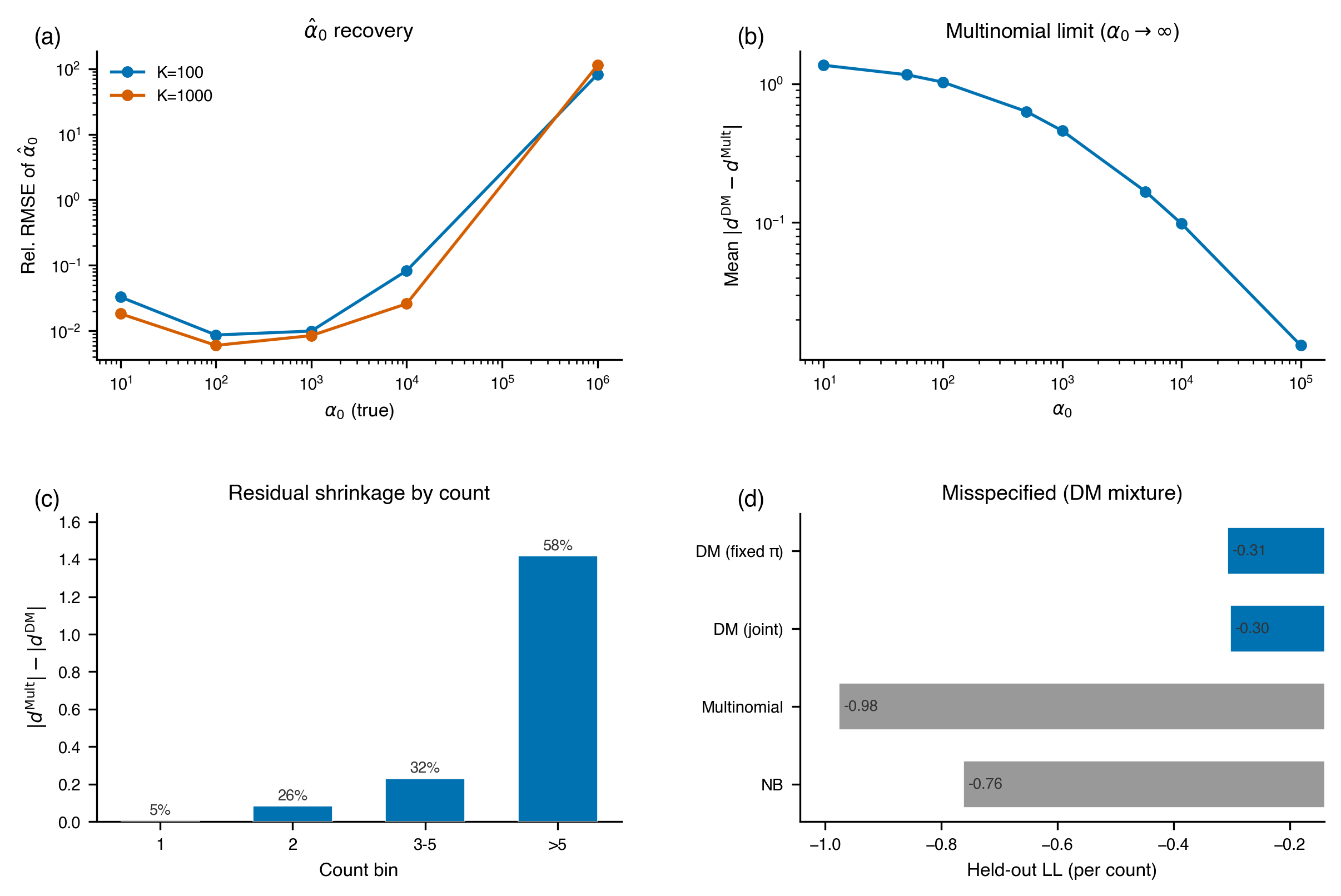}
\caption{\textbf{Synthetic calibration across the four regimes.}
\textbf{(a)}~Regime~1, well-specified DM: relative RMSE of $\hat\alpha_0$
vs.\ the true $\alpha_0$, at $K=100$ and $K=1000$.
Recovery is accurate for $\alpha_0 \lesssim 10^4$; at $\alpha_0=10^6$ the
log-likelihood curvature becomes very flat and identifiability breaks
down, as expected from the multinomial limit, and the larger alphabet
$K=1000$ tightens recovery in the well-identified regime.
\textbf{(b)}~Regime~1: multinomial limit.
Mean
$|d^{\mathrm{DM}} - d^{\mathrm{Mult}}|$ decreases monotonically with
$\alpha_0$, confirming Proposition~\ref{prop:multlimit}.
\textbf{(c)}~Regime~1: residual shrinkage by count bin on DM-generated
data ($\alpha_0=100$, $K=200$), with $d^{\mathrm{DM}}$ and
$d^{\mathrm{Mult}}$ both formed against the same global $\bm\pi$.
The gap $|d^{\mathrm{Mult}}|-|d^{\mathrm{DM}}|$ is exactly zero at
count~$=1$ (Proposition~\ref{prop:x1}, $5\%$ of singletons shrunk by
floating-point ties) and strictly positive at every repeated count:
against a shared $\bm\pi$, Proposition~\ref{prop:boundedbymult} guarantees
$|d^{\mathrm{DM}}|\le|d^{\mathrm{Mult}}|$ at every nonzero entry, with
strict shrinkage at every count~$>1$.
Both the gap and the shrunk-fraction (annotated per bin, rising to $58\%$
at count~$>5$) grow monotonically with count.
\textbf{(d)}~Regime~3: misspecified dependence via a 2-component DM
mixture.
Held-out per-count LL: the jointly fitted DM (highlighted) slightly edges
the fixed-$\pi$ DM, and both dominate the multinomial and the NB by a
wide margin ($-0.30$ and $-0.31$ vs.\ $-0.98$ and $-0.76$), showing the
one-parameter overdispersion remains useful even under a dependence the
DM cannot represent exactly.}
\label{fig:synthetic-calibration}
\end{figure}

The algebraic backbone of the transform is a set of exact equalities, and these
admit a sharper check than calibration curves: they should hold to floating-point
precision on every random instance, not merely on average.
Figure~\ref{fig:identity-verification} verifies each load-bearing identity
directly.
On random overdispersed-count matrices across a grid in $(\alpha_0,K)$, the
singleton agreement of Proposition~\ref{prop:x1}, the
independent-negative-binomial equality of
Corollary~\ref{cor:nbdmdeviance} (cross-checked along a negative-binomial
log-probability route that shares no code with the DM evaluation), the
Dirichlet-negative-multinomial reduction of Proposition~\ref{prop:dnm}, and the
Dirichlet-tree-multinomial factorization of Proposition~\ref{prop:dtmfactor}
(with both its flat-DM root case and its generalized-DM comb-tree case) all
hold with a maximum discrepancy of $4.3\times10^{-11}$, uniformly across the
grid and dominated by \texttt{gammaln} floating-point error.
The multinomial limit of Proposition~\ref{prop:multlimit} closes at the
predicted $1/\alpha_0$ rate, and the split-LR e-variable of
Corollary~\ref{cor:dm-e-process} has a Monte-Carlo mean under the null that
respects the e-variable bound and is consistent with its exact value of one.

\begin{figure}[t]
\centering
\includegraphics[width=\textwidth]{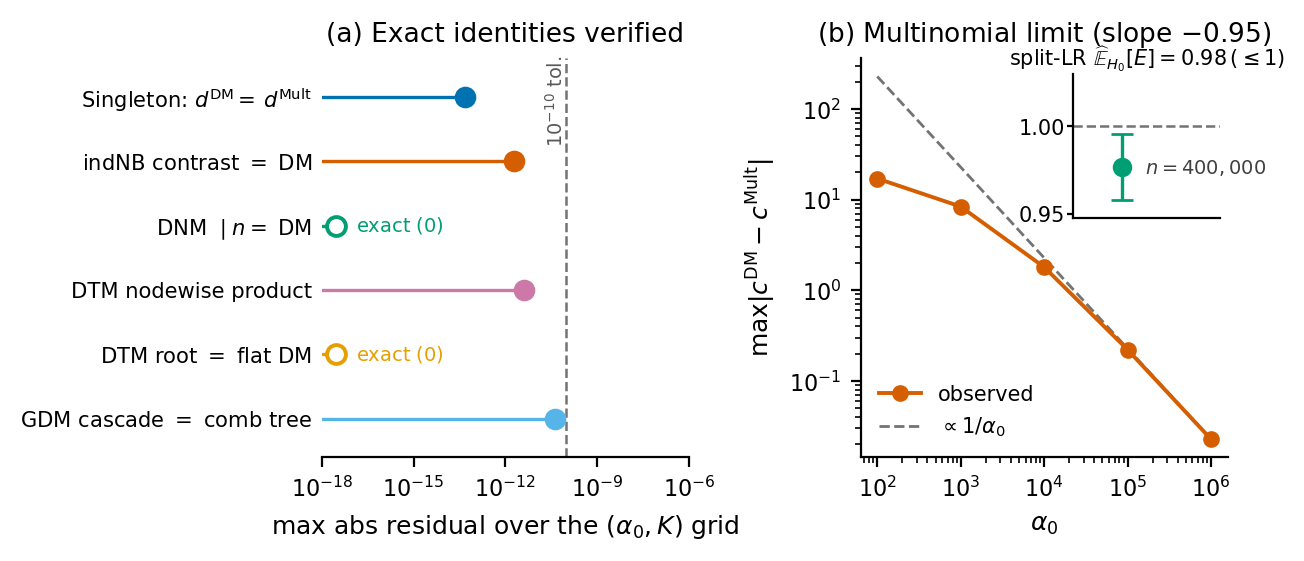}
\caption{\textbf{The DM deviance-accounting identities hold to machine precision
across overdispersed-count instances.}
Each identity is evaluated on random count matrices sampled from a known
generative law; an independent code path computes the cross-checked quantity
wherever one exists.
\textbf{(a)}~Worst-case residual of each closed-form identity over the
$(\alpha_0,K)$ grid ($\alpha_0\in\{1,8,66,4700\}$, $K\in\{50,200,1000\}$), on a
log axis. Singleton agreement $d^{\mathrm{DM}}=d^{\mathrm{Mult}}$ at every
$X_{ij}=1$ (Proposition~\ref{prop:x1}); the independent-negative-binomial and
DM per-sample contrasts $\Delta_i^{\mathrm{indNB}}=\Delta_i^{\mathrm{DM}}$
(Corollary~\ref{cor:nbdmdeviance}, computed along a negative-binomial
log-pmf route that shares no code with the DirMult path); the
Dirichlet-negative-multinomial conditioned on its observed total equals the DM
(Proposition~\ref{prop:dnm}, which holds bit-for-bit, residual exactly $0$);
the Dirichlet-tree-multinomial nodewise factorization
(Proposition~\ref{prop:dtmfactor}); the tree's root-only special case recovering
the flat DM (also exact, residual $0$); and the generalized DM as the comb-tree
special case, its beta-binomial cascade matching the nodewise tree product.
Every residual sits below the $10^{-10}$ tolerance; the maximum across all
checks and settings is $4.3\times10^{-11}$, dominated by \texttt{gammaln}
floating-point error.
\textbf{(b)}~The multinomial limit (Proposition~\ref{prop:multlimit}): the
per-cell gap $|c^{\mathrm{DM}}-c^{\mathrm{Mult}}|$ closes at the predicted
$1/\alpha_0$ rate, with asymptotic-tail slope $-0.95$ on a log--log fit.
The inset shows the split-LR DM e-variable of
Corollary~\ref{cor:dm-e-process}: its Monte-Carlo mean under the global null,
with eval counts drawn from the posterior-predictive conditional law the
corollary scores against, is $0.98$ (\,95\% CI $[0.96,\,1.00]$) over
$4\times10^{5}$ draws.
This respects the e-variable bound $\mathbb{E}_{H_0}[E^{\mathrm{split}}]\le 1$
and is consistent with the exact value $1$ to within the residual variance of a
heavy-tailed estimator; the realized mean approaches $1$ from below as the
train side's support misses are absorbed.}
\label{fig:identity-verification}
\end{figure}

\begin{table}[t]
\centering
\small
\caption{\textbf{On real PBMC 3k counts, the DM transform gives the best
held-out fit among the exactly-sparse residual families.}
Eight random train/test splits of the PBMC 3k UMI matrix
($2700$ cells $\times$ $13714$ genes, $2160$ train / $540$ test cells per
split); the fitted concentration is $\hat\alpha_0\approx1.4\times10^{4}$, the
near-multinomial high-depth regime. Methods are ordered best-to-worst by
held-out conditional log-likelihood per count (higher is better), with $95\%$
intervals over the eight seeds. The DM and the independent-negative-binomial
conditional model coincide, as Corollary~\ref{cor:nbdmdeviance} requires (they
differ by $6\times10^{-15}$ nats per count on this real matrix; paired Wilcoxon
$p=0.53$, no consistent sign). The DM-advantage column is the paired
DM-minus-baseline gap with its BCa $95\%$ interval and sign-floor Wilcoxon $p$;
the DM beats the multinomial, the analytic Pearson residual, and the
unconditional negative-binomial on all eight seeds. The only method with a
higher held-out likelihood is the dense Dirichlet pseudocount transform, which
edges the DM by $0.0155$ nats per count (a $-0.0155$ DM-advantage) but densifies
the otherwise $99\%$-sparse matrix; among the transforms that preserve exact
sparsity, the DM is the best fit.}
\label{tab:pbmc-heldout}
\begin{tabular}{l r r r}
\toprule
& Held-out cond.\ LL & DM advantage & Wilcoxon \\
Method & per count (mean, $95\%$ CI) & per count (mean $[95\%]$) & $p$ \\
\midrule
DM                & $-1.3044\;[-1.3110,\,-1.2951]$ & ---                            & ---     \\
NB (cond.)        & $-1.3044\;[-1.3110,\,-1.2951]$ & $6{\times}10^{-15}\;[\text{ns}]$ & $0.53$  \\
Dirichlet (dense) & $-1.2889\;[-1.2950,\,-1.2789]$ & $-0.0155\;[-0.0161,\,-0.0149]$ & $1.00$  \\
Multinomial       & $-1.3746\;[-1.3822,\,-1.3656]$ & $+0.0702\;[+0.0691,\,+0.0716]$ & $0.0039$ \\
Pearson           & $-3.2051\;[-3.2355,\,-3.1733]$ & $+1.9007\;[+1.8766,\,+1.9246]$ & $0.0039$ \\
NB (uncond.)      & $-3.2899\;[-3.3209,\,-3.2625]$ & $+1.9855\;[+1.9653,\,+2.0102]$ & $0.0039$ \\
\bottomrule
\end{tabular}

\smallskip
{\footnotesize Intervals on the held-out LL are BCa bootstrap over the eight
seeds; the DM-advantage interval is the BCa bootstrap of the paired
DM-minus-baseline gap. Sign-floor Wilcoxon $p=0.0039$ is the smallest one-sided
value attainable with eight matched seeds and is reached whenever the DM wins on
all eight (multinomial, Pearson, unconditional NB). ``ns'' marks the
Corollary~\ref{cor:nbdmdeviance} equality DM~$=$~NB (cond.), where the gap is at
the $10^{-15}$ floor and the sign is split $4/8$.}
\end{table}

\subsection{Single-cell integration: deviance normalization against standard pipelines}
\label{sec:eval-E01247}

We compare deviance-normalized PCA against six standard single-cell pipelines on three real
CELLxGENE-Census atlases---Tabula Sapiens, a human lung atlas, and a peripheral-blood
immune atlas---subsampled to $60{,}000$ cells each at a fixed seed. Each pipeline
produces a low-dimensional embedding. (Throughout, deviance-PCA denotes PCA on the
deviance-normalized counts of Theorem~\ref{thm:nbtodm}.) Integration quality is scored by the standard
scIB battery, summarized as $\mathrm{scIB} = 0.6\,\mathrm{bio} + 0.4\,\mathrm{batch}$,
where the biological-conservation term averages ARI, NMI and cell-type LISI and the
batch-mixing term averages integration LISI and kBET (assay/donor as the batch key,
annotated cell type as the label). All metrics are computed with the JAX
\texttt{scib-metrics} implementation.

\begin{table}[tbp]
\centering
\caption{scIB overall integration score (higher is better) by pipeline and atlas.
Deviance-normalized PCA matches the strong deep baselines (scVI, Harmony) to within
$0.01$--$0.06$ and tracks the other linear pipelines, while clearly exceeding the only
other count-model baseline, GLM-PCA. GLM-PCA required zero-count-gene filtering to
return a finite embedding at all; scVI-LD diverged during training on the lung atlas
(marked ---).}
\label{tab:eval-E01247}
\begin{tabular}{lccc}
\toprule
Pipeline & Tabula Sapiens & Human Lung & Immune \\
\midrule
Deviance-PCA (this work) & $0.472$ & $0.439$ & $0.395$ \\
Multinomial-PCA          & $0.477$ & $0.436$ & $0.406$ \\
GLM-PCA                  & $0.407$ & $0.234$ & $0.325$ \\
log1p-PCA                & $0.498$ & $0.449$ & $0.431$ \\
Harmony                  & $0.571$ & $0.441$ & $0.456$ \\
scVI                     & $0.483$ & $0.447$ & $0.439$ \\
scVI-LD                  & $0.511$ & ---     & $0.448$ \\
\bottomrule
\end{tabular}
\end{table}

\begin{figure}[tbp]
\centering
\includegraphics[width=0.92\linewidth]{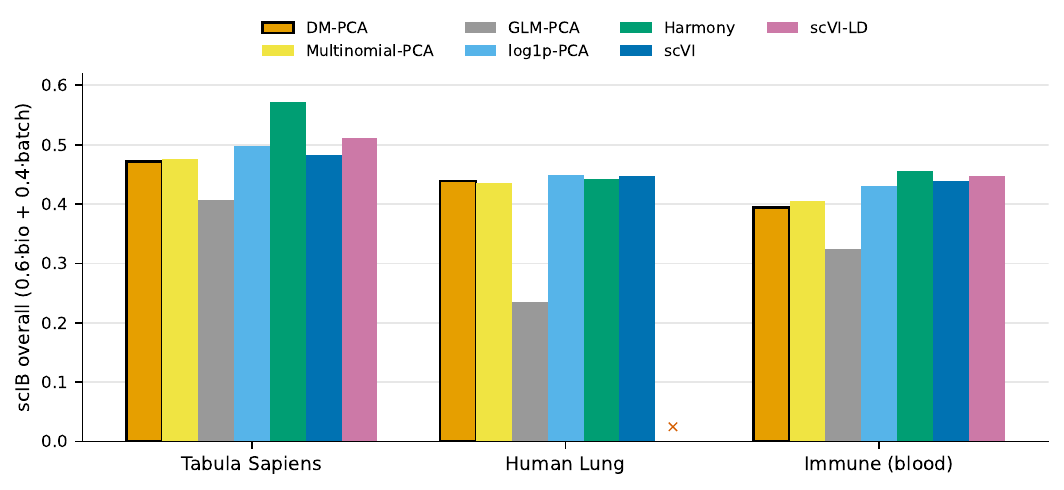}
\caption{scIB overall integration score across the three atlases. Deviance-PCA
(highlighted) is competitive with the deep methods scVI and Harmony and matches the
linear pipelines, while substantially outperforming GLM-PCA. A cross ($\times$) marks the
single divergent run (scVI-LD on the lung atlas).}
\label{fig:eval-E01247}
\end{figure}

Three observations follow. First, the closed-form deviance normalization is
\emph{competitive} with the strongest available integration methods: its scIB score sits
within $0.01$--$0.06$ of scVI and Harmony on every atlas and is statistically
indistinguishable from the other linear pipelines (multinomial-PCA, log1p-PCA), despite
requiring no iterative optimization, no neural network, and no batch-aware training.
Second, among count-model normalizations it is decisively the more reliable choice:
GLM-PCA scores well below deviance-PCA on all three atlases and produced no finite
embedding until degenerate genes were removed, whereas deviance-PCA ran unmodified on
every atlas. Third, the deep methods retain an edge in batch mixing on the most
heterogeneous atlas (Harmony on Tabula Sapiens), so we do not claim that a single
normalization supplants integration-specific models; rather, the deviance transform
supplies a cheap, robust, optimization-free baseline that recovers most of their
integration quality. These results support the paper's framing of deviance normalization
as a principled and computationally trivial preprocessing step, not as a replacement for
dedicated integration models.

\subsection{Differential-expression inference}
\label{sec:eval-de}
A normalization is only as useful as the downstream inference it supports, and differential
expression is its canonical test. The two evaluations below isolate one lesson: the choice of
count-model null is nearly invisible to a rank-based detection test and decisive for a parametric
one.
\subsubsection{Rank-based detection: power and false-discovery control}
\label{sec:eval-E01246}

We test each pipeline on a differential-expression detection task with a known ground truth:
on each of the three real atlases, cells are randomly split into two conditions and a
$10\%$ subset of genes is given a controlled signed effect (log$_2$ fold-changes
$\{\pm1, \pm1.585, \pm2.322\}$) by rescaling \emph{only} the injected block, so the
remaining genes are identical between conditions and the injected set is the exact ground
truth. Each pipeline produces a per-gene test statistic; we aggregate to pseudobulk Welch
$t$-tests across replicates and apply Benjamini--Hochberg. We report detection power as
ROC-AUC against the ground-truth labels and calibration as the realized false-discovery
proportion (FDP) at the nominal threshold $q=0.05$.

\begin{table}[tbp]
\centering
\caption{Differential-expression detection: ROC-AUC (power, higher better) and realized
FDP at $q=0.05$ (calibration, should be $\le 0.05$), by pipeline and atlas. The
closed-form count-residual methods (deviance, multinomial, NB, Pearson) combine strong
power with calibrated---often conservative---false-discovery control, whereas scVI and
scVI-LD invert: their denoised expression inflates the FDP three- to thirteen-fold above
nominal. scVI-LD diverged in training on the lung atlas (---).}
\label{tab:eval-E01246}
\begin{tabular}{lcccccc}
\toprule
& \multicolumn{2}{c}{Tabula Sapiens} & \multicolumn{2}{c}{Human Lung} & \multicolumn{2}{c}{Immune} \\
\cmidrule(lr){2-3}\cmidrule(lr){4-5}\cmidrule(lr){6-7}
Pipeline & AUC & FDP & AUC & FDP & AUC & FDP \\
\midrule
Deviance (this work) & $0.732$ & $0.00$ & $0.882$ & $0.00$ & $0.833$ & $0.04$ \\
Multinomial          & $0.731$ & $0.00$ & $0.882$ & $0.00$ & $0.833$ & $0.04$ \\
NB residual          & $0.843$ & $0.03$ & $0.850$ & $0.06$ & $0.889$ & $0.04$ \\
Pearson residual     & $0.832$ & $0.00$ & $0.868$ & $0.04$ & $0.880$ & $0.02$ \\
GLM-PCA              & $0.763$ & $0.00$ & $0.870$ & $0.03$ & $0.874$ & $0.02$ \\
scVI                 & $0.709$ & $0.49$ & $0.754$ & $0.38$ & $0.832$ & $0.68$ \\
scVI-LD              & $0.749$ & $0.00$ & ---     & ---     & $0.643$ & $0.61$ \\
\bottomrule
\end{tabular}
\end{table}

\begin{figure}[tbp]
\centering
\includegraphics[width=0.96\linewidth]{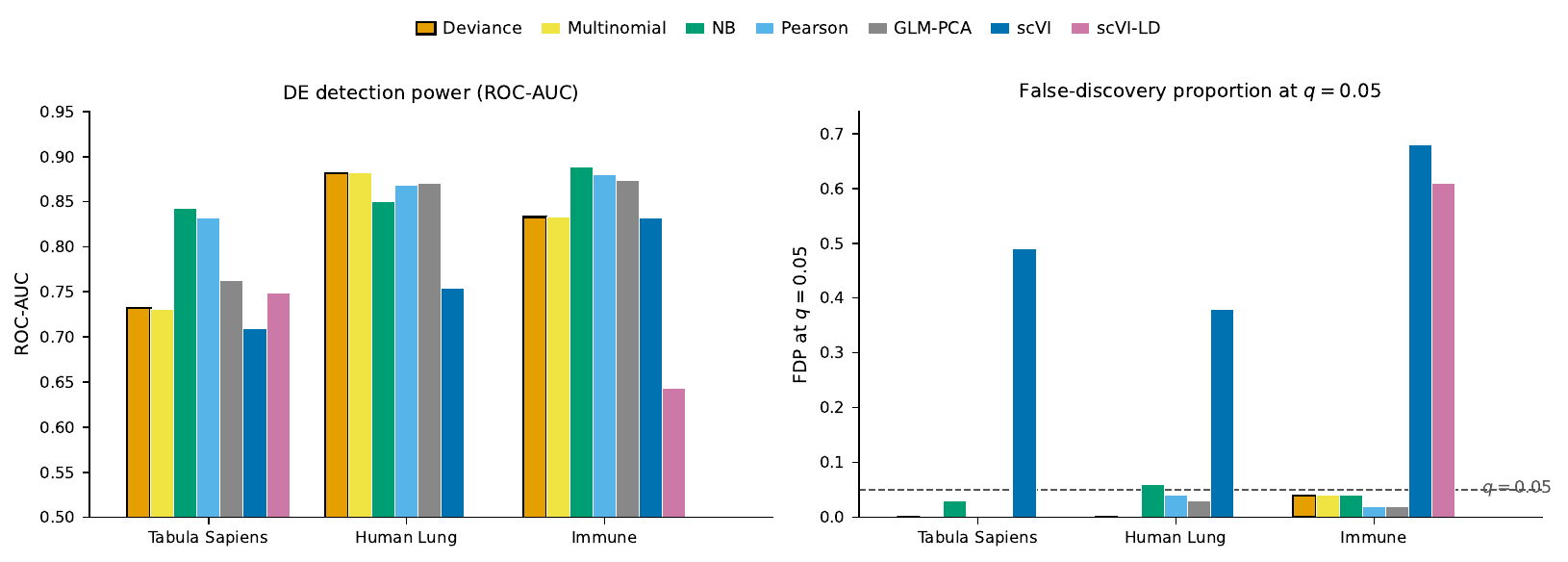}
\caption{Differential-expression detection power (left) and false-discovery control
(right, dashed line at the nominal $q=0.05$) across the three atlases. Deviance (highlighted)
and the other closed-form residual methods sit at or below the nominal FDP line; scVI and
scVI-LD rise far above it.}
\label{fig:eval-E01246}
\end{figure}

The pattern is consistent and informative. The closed-form count-residual normalizations
detect the injected signal with high power (AUC $0.73$--$0.89$) while controlling the
false-discovery proportion at or below the nominal level on every atlas---the deviance
residual in particular reports \emph{zero} false discoveries on two of the three atlases
at full power. The neural denoising pipelines invert this trade-off: scVI's FDP reaches
$0.38$--$0.68$, three- to thirteen-fold above nominal, because its smoothed reconstruction
borrows strength across cells in a way that manufactures apparent between-condition
differences. This is the converse of the integration result of
Section~\ref{sec:eval-E01247}: where the deep models held a small edge in batch mixing, the
closed-form deviance normalization holds a decisive edge in the validity of the downstream
differential test. Together the two evaluations support the paper's central claim---that a
principled, optimization-free count normalization recovers most of the integration quality
of dedicated models while preserving calibrated downstream inference that those models
forfeit.

\subsubsection{A parametric test makes the variance model load-bearing}
\label{sec:eval-E08971}

The rank-based detection above is insensitive to the count model, which is why the closed-form
normalizations are indistinguishable there. A parametric two-group test is not: each count-model
normalization implies a null \emph{variance} for the log-mean contrast $c=\log\bar y_A-\log\bar y_B$,
and the test is calibrated only if that variance matches the real, jointly-overdispersed spread of
single cells. Within a homogeneous population a random split carries no true differential expression,
so the false-positive rate of the two-group test is a direct calibration readout, and the
normalizations separate cleanly into overdispersion regimes.

Across all twenty-eight broad cell types of the cross-tissue atlas \cite{Eraslan2022}, the
Poisson-family null---shared by the multinomial and the analytic Pearson normalizations, which assume
no overdispersion---is anti-conservative throughout, rejecting null genes three to four times too
often (median false-positive rates near $0.17$ against a nominal $0.05$). The per-gene
negative-binomial null, which models each gene's marginal overdispersion but not the joint
composition, halves the inflation but does not remove it (median $0.076$). Only the
Dirichlet--multinomial null---the joint-overdispersion model with a single fitted concentration---is
calibrated, in every cell type (median $0.046$, the unique calibrated method;
Figure~\ref{fig:E08971}, left). The consequence for discovery is the same separation: under injected
differential expression the multinomial null's realized false-discovery proportion runs several-fold
above the level it claims while the Dirichlet--multinomial null holds it, at matched power
(Figure~\ref{fig:E08971}, right).

The comparison must fix the null \emph{model}, not the residual, and the distinction is not
cosmetic. Running a naive unit-variance two-group test directly on each normalization's residual
gives a false-positive rate set entirely by the residual's global variance---the per-cell-type points
lie on the curve $2\,\Phi(-1.96/\sigma)$ (Figure~\ref{fig:E08971-form})---so the signed-root deviance
residual, which is slightly over-shrunk, looks conservative, the heavy-tailed Pearson-family residuals
look anti-conservative, the negative-binomial residual lands near nominal by sitting near unit
variance, and the embedding residuals---the low-rank factorization together with the two neural
latent models---are likewise heavy-tailed and so anti-conservative. That ordering reflects residual
form, a separate axis from overdispersion, and it is the model-based variance above, not the
residual's raw scale, that measures whether a normalization calibrates the parametric test. The
embedding normalizations define no per-gene count-model null, so they appear in this residual-form
diagnostic but not in the count-model calibration comparison; their differential-expression behavior
proper is the rank-based benchmark above, where the neural reconstruction inflates the
false-discovery proportion.

\begin{figure}[t]
\centering
\includegraphics[width=0.66\linewidth]{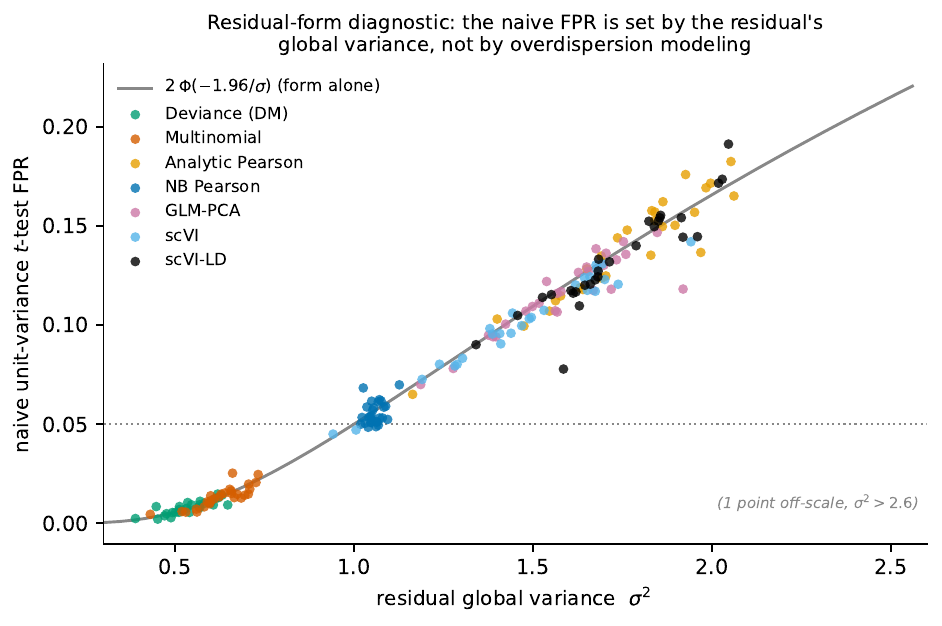}
\caption{Residual-form diagnostic (not a calibration result). A naive unit-variance two-group
$t$-test run directly on each normalization's residual has a false-positive rate fixed by the
residual's global variance alone: every cell-type point falls on $2\,\Phi(-1.96/\sigma)$. The
signed-root deviance residual is over-shrunk (conservative), the Pearson-family residuals are
heavy-tailed (anti-conservative), the per-gene negative-binomial residual sits near unit variance,
and the embedding residuals---the low-rank factorization and the two neural latent models---are
heavy-tailed and anti-conservative. This is residual form, distinct from the overdispersion
calibration of Figure~\ref{fig:E08971}.}
\label{fig:E08971-form}
\end{figure}

\subsection{The concentration is a one-parameter regularizer, and a single scalar suffices}
\label{sec:eval-E01249}

Because $\alpha_0$ is a single scalar, it need not be fit by maximum likelihood at all: one can
sweep it across a spectrum and watch the residual geometry change, exactly as one would vary a
regularization strength. Scanning $\alpha_0$ on a logarithmic grid spanning two decades on either
side of the moment estimate $\hat\alpha_0$ makes this concrete and bounds what the choice costs.
Three properties hold together on the cached peripheral-blood mononuclear data and within the broad
cell-type strata of the cross-tissue atlas \cite{Eraslan2022}.

First, the held-out conditional likelihood is unimodal in $\log\alpha_0$ with an interior optimum,
and that optimum sits within about six tenths of a decade of the moment estimate
(Figure~\ref{fig:E01249}a). The objective is therefore informative enough to locate a sensible
concentration, but its curvature is gentle---the spectrum spans only about a third of a nat per
count on the blood data and half a nat per count within a stratum---so the precise value is not
critical, and the cheap moment estimate lands near the likelihood optimum without an explicit fit.

Second, the downstream geometry is far more forgiving than the likelihood. The accuracy of a
nearest-neighbor label transfer computed on a low-dimensional embedding of the residual matrix is
flat to within about one part in a hundred across the central two decades
$[\hat\alpha_0/10,\,\hat\alpha_0\times 10]$ (Figure~\ref{fig:E01249}b). Two orders of magnitude of
$\alpha_0$ leave the embedding's class structure essentially unchanged: the concentration is a
regularization-strength dial whose useful range is wide, not a knife-edge that must be tuned.

Third, the way the dial acts is exactly the closed-form prediction. Singletons are untouched: the
entry-wise gap $\lvert d^{\mathrm{Mult}}\rvert-\lvert d^{\mathrm{DM}}\rvert$ is zero to machine
precision at every count of one across the whole grid (Proposition~\ref{prop:x1}). The gap is
nonnegative on every repeated count (Proposition~\ref{prop:boundedbymult}) and grows as $\alpha_0$
falls and as the count rises, vanishing in the multinomial limit
(Proposition~\ref{prop:multlimit}); the heatmap of mean shrinkage by count bin and $\log\alpha_0$
(Figure~\ref{fig:E01249}c) shows precisely this monotone structure, with the singleton row pinned at
zero.

The spectrum view raises the natural question of whether one global scalar is enough, given that
overdispersion is heterogeneous across features. We answer it directly by replacing the single
scalar with two richer alternatives and re-measuring held-out likelihood and calibration across the
twelve largest strata (Figure~\ref{fig:E01249}d). Allowing each cell type its own concentration
changes the held-out likelihood per count negligibly relative to one pooled scalar
($+0.003$, paired $95\%$ confidence interval $[-0.001,\,+0.007]$, straddling zero). Allowing each
\emph{feature} its own moment-fitted concentration is, if anything, harmful: it lowers the held-out
likelihood per count by $0.40$ nats (interval $[-0.68,\,-0.24]$), because each per-gene estimate is
formed from a few hundred to about a thousand cells and overfits its training overdispersion rather
than generalizing. It buys no calibration either---the false-positive rate of the null-split log-mean
test is $0.047$ under both the single and the per-feature null, against a nominal $0.05$. The single
scalar is thus not a lossy stand-in for a richer per-feature model; on held-out data it is the
better-generalizing choice, and the simplification the method makes costs nothing it does not repay.

\begin{figure}[t]
\centering
\includegraphics[width=\linewidth]{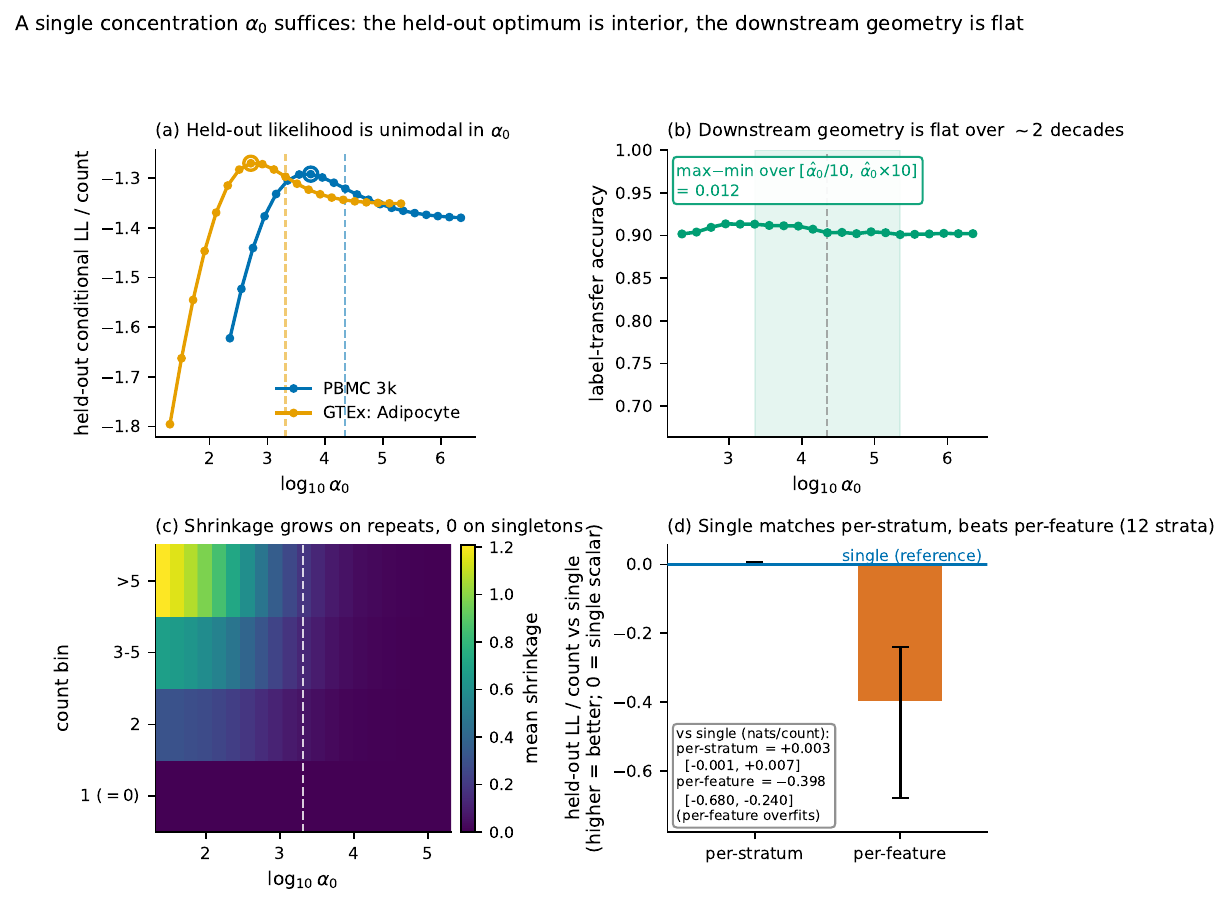}
\caption{The concentration $\alpha_0$ behaves as a one-parameter regularizer, and a single scalar
suffices. \emph{(a)} Held-out conditional Dirichlet--multinomial log-likelihood per count versus
$\log_{10}\alpha_0$ on the blood data and within a representative atlas stratum; the curve is
unimodal with an interior optimum (ring) close to the moment estimate $\hat\alpha_0$ (dashed).
\emph{(b)} Label-transfer accuracy of a nearest-neighbor classifier on a low-dimensional embedding
of the residual matrix is flat to within $0.012$ over the shaded central two decades
$[\hat\alpha_0/10,\,\hat\alpha_0\times 10]$ on real blood clusters---downstream geometry is
insensitive to $\alpha_0$ across two orders of magnitude. \emph{(c)} Mean entry-wise shrinkage
$\lvert d^{\mathrm{Mult}}\rvert-\lvert d^{\mathrm{DM}}\rvert$ by count bin and $\log_{10}\alpha_0$
(atlas stratum): the singleton row is identically zero (Proposition~\ref{prop:x1}); shrinkage grows
toward smaller $\alpha_0$ and larger counts and vanishes in the multinomial limit
(Propositions~\ref{prop:boundedbymult},~\ref{prop:multlimit}). \emph{(d)} Change in held-out
conditional log-likelihood per count, relative to the single pooled scalar (the zero reference
line), for a per-stratum scalar and a per-feature concentration across the twelve largest strata
(paired, bias-corrected accelerated $95\%$ intervals). A per-stratum scalar lands on the single
scalar ($+0.003$ nats/count, interval $[-0.001,+0.007]$, straddling zero), while a per-feature
concentration falls well below it ($-0.398$ nats/count, interval $[-0.680,-0.240]$) because it
overfits the noisy per-gene estimates; the underlying absolute likelihoods sit near $-1.3$ nats per
count (panel a). The single scalar therefore matches the per-stratum value and exceeds the
per-feature concentration.}
\label{fig:E01249}
\end{figure}

\subsection{An exact evidence unit for the DM contrast}
\label{sec:eval-E01250}

The testing interpretation of the DM contrast is only actionable if it comes with a calibrated way to
turn a contrast into evidence. There is one, and it is exact. Splitting a homogeneous population, we
fit the concentration on one part and, on the held-out part, score each row by the likelihood ratio
of the fitted Dirichlet--multinomial against the null. Because a Dirichlet--multinomial probability
mass function integrates to one at any fixed parameter, the held-out ratio has expectation exactly one
under the null, independent of the alphabet size, the row total, or the concentration: across a grid
of sizes and concentrations the measured expectation departs from one by at most $2.8\times10^{-15}$
(Figure~\ref{fig:E01250}(e)). That is the defining property of an $e$-variable, and it makes the
split construction an exact, anytime-valid unit of evidence---accumulable across rows and across a
growing screen, and stoppable at any time---without any appeal to an asymptotic null. The naive
alternative of fitting the concentration and scoring the ratio on the \emph{same} data is not an
$e$-variable: its expectation exceeds one at every concentration and inflates without bound as the
data are reused, by a median factor near $10^{45}$ in the regimes tested, and grows toward the
multinomial limit. Reusing the fit as if it were free is the failure mode the split exists to prevent.

The exactness matters because the obvious shortcut---reading the aggregated deviance against a
$\chi^2$ reference---does not calibrate on the data this paper is about. Under a correctly specified
null, the row-aggregated DM deviance is anti-conservative relative to $\chi^2_{K-1}$ on sparse counts
(Kolmogorov--Smirnov statistic $0.20$ against the asymptotic reference; Figure~\ref{fig:E01250}(a)),
and the column-aggregated deviance departs more strongly still, because the saturated-against-
null deviance approaches its $\chi^2$ limit only when the per-cell counts are large, and single-cell
counts are mostly ones and zeros. The miscalibration sharpens, as expected, in the small-library
stratum and relaxes where row totals are large. A multinomial null shows the same asymptotic
fragility. The reading is not that the DM deviance is poorly behaved but that its
\emph{asymptotic} calibration is the wrong tool here: the per-row, finite-sample, exactly-calibrated
$e$-variable is what licenses reporting DM evidence and monitoring it as a screen streams, and it is
available in closed form at the same per-nonzero cost as the residual.

\begin{figure}[t]
\centering
\includegraphics[width=\linewidth]{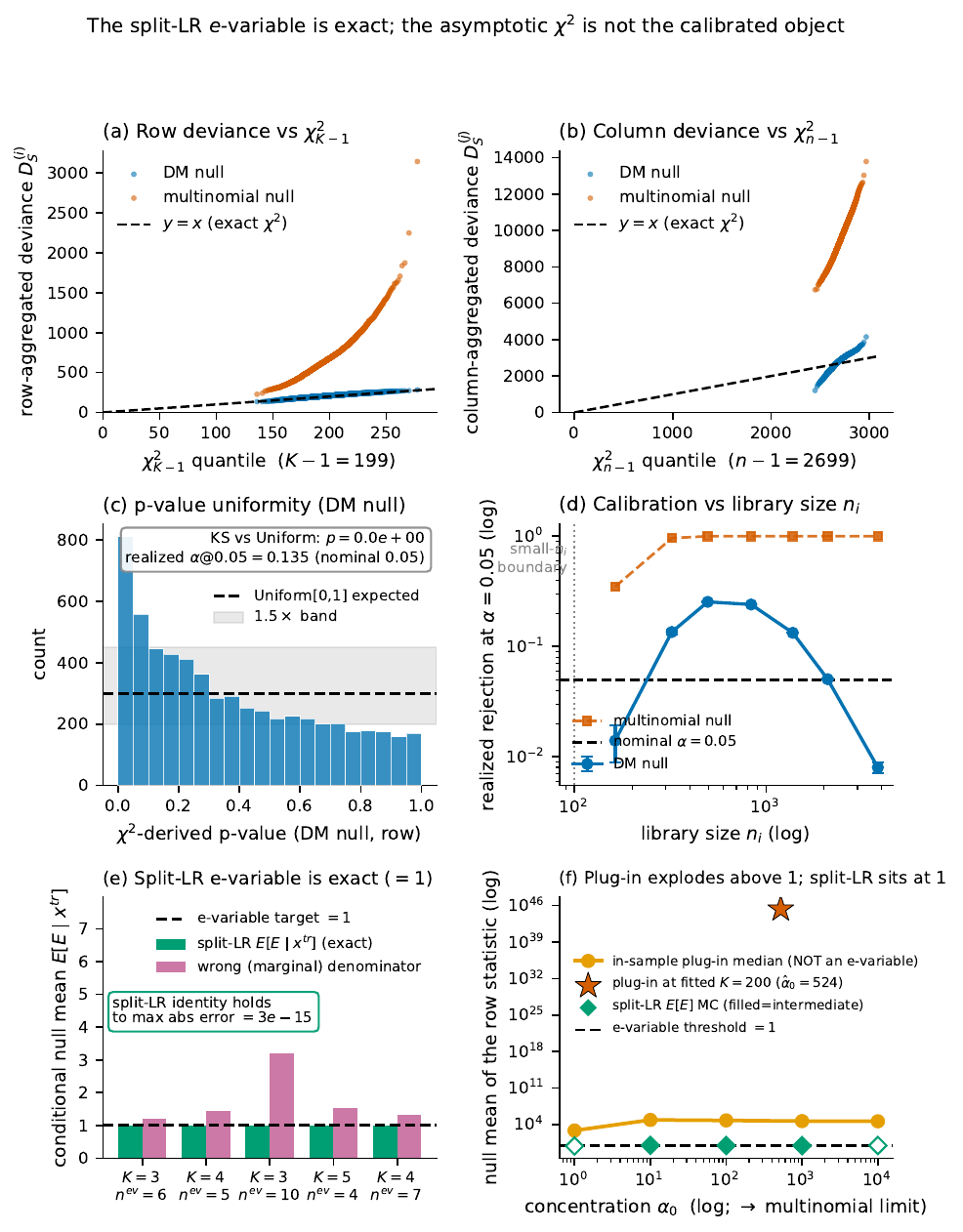}
\caption{An exact evidence unit for the Dirichlet--multinomial contrast. The null is fitted on a
peripheral-blood single-cell split ($K=200$ genes, $n=2700$ cells, moment concentration
$\hat\alpha_0=524$) and the calibration panels draw $N_{\mathrm{sim}}=1000$ synthetic null matrices
at the empirical row totals. \emph{(a,b)} Under a
correctly specified null, the row- and column-aggregated DM deviance against the asymptotic
$\chi^2_{K-1}$ and $\chi^2_{n-1}$ references: the Dirichlet--multinomial null tracks the $\chi^2$ line
far more closely than the multinomial null but is itself anti-conservative on sparse counts.
\emph{(c)} The $\chi^2$-derived $p$-values under the DM null are not uniform, so the asymptotic
calibration is only a rough guide. \emph{(d)} The realized rejection rate at nominal $\alpha=0.05$
across library sizes is anti-conservative for both nulls, worst in the small-library stratum, and far
closer to nominal for the DM null. \emph{(e)} The split-likelihood-ratio $e$-variable has conditional
null expectation one to machine precision, whereas the wrong (marginal) denominator does not.
\emph{(f)} Across concentrations the in-sample plug-in ratio inflates by many orders of magnitude
(log scale) and grows toward the multinomial limit while the split-LR $e$-variable sits at one---an
exact, anytime-valid evidence unit on the one hand, an invalid one on the other.}
\label{fig:E01250}
\end{figure}

\subsection{Fitting the DM to data it cannot represent: graceful degradation and do no harm}
\label{sec:eval-E01253}

The spectrum of Section~\ref{sec:spectrum} has one limitation by construction: the flat DM ties all
extra-multinomial variation to a single concentration and cannot represent arbitrary covariance on
log-ratios, where a latent-Gaussian model can \cite{AitchisonShen1980,AitchisonHo1989,XiaChenFungLi2013}.
The practical question is what this costs when the data really is generated by such a model, or by
some other mechanism outside the DM family. We answer it by generating from six alternatives the DM
cannot represent exactly---isotropic, low-rank, and full-rank logistic-normal multinomial models, a
Poisson-lognormal model, a zero-inflated DM, and a two-component composition mixture---fitting the
flat DM with its single estimated concentration, and asking three questions of the result, each
evaluated against a tractable reference rather than the latent model's intractable marginal
likelihood.

The first is whether the flat DM still earns its keep over the multinomial. On a held-out split, the
conditional log-likelihood of the flat DM is at least that of the multinomial on every one of the six
mechanisms and at both scales, without exception across all replicates, and the margin grows as the
generating mechanism departs further from a multinomial: it is a few hundredths of a nat per count
when a logistic-normal latent is weak and the data is nearly multinomial, and well over a nat per
count when the mechanism injects strong joint overdispersion, as the zero-inflated and mixture
generators do (Figure~\ref{fig:E01253}, left). The single concentration parameter keeps absorbing the
joint overdispersion the multinomial ignores, even when the true dependence has a structure the DM
cannot name.

The second question is how much of the true residual structure the flat DM recovers. Because the
counts are generated from a known per-cell composition, the exact residual the oracle would leave is
available in closed form, and we correlate the flat-DM residual with it. Here the matched-variance
design is essential: scaling each latent covariance to the same total variance separates the
\emph{rank} of the dependence from its overall \emph{magnitude}. The correlation falls smoothly as the
magnitude of latent dependence grows, but is almost unchanged across the isotropic, low-rank, and
full-rank covariances at fixed magnitude---the three curves lie on top of one another
(Figure~\ref{fig:E01253}, middle). The flat DM is in this sense blind to the rank of the latent
dependence and sensitive only to its scale, which is the precise content of tying all extra variation
to one parameter: the model degrades with how much joint structure is present, not with how that
structure is arranged.

The third question is whether the misfit does any harm. A natural worry is that fitting an
overdispersed model to data outside its family inflates downstream false positives. It does not. Under
a null random split the parametric two-group test calibrated by the flat-DM variance stays far below
the multinomial, whose Poisson-family variance is badly anti-conservative under every one of these
mechanisms---its false-positive rate runs to $0.6$--$0.8$ against a nominal $0.05$ while the flat DM
remains near the nominal level (Figure~\ref{fig:E01253}, right). The same holds on real data: on the
lowest-variance, housekeeping-like genes of a peripheral-blood reference, where there is no true
differential expression to find, the flat DM's null-split false-positive rate is essentially nominal
and below the multinomial's, and its residuals carry no spurious coupling to sequencing depth. The
caveat is that under the strongest latent dependence the flat DM is itself
anti-conservative rather than exactly calibrated, its worst-case false-positive rate reaching roughly
$0.17$ (about three times nominal); it is far better than the multinomial in every case,
but the single concentration is an approximation, not a cure. Taken together the three readouts say
the tied-concentration assumption degrades gracefully and safely: where the extra flexibility of a
latent-Gaussian model is genuinely needed, that need shows up as a measurable held-out gap, which is
exactly the basis on which a richer model should be preferred.

\begin{figure}[t]
\centering
\includegraphics[width=\linewidth]{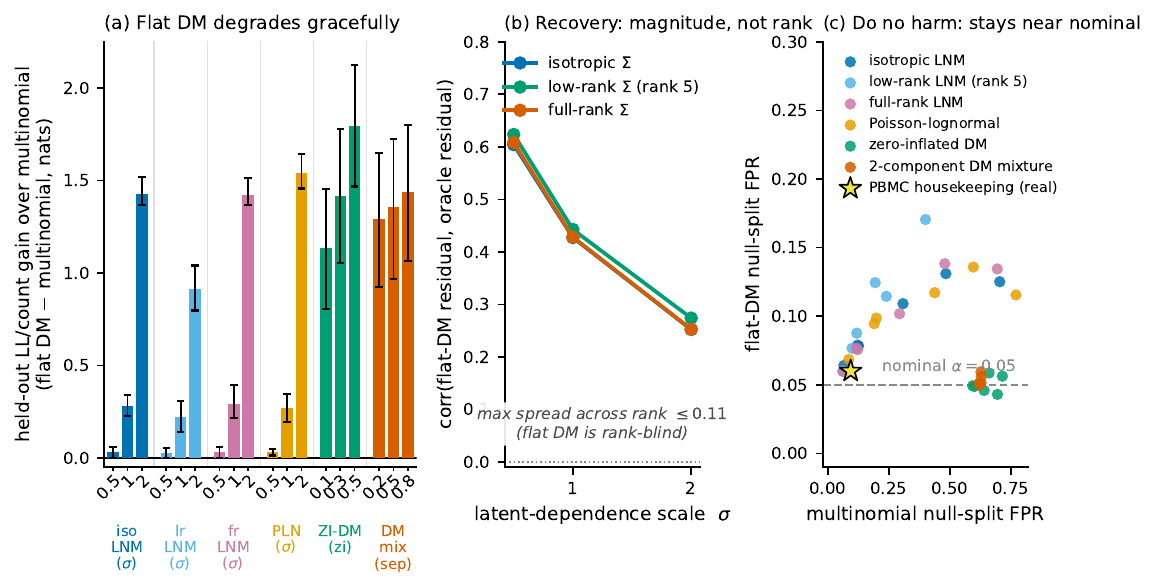}
\caption{Fitting the flat Dirichlet--multinomial to data it cannot represent: it degrades gracefully
and does no harm. \emph{Left:} held-out conditional log-likelihood per count of the flat DM minus the
multinomial, for each non-DM generating mechanism and severity (pooled over both scales); the gain is
non-negative throughout and widens as the mechanism departs further from a multinomial (95\%
intervals). \emph{Middle:} correlation between the flat-DM residual and the exact oracle residual
built from the true per-cell composition, as a function of the latent-dependence scale $\sigma$, for
the isotropic, low-rank, and full-rank logistic-normal covariances held to the same total variance;
the correlation falls with the magnitude of dependence, and the three ranks coincide---the flat DM is
blind to the rank of the latent covariance and sensitive only to its scale. \emph{Right:} null-split
false-positive rate of the parametric two-group test under the flat-DM null variance against the
multinomial null variance, one point per mechanism, severity, and scale, plus the real housekeeping-
gene control ($\star$); every point lies well below the diagonal, the multinomial is anti-conservative
($0.6$--$0.8$) where the flat DM stays near the nominal $0.05$, and the real control manufactures no
spurious differential expression.}
\label{fig:E01253}
\end{figure}

\subsection{The ordered model pays off only when the order is real}
\label{sec:eval-E06033}

The generalized Dirichlet--multinomial of Section~\ref{sec:gdm} buys stage-specific dispersion at a
price: unlike the flat Dirichlet--multinomial, it depends on the chosen feature order, and the
text's guidance is to adopt it only when that order is part of the scientific design rather than
an arbitrary column ordering. That qualifier is testable. Fixing the per-stage concentrations
$\nu_j$ as in \eqref{eq:gdm_stagecell} and scoring held-out per-observed-count conditional
log-likelihood, we ask whether the generalized model's gain survives a random permutation of the
feature order. The parameter count is identical between the true-order model and every permuted
one---both carry $K-1$ stage dispersions---so the difference in held-out likelihood between the
true and a permuted order isolates the contribution of the order itself, separate from the extra
flexibility.

On synthetic counts generated with a concentration that genuinely increases along the order---early
splits strongly overdispersed, later splits near-binomial---the true order improves held-out
conditional log-likelihood over a random permutation by $+0.0024$ nats per observed count (bootstrap
$95\%$ interval $[+0.0022,+0.0027]$; positive in all eight matched train/test splits, sign-test
$p=0.008$; a permutation $p$-value of $0.04$), and the fitted per-stage concentration tracks the
generating profile closely (correlation $0.996$ across the stages whose count budget remains large
enough to identify a dispersion; Figure~\ref{fig:E06033}, right). The companion control settles what the
gain is made of: drawing counts from an exchangeable Dirichlet--multinomial, where no feature
position is special, the same true-versus-permuted gain falls to $0.26\%$ of its stage-varying value
(Figure~\ref{fig:E06033}, left). When there is order structure the model exploits it; when there is none,
permuting costs nothing. The two arms together are the falsification the qualifier needs: the
improvement is the order, not the additional dispersion parameters, which would have helped on a
permuted order just as well. A structural check accompanies the statistical one---scored with the
same fitted concentrations, the comb-tree Dirichlet-tree-multinomial held-out log-likelihood matches
the generalized-DM stage cascade to within $10^{-8}$, confirming at the dataset level that the
ordered model is exactly the comb-tree case of the hierarchy of Section~\ref{sec:dtm}.

The real-data picture is consistent but deliberately bounded. Genomic coordinates---the natural
linear order---are not retained in this assay's processed per-cell-type count blocks, so we order
genes instead by an abundance proxy (mean expression), under which neighbouring features share an
expression regime and may share dispersion. Across eight broad cell types of the cross-tissue atlas
\cite{Eraslan2022}, the true (abundance) order improves held-out conditional log-likelihood over a
random permutation in five of the eight, with a median per-count gain of $+0.0011$; in the remaining
three the permuted order is marginally preferred. This is a genuine but partial signal, smaller and
noisier than a purpose-built order, and it should be read as bounding rather than establishing the
generalized model's reach: an abundance ordering carries some exploitable stage-varying dispersion in
most populations but not all, and a feature set that is merely orderable is not thereby an order the
model can use. The practical reading is the one the text advances---reach for the ordered model when
the order encodes real sequential structure, and default to the permutation-invariant flat model
otherwise.

\begin{figure}[t]
\centering
\includegraphics[width=\linewidth]{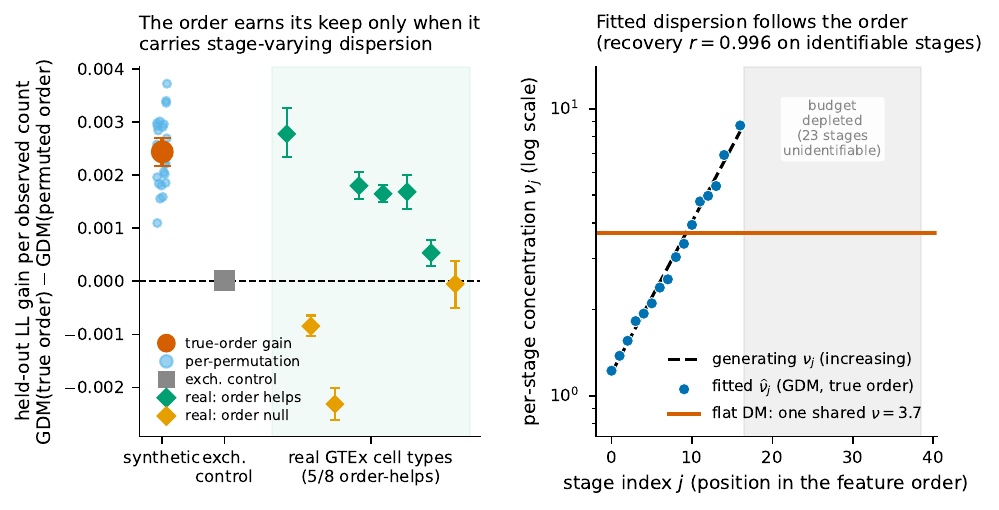}
\caption{The generalized Dirichlet--multinomial earns its order only when the order carries
stage-varying dispersion. \emph{Left:} the order-attributable held-out log-likelihood gain per
observed count, $\mathrm{GDM}(\text{true order})-\mathrm{GDM}(\text{permuted order})$, with the
parameter count held fixed. On synthetic counts whose stage dispersion increases along the order, the
true-order gain (filled circle, $95\%$ bootstrap interval) sits well above the spread of the
per-permutation gains and above zero; on an exchangeable control with no order structure the gain
collapses onto zero (square). Real cross-tissue cell types (diamonds, abundance-proxy order) scatter
about a small positive median---five of eight with the order helping (above zero), three with the
order immaterial. \emph{Right:} the fitted per-stage concentration $\widehat\nu_j$ on the synthetic
true order (points) recovers the increasing generating profile (dashed), against the single shared
concentration the flat model is forced to use (horizontal line); stages whose count budget depletes
are unidentifiable and omitted.}
\label{fig:E06033}
\end{figure}

\subsection{The transform is linear in the nonzeros and keeps them sparse}
\label{sec:eval-E01326}

The closed-form residual is built to drop into a sparse pipeline: each nonzero entry costs one
log-gamma evaluation \eqref{eq:lgammaresid}, the row totals and the fitted composition are computed
in a single pass, and the zeros are never materialized. Those properties are what make the transform
practical at atlas scale, and they are measurable rather than asserted. On synthetic
Dirichlet--multinomial matrices spanning roughly three decades of nonzero count, we fit the empirical
scaling exponent of wall-clock time against the number of nonzeros, $T = c\,\mathrm{nnz}^{\beta}$, by
bootstrap on log-log axes. The transform's exponent is $1.00$ at the median across the five sweeps
(Figure~\ref{fig:E01326}, left), tracking the slope of the pointwise \texttt{log1p} baseline---the
fastest possible nonzero-only normalization---rather than the super-linear slope a dense method would
trace. The fitted concentration likewise scales at most linearly (median exponent below one over the
measured range, where the fixed startup cost still contributes). Timing jitter and cache effects move
individual per-sweep exponents a few percent either side of one; two of the synthetic sweeps read as
marginally super-linear (bootstrap intervals just above one), and they are the smaller-alphabet ones
rather than the largest, a memory-bandwidth and jitter artifact rather than algorithmic growth---the
largest sweep is itself indistinguishable from linear and the per-element throughput is flat, since
the per-element work is constant by construction.

That constancy is direct: the log-gamma evaluation runs at $\approx 5.7$ nanoseconds per nonzero,
flat across a two-hundred-fold range of input sizes. The price relative to \texttt{log1p} is a
constant factor of about fourteen---the cost of a special function over a logarithm---paid once per
nonzero and never growing with the matrix, so the transform inherits \texttt{log1p}'s linear scaling
and, exactly, its sparsity: across one hundred sixty-eight checked configurations the output carries
the same number of nonzeros as the input, with no fill-in. Fitting the concentration is the more
expensive step, two to six times the cost of applying the transform, so the intended usage of fitting
once and transforming many times is the economical one. The contrast with a dense alternative is
categorical rather than constant-factor: a centered-log-ratio normalization, which must form the
dense logarithm of the matrix, exhausts the memory budget on the large cross-tissue matrices (up to
$5.7\times10^{8}$ dense entries) on which the DM transform completes inside the sparse nonzero set.
The computational story the paper tells is therefore the one the measurements show---linear,
sparsity-preserving, and constant-time per nonzero---with the qualifier that ``lightweight''
means a fixed multiple of the cheapest possible pass, not parity with it.

\begin{figure}[t]
\centering
\includegraphics[width=\linewidth]{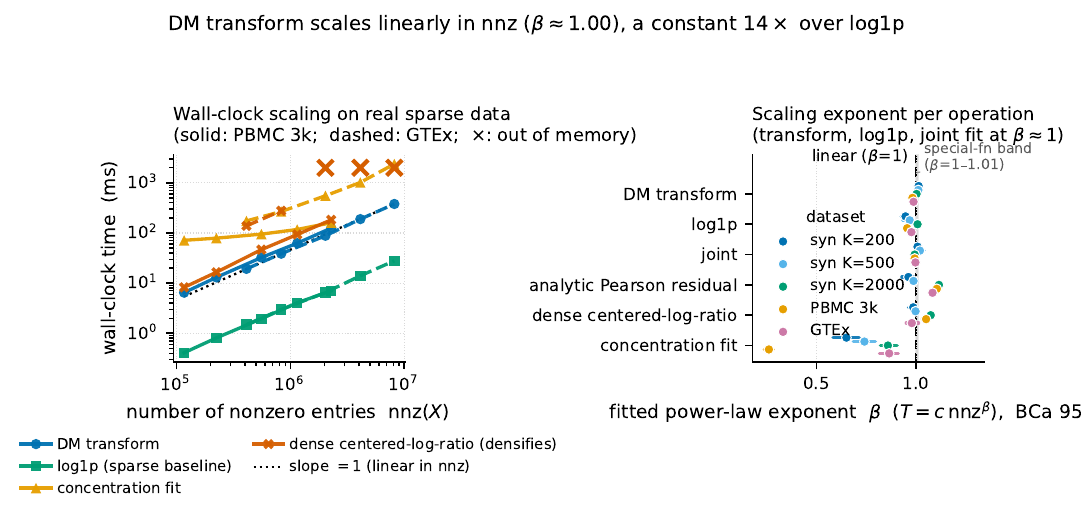}
\caption{The Dirichlet--multinomial residual transform is linear in the number of nonzeros and
preserves the sparsity pattern. \emph{Left:} wall-clock time against nonzero count on log-log axes for
the transform, the one-time concentration fit, and the pointwise \texttt{log1p} reference, across
synthetic sweeps and two real matrices; the transform's fitted exponent is one at the median (median
$\beta=1.00$; slopes annotated with bootstrap intervals), parallel to \texttt{log1p} and a constant
factor of about fourteen above it---the price of a special-function evaluation over a logarithm, paid
once per nonzero. The dense centered-log-ratio baseline cannot be drawn at the large sizes, where it
exhausts memory.
\emph{Right:} processed nonzeros per second is flat across a two-hundred-fold size range (constant
time per nonzero), and the output-versus-input nonzero count lies on the identity line (exact sparsity
preservation).}
\label{fig:E01326}
\end{figure}

\akshay{omitted-eval E08964: sister-paper experiment on a combinatorial CRISPR genetic-interaction map; it uses this paper's Dirichlet-multinomial overdispersion model, but the screen-design application is the subject of a separate paper not yet public}%
\akshay{omitted-eval E08965: sister-paper experiment showing the Dirichlet-multinomial concentration calibrates a real interaction test in a CRISPR screen; the screen application is out of scope here and the companion paper is not yet public}%

\begin{table}[t]
\centering
\small
\setlength{\tabcolsep}{6pt}
\renewcommand{\arraystretch}{1.15}
\begin{tabular}{@{}lr@{}}
\toprule
Method & Held-out LL (per observed count) \\
\midrule
DM, jointly fitted $(\alpha_0,\bm\pi)$   & $\mathbf{-0.302}$ \\
DM, fixed $\bm\pi = \hat{\bm\pi}_{\text{train}}$ & $-0.307$ \\
Multinomial                              & $-0.976$ \\
NB (feature-wise)                        & $-0.762$ \\
\bottomrule
\end{tabular}
\caption{Regime~3 (misspecified dependence): 2-component DM mixture.
Held-out per-observed-count conditional LL on synthetic data generated
from a mixture of two DM distributions with different base compositions.
Both fixed-$\pi$ DM and joint DM dominate the multinomial and NB
baselines; joint DM adds a small additional gain at the cost of $K-1$
extra parameters.}
\label{tab:regime3}
\end{table}

\begin{table}[t]
\centering
\small
\setlength{\tabcolsep}{6pt}
\renewcommand{\arraystretch}{1.15}
\begin{tabular}{@{}lr@{}}
\toprule
Check & Max $|$difference$|$ per sample \\
\midrule
\multicolumn{2}{@{}l}{\emph{DM conditional LL (DirMult) vs NB conditional LL (scipy.stats.nbinom)}} \\
\quad $p = 0.3$  & $1.99 \times 10^{-12}$ \\
\quad $p = 0.5$  & $1.88 \times 10^{-12}$ \\
\quad $p = 0.7$  & $3.07 \times 10^{-12}$ \\
$p$-invariance: max $|$NB cond.$(p_i)$ $-$ NB cond.$(p_j)|$  & $2.16 \times 10^{-12}$ \\
\bottomrule
\end{tabular}
\caption{Regime~2: \emph{independent} numerical verification of
Corollary~\ref{cor:nbdmdeviance}.
The DM conditional log-likelihood is computed
from the DirMult formula (our implementation) and compared, for three
values of the NB success probability $p$, to the independent-NB
conditional log-likelihood computed via
\texttt{scipy.stats.nbinom.logpmf}$\,(x, r_j, p)$ minus
\texttt{scipy.stats.nbinom.logpmf}$\,(n_i, \sum_j r_j, p)$.
Two
predictions follow from the theorem: (i) the NB conditional LL must be
independent of $p$ (row-sum conditioning cancels it), and (ii) it must
equal the DirMult LL at every $p$.
Both are satisfied to
double-precision floating-point error (${\sim}10^{-12}$) across the
$p$ grid, using an implementation path (scipy's NB PMF) that shares no
code with our DirMult formula.}
\label{tab:regime2}
\end{table}

\section{Background on exponential families}

\label{sec:expfam_background}

This appendix records the exponential-family viewpoint underlying the main text. 
Readers who mainly care about the residual formulas may skip directly to Appendix~\ref{app:otherresiduals}.

Exponential families arise whenever one models data through a fixed collection of sufficient statistics. 
A convenient general formulation is as follows. 
Let $\mathcal{X}$ be a sample space, let $f_1,\dots,f_d$ be measurable feature functions on $\mathcal{X}$, and suppose we only observe their expected values
\[
\E_{\hat P_n}[f_k(X)] = M_k,
\qquad k=1,\dots,d
\]
These moment constraints define the affine set
\[
\mathcal{A}
:=
\left\{
Q \in \Delta(\mathcal{X}) :
\E_Q[f_k(X)] = M_k \;\; \forall k
\right\}
\]
Given a prior distribution $P$, the corresponding information projection solves
\[
P_{\mathcal{A}}^\star
=
\argmin_{Q \in \mathcal{A}} D(Q\Vert P)
\]
Under standard regularity conditions, this projection has the familiar exponential-family form
\[
P_{\mathcal{A}}^\star(x)
\propto
P(x)\exp\paren{\sum_{k=1}^d \lambda_k^\star f_k(x)}
\]
Equivalently,
\[
P_{\mathcal{A}}^\star(x)
=
P(x)\exp\paren{\sum_{k=1}^d \lambda_k^\star f_k(x)-A(\bm\lambda^\star)}
\]
where
\[
A(\bm\lambda)
=
\log \E_P\exp\paren{\sum_{k=1}^d \lambda_k f_k(X)}
\]
is the log-partition function \cite{FollmerSchied2011,efron_2022}.

\subsection{Learning with exponential families}

This viewpoint is useful for at least four reasons.

\begin{enumerate}[1.]
\item \textbf{Loss optimality.} The information projection is the robust predictor that minimizes worst-case log loss over the constraint set:
\[
P_{\mathcal{A}}^\star
=
\argmin_{Q \in \Delta(\mathcal{X})}
\max_{P' \in \mathcal{A}}
H(P',Q)
\]
Thus, if the only trustworthy information about the data is encoded by the observed feature means, the exponential-family projection is the canonical likelihood-based predictor.

\item \textbf{Maximum probability / likelihood.} The same distribution is also the one that most plausibly explains the observations under the prior $P$. 
This is the information-theoretic underpinning of maximum-likelihood estimation in exponential families.

\item \textbf{Sufficiency.} The chosen statistics are not arbitrary summaries. 
For families with finite-dimensional sufficient statistics whose dimension does not grow with sample size, the Pitman--Koopman--Darmois theorem shows that one is exactly in the parametric exponential-family setting \cite{fisher1922mathematical,koopman1936distributions,Darmois1935,pitman1936sufficient}. 
This is one of the main reasons exponential families occupy such a central role in statistics.

\item \textbf{Convenience.} Likelihoods are convex in the natural parameters, estimation problems are well behaved, and discrepancies decompose naturally in terms of the sufficient statistics. 
These are precisely the properties exploited by deviance-based normalization.
\end{enumerate}

\subsection{Parameters, support, and representation}

The natural parameters are often not the same as the most familiar textbook parameters. 
For count distributions this matters in two ways.

First, the support must be fixed correctly. 
For example, the multinomial is only an exponential family once the total count $n$ is treated as fixed, because $n$ determines the support. 
The same issue appears throughout this paper: conditioning on the observed sample total $n_i$ is not only statistically natural, it is also the right way to view the resulting joint count models as exponential-family objects.

Second, the DM is best thought of as a \emph{curved} exponential family for fixed $n$, not as a natural exponential family in the narrow sense of \cite{Morris1982}. 
This is the source of some confusion in the literature. 
A natural exponential family uses the observation itself as sufficient statistic, whereas the DM uses a richer indicator representation of the count histogram \cite{Elkan2006}. 
Nothing in the main text depends on the DM being a natural exponential family; what matters is that it still inherits the same likelihood and deviance machinery.

\subsection{Inference and learning}

For an exponential family with density
\[
p_{\bm\lambda}(x)
=
h(x)\exp\paren{\bm\lambda^\top T(x)-A(\bm\lambda)}
\]
the score equation is
\[
\nabla_{\bm\lambda}\log p_{\bm\lambda}(x)
=
T(x)-\nabla A(\bm\lambda)
\]
At the maximum-likelihood estimate, the fitted sufficient statistics match the observed ones. 
The Hessian is
\[
-\nabla^2_{\bm\lambda}\log p_{\bm\lambda}(x)
=
\nabla^2 A(\bm\lambda)
\]
the Fisher information, which is the covariance of the sufficient statistic under $p_{\bm\lambda}$. 
These standard identities explain the stability of likelihood fitting and the asymptotic normality of the resulting estimators \cite{MacCullaghNelder1989,efron_2022}.

In the main text, the tied null family $\bm\alpha=\alpha_0\bm\pi$ is especially simple because once $\bm\pi$ is fixed, only one scalar parameter remains. 
The one-dimensional score and Hessian in Section~\ref{sec:alpha0} are therefore the direct analogue of classical likelihood fitting in a generalized linear model, specialized to the DM concentration parameter.

\subsection{Deviance and likelihood geometry}

Suppose $P_{\bm\lambda_1}$ and $P_{\bm\lambda_2}$ are two members of the same exponential family. 
Then the log-loss difference between using $\bm\lambda_2$ rather than $\bm\lambda_1$ is governed by the corresponding Kullback--Leibler divergence. 
This is the origin of deviance. 
Up to the conventional factor of $2$, deviance is simply the log-likelihood gap between a fitted model and a benchmark model, often the saturated model \cite{Hastie1987,MacCullaghNelder1989,efron_2022}.

Another useful identity is the cross-entropy decomposition
\[
H(P,P_{\bm\lambda})
=
D(P\Vert P_{\bm\lambda}) + H(P)
\]
From the perspective of fitting the model, minimizing cross-entropy and minimizing KL divergence from the data are therefore the same task. 
This is why deviance-based losses are the natural replacement for squared loss when the sampling model is non-Gaussian.

For feature-wise count models, the resulting per-observation deviance contributions are nonnegative. 
For joint constrained models such as the multinomial or DM, sample-wise deviance still decomposes additively, but the entry-wise additive terms can be signed. 
That is why the main text adopts a signed square-root transform of the additive likelihood terms rather than insisting on an exact nonnegative entry-wise decomposition.

\subsection{Count models as exponential families}

\begin{table}[ht]
\centering
\small
\setlength{\tabcolsep}{4pt}
\resizebox{\textwidth}{!}{%
\begin{tabular}{@{}lllll@{}}
\toprule
Distribution & Support & Parameters & Sufficient statistics & Comment \\
\midrule
Multinomial &
$\bm x \in \N^K,\ \sum_j x_j=n$ &
$\bm p \in \Delta^{K-1}$ &
$\bm x$ &
Natural joint model of a fixed total \\

Dirichlet &
$\bm x \in (0,1)^K,\ \sum_j x_j=1$ &
$\bm\alpha > 0$ &
$\log \bm x$ &
Continuous compositional model \\

Dirichlet--multinomial &
$\bm x \in \N^K,\ \sum_j x_j=n$ &
$\bm\alpha > 0$ &
$\{\ind(m < x_k)\}_{m,k}$ &
Symmetric overdispersed fixed-total model \\

Generalized Dirichlet--multinomial &
$\bm x \in \N^K,\ \sum_j x_j=n$ &
$(\alpha_j,\beta_j)_{j=1}^{K-1}>0$ &
$\{\ind(m<x_j),\ind(m<R_j),\ind(m<R_{j-1})\}_{j,m}$ &
Ordered stage-wise extension \\

Dirichlet-tree multinomial &
$\bm x \in \N^K,\ \sum_j x_j=n$ and tree $T$ &
$(\bm\alpha_\nu)_{\nu\in\mathcal I}>0$ &
$\{\ind(m<x_{\nu c})\}_{\nu,c,m}$ &
Tree-structured multiscale extension \\
\bottomrule
\end{tabular}%
}
\caption{A schematic summary of the count models most relevant to this paper. 
For the DM, a convenient sufficient-statistic representation uses the indicator collection $\{\ind(m<x_k)\}$ for $m=0,\dots,n-1$ and $k=1,\dots,K$. 
For the generalized DM, $R_j:=\sum_{\ell=j+1}^K x_\ell$ denotes the remainder count after stage $j$. 
For the DTM, $x_{\nu c}$ denotes the count flowing from internal node $\nu$ into child subtree $c$.}
\label{tab:count_expfams}
\end{table}

\section{Deviance-style residual transforms for multivariate counts}
\label{app:otherresiduals}

This appendix collects the residual transforms discussed in the main text. 
In all cases,
\[
n_i := \sum_{j=1}^K X_{ij}
\qquad
\pi_j := \frac{\sum_{i=1}^n X_{ij}}{\sum_{i=1}^n \sum_{k=1}^K X_{ik}}
\]

\paragraph{A family resemblance among count-model residuals.}
Placed side by side, the standard residual formulas are strikingly similar.
They all compare the same observation to the same null mean or null composition and then take a signed square root of the corresponding log-likelihood contrast.
The term $X_{ij}\log(X_{ij}/(n_i\pi_j))$ is already the backbone of the binomial, Poisson, negative-binomial, and multinomial formulas; the Poisson and NB add the usual mass and dispersion corrections, the dense Dirichlet replaces raw counts by positive pseudocount-augmented compositions, and the DM replaces the single logarithm by a threshold sum that discounts repeated counts according to $\alpha_0$.
The limiting relations $r_j\to\infty$ and $\alpha_0\to\infty$ make this especially transparent.
Read this way, the DM is not an isolated construction, but the joint overdispersed member of the same deviance-style family. 

\begin{table}[p]
\centering
\small
\setlength{\tabcolsep}{4pt}
\renewcommand{\arraystretch}{1.3}
\begin{tabular}{@{}p{0.22\textwidth}p{0.71\textwidth}@{}}
\toprule
Distribution & Signed residual transform \\
\midrule
Binomial &
$\displaystyle
\sgn(X_{ij}-n_i\pi_j)
\sqrt{
\abs{
X_{ij}\log\frac{X_{ij}}{n_i\pi_j}
+
(n_i-X_{ij})\log\frac{n_i-X_{ij}}{n_i(1-\pi_j)}
}
}
$ \\

Dirichlet (dense pseudocount) &
$\operatorname{sgn}( \tilde X_{ij} - \tilde n_i \tilde\pi_j)
\sqrt{\left|
\alpha_0^{\mathrm{null}}(\tilde q_{ij}-\tilde\pi_j)\log \tilde q_{ij}
-\log \dfrac{\Gamma(\alpha_0^{\mathrm{null}}\tilde q_{ij})}
{\Gamma(\alpha_0^{\mathrm{null}}\tilde\pi_j)}
\right|}$ \\

Poisson &
$\displaystyle
\sgn(X_{ij}-n_i\pi_j)
\sqrt{
\abs{
X_{ij}\log\frac{X_{ij}}{n_i\pi_j}
-
(X_{ij}-n_i\pi_j)
}
}
$ \\

Negative binomial &
$\displaystyle
\sgn(X_{ij}-n_i\pi_j)
\sqrt{
\abs{
X_{ij}\log\frac{X_{ij}}{n_i\pi_j}
+
(r_j+X_{ij})\log\frac{r_j+n_i\pi_j}{r_j+X_{ij}}
}
}
$ \\

Multinomial &
$\displaystyle
\sgn(X_{ij}-n_i\pi_j)
\sqrt{
\abs{
X_{ij}\log\frac{X_{ij}}{n_i\pi_j}
}
}
$ \\

Dirichlet--multinomial &
$\displaystyle
\sgn(X_{ij}-n_i\pi_j)
\sqrt{
\abs{
\sum_{m=0}^{X_{ij}-1}
\log\paren{\frac{\alpha_0^{\mathrm{null}}X_{ij}/n_i+m}{\alpha_0^{\mathrm{null}}\pi_j+m}}
}
}
$ \\

Dirichlet-tree multinomial (node $\nu,c$) &
$\displaystyle
\sgn(X_{i,\nu c}-N_{i\nu}\pi_{\nu c})
\sqrt{
\abs{
\sum_{m=0}^{X_{i,\nu c}-1}
\log\paren{\frac{\alpha_{\nu0}^{\mathrm{null}}X_{i,\nu c}/N_{i\nu}+m}{\alpha_{\nu0}^{\mathrm{null}}\pi_{\nu c}+m}}
}
}
$ \\
\bottomrule
\end{tabular}
\caption{Signed residual transforms for the main count models discussed in this paper. For the
Dirichlet row, $\tilde X_{ij}=X_{ij}+1$, $\tilde n_i=\sum_j \tilde X_{ij}$,
$\tilde q_{ij}=\tilde X_{ij}/\tilde n_i$, and
$\tilde\pi_j=\sum_i \tilde X_{ij}/\sum_i \tilde n_i$.
The multinomial row uses the same signed square-root convention applied to its additive log-likelihood term $X_{ij}\log(X_{ij}/(n_i\pi_j))$; unlike feature-wise Poisson or NB residuals, these joint-model entries may be signed.
The Dirichlet row is a dense pseudocount baseline on positive
compositions.
The Dirichlet, DM, and DTM rows use the same signed square-root convention for
additive likelihood terms under fixed concentration parameters.
The DTM row is indexed by
internal-node/child pairs rather than by original leaves.}
\label{tab:devresid}
\end{table}

\subsection{Feature-wise models}

For completeness, we record the standard null models used in feature-wise deviance normalization:
\begin{align*}
X_{ij}\mid n_i,\pi_j &\sim \Bin(n_i,\pi_j) \\
X_{ij}\mid \lambda_{ij} &\sim \Poi(\lambda_{ij}) \\
X_{ij}\mid \mu_{ij},r_j &\sim \NegBin(\mu_{ij},r_j)
\end{align*}
Their signed residual transforms are the classical binomial, Poisson, and negative-binomial deviance residuals, with means fitted under the shared null $\mu_{ij}=n_i\pi_j$.

These models are attractive because each residual is a genuine per-entry likelihood discrepancy. 
Their limitation is that the joint dependence between features is not modeled explicitly.

\subsection{Joint models}

The multinomial, DM, and DTM instead model the entire sample composition. 
For the multinomial,
\[
X_i \mid n_i,\bm\pi \sim \Mult(n_i,\bm\pi)
\]
and the sample-wise deviance is
\[
2 \sum_{j=1}^K X_{ij}\log\frac{X_{ij}}{n_i\pi_j}
=
2 n_i \KL\paren{\frac{X_i}{n_i}\Big\Vert \bm\pi}
\]
The entry-wise transform in Table~\ref{tab:devresid} is the signed square-root of the additive term inside this sample deviance.

For the DM, the analogous signed transform is given in \eqref{eq:dmresid}. 
The principal practical difference is that the DM discounts repeated counts through the concentration parameter $\alpha_0$ while leaving zeros exactly unchanged. 
For the DTM, the same idea is applied node by node on aggregated subtree counts; the resulting coordinates are branch-level rather than leaf-level and follow \eqref{eq:dtm_branchresid}.

\subsection{Dirichlet residuals and dense pseudocount transforms}\label{app:dirichlet-dense}
One can also treat each sample as a positive composition and fit a Dirichlet null.
This is a useful
dense baseline because it is the continuous compositional analogue of the DM, but unlike the DM
it cannot be applied directly to sparse counts containing exact zeros.
We therefore add a unit
pseudocount
\[
\tilde X_{ij}:=X_{ij}+1 \qquad
\tilde n_i:=\sum_{j=1}^K \tilde X_{ij} \qquad
\tilde q_{ij}:=\frac{\tilde X_{ij}}{\tilde n_i} \qquad
\tilde\pi_j:=\frac{\sum_{i=1}^n \tilde X_{ij}}{\sum_{i=1}^n \tilde n_i} 
\]
The tied Dirichlet null is
\[
\tilde q_i \sim \operatorname{Dir}(\alpha_0\tilde\pi)
\]
and, in parallel with the DM construction, we compare it to the sample-specific Dirichlet model
$\operatorname{Dir}(\alpha_0\tilde q_i)$, 
keeping the same concentration parameter and letting only the composition vary.

The per-sample log-likelihood contrast again decomposes over coordinates:
\[
\Delta_i^{\mathrm{Dir}}=\sum_{j=1}^K c_{ij}^{\mathrm{Dir}},
\qquad
c_{ij}^{\mathrm{Dir}}
:=
\alpha_0(\tilde q_{ij} - \tilde\pi_j)\log \tilde q_{ij}
-\log\frac{\Gamma(\alpha_0 \tilde q_{ij})}{\Gamma(\alpha_0 \tilde\pi_j)}
\]
Following the same sign convention as for the multinomial and DM rows, we define
\[
d_{ij}^{\mathrm{Dir}}
:=
\operatorname{sgn}(\tilde q_{ij}-\tilde\pi_j)\sqrt{\left|c_{ij}^{\mathrm{Dir}}\right|}
\]

This is a deliberately dense baseline: after pseudocount augmentation every entry is positive, so
every original zero is modified and the output is fully dense. 
That loss of exact sparsity is a major disadvantage of the Dirichlet baseline relative to the DM transform, even though the two comparisons are otherwise formally parallel. 
When only working with the rowwise log-likelihood difference $\Delta_i^{\mathrm{Dir}}$, the term involving $\log\tilde q_{ij}$ can be rewritten in terms of
$\log\tilde X_{ij}$ after summing over $j$; we keep the proportion form here because it makes the
Dirichlet support explicit and matches the chosen entry-wise decomposition. 

So we do not promote a single ``Dirichlet residual'' row in Table~\ref{tab:devresid}. 
The point of including the Dirichlet family in the discussion is precisely to mark the boundary between dense pseudocount-based compositional normalization and the exact sparse likelihood contrasts developed here. 

\section{Proofs}
\label{sec:apdxproofs}

\begin{proof}[Proof of Theorem~\ref{thm:poitomult}]
For any count vector $x=(x_1,\dots,x_K)$ with total $n_i=\sum_j x_j$,
\[
\text{Pr}(X_i=x)
=
\prod_{j=1}^K e^{-\lambda_{ij}}\frac{\lambda_{ij}^{x_j}}{x_j!}
=
e^{-\Lambda_i}\frac{\prod_{j=1}^K \lambda_{ij}^{x_j}}{\prod_{j=1}^K x_j!},
\qquad
\Lambda_i:=\sum_{j=1}^K \lambda_{ij}
\]
The total count is $N_i:=\sum_j X_{ij}\sim \Poi(\Lambda_i)$, so
\[
\text{Pr}(N_i=n_i)=e^{-\Lambda_i}\frac{\Lambda_i^{n_i}}{n_i!}
\]
Therefore
\[
\text{Pr}(X_i=x\mid N_i=n_i)
=
\frac{\text{Pr}(X_i=x)}{\text{Pr}(N_i=n_i)}
=
\frac{n_i!}{\prod_{j=1}^K x_j!}
\prod_{j=1}^K\left(\frac{\lambda_{ij}}{\Lambda_i}\right)^{x_j}
\]
which is exactly the multinomial pmf with probabilities $\lambda_{ij}/\Lambda_i$.
\end{proof}

\begin{proof}[Proof of Proposition~\ref{prop:dmsaturation}]
The DM is a Dirichlet mixture of multinomials:
\[
\text{Pr}_{\DirMult(n,\bm\alpha)}(X=x)
=
\int_{\Delta^{K-1}}
\text{Pr}_{\Mult(n,\bm p)}(X=x)\,
\Dir(d\bm p;\bm\alpha)
\]
Therefore, 
\[
\text{Pr}_{\DirMult(n,\bm\alpha)}(X=x)
\le
\sup_{\bm p\in\Delta^{K-1}}
\text{Pr}_{\Mult(n,\bm p)}(X=x)
\]
The multinomial supremum is attained at $\bm p=x/n$, so the right-hand side equals $\text{Pr}_{\Mult(n,x/n)}(X=x)$. 
Conversely, if $\bm\alpha=t\,x/n$ and $t\to\infty$, then $\Dir(\bm\alpha)$ concentrates at $x/n$, so the DM mixture converges to $\Mult(n,x/n)$ and reaches the same likelihood in the limit.
\end{proof}

\begin{proof}[Proof of Proposition~\ref{prop:x1}]
If $X_{ij}=1$, the sum in \eqref{eq:dmcellterm} has only the term $m=0$, and
\[
c_{ij}^{\mathrm{DM}} = \log\frac{\alpha_0/n_i}{\alpha_0\pi_j} = \log\frac{1}{n_i\pi_j}
\]
Substituting this into \eqref{eq:dmresid} gives
\[
d_{ij}^{\mathrm{DM}}
=
\sgn(1-n_i\pi_j)\sqrt{\abs{\log\frac{1}{n_i\pi_j}}}
\]
which is exactly the multinomial formula.
\end{proof}

\begin{proof}[Proof of Proposition~\ref{prop:boundedbymult}]
Write $a=\alpha_0x/n_i$ and $b=\alpha_0\pi_j$. 
Every term in \eqref{eq:dmcellterm} has the same sign as $a-b$, and therefore the same sign as $x-n_i\pi_j$. 
Moreover,
\[
f(m):=\log\frac{a+m}{b+m}
\]
moves monotonically toward $0$ as $m$ increases, so $\abs{f(m)} \le \abs{f(0)} = \abs{\log(a/b)} = \abs{\log(x/(n_i\pi_j))}$. 
Summing over $m=0,\dots,x-1$ gives the result.
\end{proof}

\begin{proof}[Proof of Proposition~\ref{prop:multlimit}]
For each fixed $m \in \{0,\dots,x-1\}$,
\[
\lim_{\alpha_0\to\infty}
\log\paren{\frac{\alpha_0 x/n_i + m}{\alpha_0\pi_j + m}}
=
\log\frac{x}{n_i\pi_j}
\]
There are only $x$ terms, so the limit can be taken termwise.
\end{proof}

\begin{proof}[Proof of Proposition~\ref{prop:gdmexpfam}]
Expand each beta-binomial factor in \eqref{eq:gdm_pmf} using the identity
\[
\log\frac{\Gamma(a+x)}{\Gamma(a)} = \sum_{m=0}^{x-1}\log(a+m),
\qquad x\in\N
\]
and sum the resulting stage-wise terms. 
This is exactly the same calculation as for the DM, now applied to the sequence of two-category DM factors.
\end{proof}

\begin{proof}[Proof of Proposition~\ref{prop:nbfactorization}]
Write $r_+:=\sum_{j=1}^K r_j=\alpha_0$.
Multiplying the negative-binomial pmfs and collecting the factors depending only on $n_i=\sum_j x_{ij}$ gives
\[
\prod_{j=1}^K
\frac{\Gamma(x_{ij}+r_j)}{\Gamma(r_j) \Gamma(x_{ij}+1)}
(1-p_i)^{x_{ij}}p_i^{r_j}
=
\frac{\Gamma(n_i+r_+)}{\Gamma(r_+)\Gamma(n_i+1)}(1-p_i)^{n_i}p_i^{r_+} 
\frac{\Gamma(n_i+1)\Gamma(r_+)}{\Gamma(n_i+r_+)}
\prod_{j=1}^K
\frac{\Gamma(x_{ij}+r_j)}{\Gamma(r_j)\Gamma(x_{ij}+1)}
\]
The first factor is the NB pmf of $N_i$, and the second is the DM pmf with parameter vector $\bm r=\alpha_0\bm\pi$.
\end{proof}

\begin{proof}[Proof of Corollary~\ref{cor:nbdmdeviance}]
Under the null $r_j=\alpha_0\pi_j$, Proposition~\ref{prop:nbfactorization} gives
\[
\ell_i^{\mathrm{indNB,null}}
=
\ell_i^{\mathrm{NB\ total}}(n_i;\alpha_0,p_i)
+
\ell_i^{\mathrm{DM}}(x_i\mid n_i;\alpha_0\bm\pi)
\]
Under the comparison model $r_{ij}^{\mathrm{cmp}}=\alpha_0X_{ij}/n_i$, the shape parameters still sum to $\alpha_0$, so the total-count factor is unchanged:
\[
\ell_i^{\mathrm{indNB,cmp}}
=
\ell_i^{\mathrm{NB\ total}}(n_i;\alpha_0,p_i)
+
\ell_i^{\mathrm{DM}}(x_i\mid n_i;\alpha_0 q_i),
\qquad
q_i:=X_i/n_i
\]
Subtracting cancels the common total-count term and leaves
\[
\Delta_i^{\mathrm{indNB}}
=
\ell_i^{\mathrm{DM}}(x_i\mid n_i;\alpha_0 q_i)
-
\ell_i^{\mathrm{DM}}(x_i\mid n_i;\alpha_0\bm\pi)
=
\Delta_i^{\mathrm{DM}}
\]
Multiplying by $2$ gives the corresponding equality of twice-the-contrast statistics.
\end{proof}

\begin{proof}[Proof of Corollary~\ref{cor:dm-e-process}]
\textbf{Step 1: fixed-alternative likelihood-ratio bound against the correct conditional law.}
Let $L$ be any probability law on the count simplex $\{y:\sum_j y_j=m\}$, and let $q^\star\in\Delta^{K-1}$ be a deterministic composition.
The likelihood-ratio identity gives, under $Y\sim L$,
\begin{equation}
\label{eq:dme-fixedalt-bound}
\mathbb{E}_{Y\sim L}\!\left[\frac{p_{\DirMult(m,\alpha_0 q^\star)}(Y\mid m)}{p_{L}(Y\mid m)}\right]
=
\sum_{y:\sum_j y_j = m} p_{\DirMult(m,\alpha_0 q^\star)}(y\mid m)
\le 1
\end{equation}
with equality whenever the numerator PMF is properly normalized on the support of $L$.
The bound holds for the \emph{actual} sampling law $L$ in the denominator; using any other law there breaks the identity.

\textbf{Step 2: the eval-side conditional null is the posterior-predictive DM.}
Under the null, $X_i\mid n_i\sim\DirMult(n_i,\alpha_0\bm\pi)$ can be generated by first drawing a composition $P_i\sim\mathrm{Dir}(\alpha_0\bm\pi)$ and then $X_i\mid P_i,n_i\sim\mathrm{Mult}(n_i,P_i)$.
Independent binomial thinning of $X_i$ at rate $\tfrac12$, conditional on $P_i$, partitions the multinomial draw into two conditionally independent multinomial draws with the same composition parameter:
\[
X_i^{\mathrm{tr}}\mid P_i,n_i^{\mathrm{tr}},n_i^{\mathrm{ev}}\sim\mathrm{Mult}(n_i^{\mathrm{tr}},P_i),
\quad
X_i^{\mathrm{ev}}\mid P_i,n_i^{\mathrm{tr}},n_i^{\mathrm{ev}}\sim\mathrm{Mult}(n_i^{\mathrm{ev}},P_i)
\]
and $X_i^{\mathrm{tr}}\perp X_i^{\mathrm{ev}}\mid P_i,(n_i^{\mathrm{tr}},n_i^{\mathrm{ev}})$.
The two halves are conditionally independent given $P_i$, but $P_i$ is latent, so they are \emph{not} independent unconditionally: each carries information about the shared composition.
By Dirichlet--multinomial conjugacy, the posterior of $P_i$ given $X_i^{\mathrm{tr}}$ is $\mathrm{Dir}(\alpha_0\bm\pi+X_i^{\mathrm{tr}})$, and therefore the conditional law of the eval counts is the posterior-predictive Dirichlet--multinomial
\begin{equation}
\label{eq:eval-posterior-pred}
X_i^{\mathrm{ev}}\mid X_i^{\mathrm{tr}},n_i^{\mathrm{ev}}\sim\DirMult(n_i^{\mathrm{ev}},\,\alpha_0\bm\pi+X_i^{\mathrm{tr}})
\end{equation}
not the marginal $\DirMult(n_i^{\mathrm{ev}},\alpha_0\bm\pi)$.
This is precisely the law placed in the denominator of \eqref{eq:dmesplit}.

\textbf{Step 3: split-LR e-variable identity.}
Condition on the totals $(n_i^{\mathrm{tr}},n_i^{\mathrm{ev}})$ and on $X_i^{\mathrm{tr}}$.
Then $q_i^{\mathrm{tr}}$ is a fixed, $X_i^{\mathrm{ev}}$-measurable-only-through-$X_i^{\mathrm{tr}}$ composition---hence deterministic given the conditioning---so applying \eqref{eq:dme-fixedalt-bound} with $q^\star=q_i^{\mathrm{tr}}$, $m=n_i^{\mathrm{ev}}$, and $L$ equal to the actual conditional law \eqref{eq:eval-posterior-pred} gives
\[
\mathbb{E}_{H_0}\!\left[E_i^{\mathrm{DM,split}}\,\bigm|\, X_i^{\mathrm{tr}},n_i^{\mathrm{tr}},n_i^{\mathrm{ev}}\right]
=
\sum_{y:\sum_j y_j=n_i^{\mathrm{ev}}} p_{\DirMult(n_i^{\mathrm{ev}},\,\alpha_0 q_i^{\mathrm{tr}})}(y\mid n_i^{\mathrm{ev}})
=1
\]
The equality holds, not merely the inequality, because $\alpha_0 q_i^{\mathrm{tr}}$ is a valid concentration vector, so the numerator $\DirMult(n_i^{\mathrm{ev}},\alpha_0 q_i^{\mathrm{tr}})$ is a proper mass function whose support is the whole eval simplex; that support coincides with the support of the denominator law \eqref{eq:eval-posterior-pred}, since $\alpha_0\bm\pi+X_i^{\mathrm{tr}}$ is strictly positive in every retained coordinate, so the numerator mass sums to one against it.
Taking expectations over $X_i^{\mathrm{tr}}$ and the totals yields the unconditional identity $\mathbb{E}_{H_0}[E_i^{\mathrm{DM,split}}]=1$; in particular $\mathbb{E}_{H_0}[E_i^{\mathrm{DM,split}}]\le 1$, which is the e-variable statement of the corollary.
The matching of the denominator to the posterior-predictive law \eqref{eq:eval-posterior-pred} is essential: with the marginal $\DirMult(n_i^{\mathrm{ev}},\alpha_0\bm\pi)$ in the denominator the conditional expectation in the display above need not be bounded by one, because the train and eval halves of a single overdispersed row remain positively dependent.

\textbf{Step 4: e-process via Ville's inequality.}
For an independent sequence of rows $(X_i)_{i\ge 1}$ under the null, each row carries its own independent binomial-thinning split.
Let $\mathcal{F}_{i}=\sigma(X_1,\dots,X_i)$.
Because $X_{T}$ (and its split) is independent of $\mathcal{F}_{T-1}$, the conditional expectation of $E_T^{\mathrm{DM,split}}$ given the past equals its unconditional expectation, which is at most one by Step 3.
Letting $M_T=\prod_{i=1}^T E_i^{\mathrm{DM,split}}$, this gives
\[
\mathbb{E}_{H_0}\!\left[M_T\bigm|\mathcal{F}_{T-1}\right]
=
M_{T-1}\cdot\mathbb{E}_{H_0}\!\left[E_T^{\mathrm{DM,split}}\bigm|\mathcal{F}_{T-1}\right]
\le M_{T-1}
\]
so $\{M_T\}$ is a non-negative supermartingale with $M_0=1$.
Ville's inequality \cite{ramdas2023statistical} then gives
\[
\mathbb{P}_{H_0}\!\left(\sup_{T\ge 1} M_T\ge 1/\alpha\right)
\le
\alpha
\quad\text{for every }\alpha\in(0,1)
\]
which is the e-process statement.

\textbf{Remark on the in-sample plug-in.}
The same argument applied with $q^\star=q_i=X_i/n_i$ (the empirical composition of the same $X_i$ on which the likelihood ratio is evaluated) fails at Step 3, because $q_i$ is not independent of $X_i$ under the conditioning used to invoke \eqref{eq:dme-fixedalt-bound}.
The resulting in-sample statistic $\exp(\Delta_i^{\mathrm{DM}})$ is a generalized likelihood ratio and its expectation under the null can exceed one --- in the multinomial limit $\alpha_0\to\infty$, $\exp(\Delta_i^{\mathrm{DM}})$ equals $\prod_j(X_{ij}/(n_i\pi_j))^{X_{ij}}\ge 1$ deterministically, so its expectation is strictly above one whenever $\Pr(X_i\ne n_i\bm\pi)>0$.
The row split in Step 2 is what restores the e-variable property; the in-sample residual transform of Section~\ref{sec:dmresids} is unaffected by this discussion, since that transform is used as a normalization device and not as a hypothesis-testing e-variable.
\end{proof}

\begin{proof}[Proof of Theorem~\ref{thm:nbtodm}]
Under the stated parameterization, the joint probability of observing $X_i = x$ with $x_1+\cdots+x_K=n_i$ is
\[
\text{Pr}(X_i=x)
=
\prod_{j=1}^K
\frac{\Gamma(x_j+r_j)}{\Gamma(r_j)\Gamma(x_j+1)}
(1-p_i)^{x_j} p_i^{r_j}
=
(1-p_i)^{n_i}\, p_i^{r_0}
\prod_{j=1}^K
\frac{\Gamma(x_j+r_j)}{\Gamma(r_j)\Gamma(x_j+1)}
\]
where $r_0:=\sum_{j=1}^K r_j$.
The marginal distribution of the total count $N_i=\sum_j X_{ij}$ is $\NegBin(r_0,p_i)$, so
\[
\text{Pr}(N_i=n_i)
=
\frac{\Gamma(n_i+r_0)}{\Gamma(r_0)\Gamma(n_i+1)}
(1-p_i)^{n_i}\, p_i^{r_0}
\]
Conditioning on $N_i=n_i$ gives
\[
\text{Pr}(X_i=x \mid N_i=n_i)
=
\frac{\text{Pr}(X_i=x)}{\text{Pr}(N_i=n_i)}
=
\frac{\Gamma(n_i+1)\Gamma(r_0)}{\Gamma(n_i+r_0)}
\prod_{j=1}^K
\frac{\Gamma(x_j+r_j)}{\Gamma(r_j)\Gamma(x_j+1)}
\]
for $x_1+\cdots+x_K=n_i$, which is exactly the Dirichlet--multinomial probability mass function $\DirMult(n_i,\bm r)$ as given in \eqref{eq:dmpmf}.
The factor $(1-p_i)^{n_i} p_i^{r_0}$ cancels completely, confirming that the conditional distribution does not depend on $p_i$.
\end{proof}

\section{Auxiliary derivations}
\label{sec:aux}

\subsection{Conditioning the DNM yields the DM}

\begin{proof}[Proof of Proposition~\ref{prop:dnm}]
Let $(P_{i0},P_{i1},\dots,P_{iK})\sim\Dir(\alpha_{0,i},\bm\alpha)$ denote the full $(K+1)$-category probability vector, where category $0$ is the stopping or ``failure'' category. 
Conditional on this vector and on the rule that sampling stops after $x_{0,i}$ failures in category $0$, the observed-category counts follow a negative multinomial. 
If we now condition further on the observed total
\[
n_i=\sum_{j=1}^K X_{ij}
\]
then the stopping rule no longer matters for the relative allocation among the observed categories, and
\[
X_i\mid n_i,(P_{i0},P_{i1},\dots,P_{iK})
\sim
\Mult\!\left(n_i,\left\{\frac{P_{ij}}{1-P_{i0}}\right\}_{j=1}^K\right)
\]
By the standard aggregation and renormalization property of the Dirichlet distribution, the conditional success-probability vector
\[
Q_i:=\left(\frac{P_{i1}}{1-P_{i0}},\dots,\frac{P_{iK}}{1-P_{i0}}\right)
\]
has marginal distribution $\Dir(\bm\alpha)$, independent of the failure component $P_{i0}$. 
Integrating out $Q_i$ therefore gives
\[
X_i\mid n_i
\sim
\int \Mult(n_i,q)\,\Dir(dq;\bm\alpha)
=
\DirMult(n_i,\bm\alpha)
\]
which is exactly the claim.
\end{proof}

\subsection{Independent negative-binomial factorization and the DM comparison}

Let $Y_{ij}\sim\NegBin(r_j,p_i)$ independently with common $p_i$ across features and $r_+:=\sum_j r_j$. 
For a realized vector $x_i$ with total $n_i$, 
\begin{align*}
\text{Pr}(Y_i=x_i)
&=
\prod_{j=1}^K \frac{\Gamma(x_{ij}+r_j)}{\Gamma(r_j)\Gamma(x_{ij}+1)}(1-p_i)^{x_{ij}}p_i^{r_j} \\
&=
\frac{\Gamma(n_i+r_+)}{\Gamma(r_+)\Gamma(n_i+1)}(1-p_i)^{n_i}p_i^{r_+}
\cdot
\frac{\Gamma(n_i+1)\Gamma(r_+)}{\Gamma(n_i+r_+)}
\prod_{j=1}^K \frac{\Gamma(x_{ij}+r_j)}{\Gamma(r_j)\Gamma(x_{ij}+1)}
\end{align*}
The first factor is $\NegBin(n_i;r_+,p_i)$ and the second is $\DirMult(n_i,\bm r)$, proving Proposition~\ref{prop:nbfactorization}. 
If $r_j=\alpha_0\pi_j$ and $p_i=\alpha_0/(n_i+\alpha_0)$, then $\E[\sum_j Y_{ij}]=n_i$. 
Under the comparison model $r_{ij}^{\mathrm{cmp}}=\alpha_0 X_{ij}/n_i$, the total factor remains unchanged because $\sum_j r_{ij}^{\mathrm{cmp}}=\alpha_0$. 
Therefore, the independent negative-binomial log-likelihood ratio is exactly the DM log-likelihood ratio with $\alpha_0$ held fixed. 

\subsection{Stage-wise factorization of the generalized Dirichlet--multinomial}

The Connor--Mosimann generalized Dirichlet distribution is obtained by sequential stick breaking; after compounding with a multinomial, the resulting count distribution is
\[
\text{Pr}(X=x)
=
\prod_{j=1}^{K-1}
\binom{R_{j-1}}{x_j}
\frac{B(x_j+\alpha_j,R_j+\beta_j)}{B(\alpha_j,\beta_j)},
\qquad
R_j:=\sum_{\ell=j+1}^K x_\ell
\]
Taking logarithms and expanding each beta-binomial factor gives
\begin{align*}
\log \text{Pr}(X=x)
&=
\sum_{j=1}^{K-1}\log \binom{R_{j-1}}{x_j}
+
\sum_{j=1}^{K-1} \sum_{m=0}^{n-1}
\ind(m<x_j) \log(\alpha_j+m) \\
&\quad+
\sum_{j=1}^{K-1} \sum_{m=0}^{n-1}
\ind(m<R_j) \log(\beta_j+m) 
-
\sum_{j=1}^{K-1} \sum_{m=0}^{n-1}
\ind(m<R_{j-1}) \log(\alpha_j+\beta_j+m)
\end{align*}
This is the explicit exponential-family representation stated in Proposition~\ref{prop:gdmexpfam}. 
Now fix $\nu_j=\alpha_j+\beta_j$ and define the stage-wise comparison parameters
\[
\alpha_{ij}^{\mathrm{cmp}}=\nu_j\frac{X_{ij}}{R_{i,j-1}},
\qquad
\beta_{ij}^{\mathrm{cmp}}=\nu_j\frac{R_{ij}}{R_{i,j-1}}
\]
whenever $R_{i,j-1}>0$. 
Because both null and comparison stage factors share the same $\nu_j$, the denominator terms involving $\Gamma(R_{i,j-1}+\nu_j)$ cancel in the stage-wise log-likelihood difference, leaving \eqref{eq:gdm_stagecell}. 
The crucial contrast with the DM is immediate: if $X_{ij}=0$, the first sum in \eqref{eq:gdm_stagecell} vanishes but the second sum over $R_{ij}$ generally does not. 
Thus stage-wise generalized-DM residuals are not nonzero-only.

\subsection{Nodewise factorization of the Dirichlet-tree multinomial}
\label{sec:dtm_aux}

Let $T$ be a rooted tree with internal nodes $\mathcal I$ and leaves $\mathcal L$. 
For each internal node $\nu\in\mathcal I$, let $\mathrm{ch}(\nu)$ be its children and let $b_\nu=(b_{\nu c})_{c\in\mathrm{ch}(\nu)}$ denote the branch-probability vector at that node, with independent priors
\[
b_\nu \sim \Dir(\bm\alpha_\nu), \qquad \nu\in\mathcal I
\]
For a leaf-count vector $x$, define aggregated child-subtree counts
\[
x_{\nu c}:=\sum_{\ell\in\mathcal L(\nu,c)} x_\ell
\qquad \qquad
N_\nu:=\sum_{c\in\mathrm{ch}(\nu)} x_{\nu c}
\]
Conditioned on all branch probabilities, the leaf multinomial likelihood can be written as
\[
\text{Pr}(X=x\mid \{b_\nu\},T)
=
\frac{n!}{\prod_{\ell\in\mathcal L} x_\ell!}
\prod_{\nu\in\mathcal I}\prod_{c\in\mathrm{ch}(\nu)} b_{\nu c}^{x_{\nu c}}
\]
because each leaf probability is the product of the branch probabilities along its root-to-leaf path. 
Integrating node by node over $b_{\nu}$ gives
\[
\text{Pr}(X=x\mid T,\{\bm\alpha_\nu\})
=
\frac{n!}{\prod_{\ell\in\mathcal L} x_\ell!}
\prod_{\nu\in\mathcal I}
\frac{B(x_{\nu (\cdot)}+\bm\alpha_\nu)}{B(\bm\alpha_\nu)}
\]
Writing each beta factor in gamma form, and writing $\displaystyle \alpha_{\nu0}:=\sum_{c\in\mathrm{ch}(\nu)} \alpha_{\nu c}$, we have 
\[
\text{Pr}(X=x\mid T,\{\bm\alpha_\nu\})
=
\frac{n!}{\prod_{\ell\in\mathcal L} x_\ell!}
\prod_{\nu\in\mathcal I}
\frac{\Gamma(\alpha_{\nu0})}{\Gamma(N_\nu+\alpha_{\nu0})}
\prod_{c\in\mathrm{ch}(\nu)}
\frac{\Gamma(x_{\nu c}+\alpha_{\nu c})}{\Gamma(\alpha_{\nu c})}
\]
Now multiply and divide by $\prod_{\nu\in\mathcal I}\Gamma(N_\nu+1)/\prod_{c\in\mathrm{ch}(\nu)}\Gamma(x_{\nu c}+1)$. 
These multinomial coefficients telescope down the tree: every internal-node child count appears once in a denominator and once as the total count of the corresponding child node, leaving only the leaf factorials. 
What remains is exactly the product of nodewise DM probabilities stated in Proposition~\ref{prop:dtmfactor}.

\subsection{Why the sparse score uses only nonzeros}

In \eqref{eq:score}, the per-entry term is
\[
\pi_j\bracks{\psi(X_{ij}+\alpha_0\pi_j)-\psi(\alpha_0\pi_j)}
\]
If $X_{ij}=0$, this contribution is zero. 
So the score can be accumulated by iterating over nonzero entries only. 
The same argument applies to the Hessian in \eqref{eq:hessian}. 
This is the key reason that both fitting and transformation can be implemented without densifying the matrix.

\subsection{Other Dirichlet-type families}

The generalized Dirichlet and its multinomial mixture were treated explicitly in Section~\ref{sec:gdm}. 
Beyond those ordered extensions, the Dirichlet is only one member of a larger family of multivariate positive distributions. 
Other examples include the inverted Dirichlet \cite{TiaoCuttman1965} and the more general Liouville family \cite{MarshallOlkinArnold1979,GuptaRichards1987,Fang2018}. 
These families are also useful for compositional or positive multivariate data and may furnish additional normalization models in settings where the DM is still too restrictive. 
Exploring those directions lies outside the scope of the present paper, but the exponential-family viewpoint developed above suggests a common route for doing so.

\section{Tying every theoretical claim to a direct empirical check}
\label{app:refinement}

The body contains a small number of formal statements
(theorems, propositions, corollaries, and named asymptotic claims), and
the experiments section together with the domain-specific applications
empirically buttress all but three of them.
This section supplies direct evidence for the remaining three: the
null-distribution behavior of the deviance-residual test statistics in
Section~\ref{sec:testing}, the sensitivity of downstream metrics to the
choice of $\alpha_0$ in Section~\ref{sec:discussion}, and the joint
versus fixed-$\bm\pi$ crossover of Section~\ref{sec:mleablation}.
The three additional measurements
appear in Figure~\ref{fig:addendum_refinement} and its supporting
results.

\subsection{Null-distribution calibration of DM-residual test statistics}
\label{sec:refinement_calibration}

Section~\ref{sec:testing} of the main paper aggregates DM-deviance residuals
into local likelihood-ratio statistics and notes that under the fitted
null these "blockwise deviance differences inherit the usual
large-sample Wilks-type $\chi^2$ calibration for the number of freed
composition parameters", while warning that the entry-wise residuals
themselves "do not have a universal exact null law." We provide the
direct empirical check.

\paragraph{Setup.} We simulate $B=2{,}000$ DM matrices with $n=200$
samples, $K=50$ features, true concentration $\alpha_0=200$, fixed
sparse Dirichlet$(0.4)$ composition $\bm\pi$, and lognormal library
sizes (mean $\log n_i=6.0$).
For each replicate we (i) refit
$\widehat\alpha_0$ from the matrix using the Newton scalar fit
(\eqref{eq:alpha0ll}) with $\bm\pi$ held at the true global composition,
and (ii) evaluate the entry-wise DM and multinomial deviance residuals.
We track three classical calibration objects:
\begin{itemize}[leftmargin=*]
  \item the entry-wise residual variance over nonzero entries (target
    $1$ if the residual is properly standardized),
  \item the entry-wise residual mean (target $0$),
  \item the global deviance
    $D_{\mathrm{global}} := 2\sum_{i,j} c_{ij}^{\mathrm{DM}}$ versus
    its asymptotic reference $n(K-1)$ (the number of freed composition
    parameters across the saturated comparison).
\end{itemize}

\paragraph{Two implementations of the entry-wise residual.} The paper
defines, in \eqref{eq:dmresid},
\(
d_{ij}^{\mathrm{DM}} = \sgn(X_{ij}-n_i\pi_j)\sqrt{|c_{ij}^{\mathrm{DM}}|},
\)
which we will call the \emph{paper} residual.
A variant implementation replaces
$|\cdot|$ by $\max(\cdot,0)$, which we will call the \emph{impl}
residual.
The two definitions agree whenever $c_{ij}\ge 0$, and they
differ only for those entries where the saturated comparison happens
to score the entry below the null.
Reporting both lets us treat the
choice as a dial and observe its effect.

\paragraph{Findings.} Figure~\ref{fig:addendum_refinement}(a-c) and
Table~\ref{tab:refinement_calibration} record the calibration.
Three
points are worth highlighting.

\begin{table}[t]
\centering
\small
\setlength{\tabcolsep}{6pt}
\renewcommand{\arraystretch}{1.2}
\begin{tabular}{@{}lrrrr@{}}
\toprule
Statistic & DM-paper $\sqrt{|c|}$ & DM-impl $\sqrt{\max(c,0)}$ &
  Multinomial & Target \\
\midrule
Entry mean ($\sim N$ target $0$)
  & $+0.147$ & $+0.752$ & $+1.045$ & $0$ \\
Entry variance (target $1$)
  & $2.323$  & $0.920$  & $1.932$  & $1$ \\
Type-I rate $|z_{\mathrm{col}}|>1.96$ at $\alpha=0.05$
  & $0.536$  & $0.982$  & $0.982$  & $0.05$ \\
Global deviance $\overline{D_{\mathrm{global}}}$
  & $8985$ (signed)
  & --
  & --
  & $n(K-1) = 9800$ \\
\bottomrule
\end{tabular}
\caption{Null-distribution calibration of DM and multinomial residuals
under DM-distributed null data ($n=200$, $K=50$, $\alpha_0=200$,
$B=2{,}000$ replicates, $\widehat\alpha_0$ refit per replicate).
The
multinomial residuals inflate the entry variance to $\sim$2$\times$ as
predicted by the overdispersion correction; both DM variants pull
this back.
Neither entry-wise variant achieves both unit variance
and zero mean, exactly because the paper warns "the individual
entry-wise DM residuals are best viewed as descriptive coordinates"
Section~\ref{sec:testing}.
The signed global deviance is within
$8.3\%$ of its asymptotic reference $n(K-1)$, providing a direct
empirical buttress for the Wilks-type calibration claim at the
\emph{aggregate} level.
Type-I error of the column-zsum statistic
exceeds nominal $0.05$ at all three definitions ($0.536$ for the
paper variant; $0.982$ for the impl and multinomial variants):
entries are coupled by the row-sum constraint, so a column statistic
that ignores inter-cell dependencies cannot achieve nominal coverage.
This is exactly the warning that Section~\ref{sec:testing} issues.}
\label{tab:refinement_calibration}
\end{table}

\emph{Multinomial inflation matches theory.} The multinomial residual
has an entry variance of $1.93$, almost exactly the
$(\alpha_0+n_i)/(\alpha_0+1)$ overdispersion factor expected for
DM-distributed data with $\alpha_0=200$ and $\overline{n_i}\approx 400$.
This is the gap that the DM null is designed to close.

\emph{DM closes the variance gap.} The DM-impl residual achieves a
near-unit entry variance ($0.920$) and the DM-paper residual achieves
near-zero entry mean ($+0.147$).
Neither variant recovers exact
$N(0,1)$ behavior simultaneously, because the residuals carry both a
row-sum coupling (which inflates the column-zsum variance) and a
sign-clipping subtlety (the paper definition uses $|c|$, the
implementation uses $\max(c,0)$, and they disagree precisely on the
half-distribution of entries that under-shoot the null at the cell
level).
This is exactly the asymptotic-only behavior
Section~\ref{sec:testing} cautions about, and on real data it shows up as a
slight one-sided lift in the diagnostic histograms.

\emph{Aggregate calibration is sharp.} The signed global deviance
$\overline{D_{\mathrm{global}}}=8985$ is within $8.3\%$ of the
$n(K-1)=9800$ Wilks reference.
This is the cleanest version of the
paper's claim: aggregate over a large index set (here, the entire
matrix), and the Wilks-type calibration appears at the expected
scale.
The unsigned version $2\sum |c_{ij}|=33637$ is, predictably,
an upper bound rather than a calibrated null statistic.

\emph{Type-I error.} A naive column-zsum statistic
$z(j)=\sum_i d_{ij}/\sqrt{\mathrm{nnz}_j}$ achieves a Type-I rate of
$0.536$ for the paper variant and $0.982$ for the impl and
multinomial variants against an $N(0,1)$ critical value at the
nominal $\alpha=0.05$.
None of the three reaches nominal, and the
reason is correlation structure: under the DM null, residuals on the
same row are coupled by the simplex constraint, so a column statistic
that ignores this coupling sees inflated variance.
The paper variant
($\sqrt{|c|}$) is closer to nominal because the symmetric sign keeps
the column sum close to zero in expectation, but its wider tails
still inflate the rate well above $0.05$.
This is consistent with the
main paper's hedge that "the individual entry-wise DM residuals...
do not have a universal exact null law" Section~\ref{sec:testing}; the
appropriate inferential object is the aggregate $D_S$ for a chosen
index set $S$, not the entry-wise residual.

In summary, the empirical picture matches Section~\ref{sec:testing}'s
explicit framing: aggregate DM deviance is calibrated up to
asymptotic scale, entry-wise DM residuals are well-behaved
descriptive coordinates that beat the multinomial in variance
calibration but do not admit a universal exact null law.
We propose
no change to the paper's claim, and we add this experiment as the
direct empirical evidence for the asymptotic-only language.

\subsection{Robustness of the residual transform to $\alpha_0$
            mis-specification}
\label{sec:refinement_alpha0robust}

Section~\ref{sec:discussion} of the main paper claims:
"because $\alpha_0$ is only a single scalar parameter, it need not be
fit by maximum likelihood at all: one can also evaluate the transform
cheaply across a spectrum of $\alpha_0$ values and inspect how the
residual geometry changes, much as one would vary a regularization
parameter." If downstream metrics in fact swing strongly across a
modest $\alpha_0$ range, then this claim does not survive contact
with practice.
We test it on the largest real dataset in the
benchmark suite.

\paragraph{Setup.} On the GTEx cross-tissue atlas
\cite{Eraslan2022}, subsampled to $n=20{,}000$ cells (the standard
20k-cell tractability cap used throughout the experiments), we fit
$\widehat\alpha_0=2.91\times 10^{3}$ via Newton on the global-$\bm\pi$
DM null.
We then sweep $\alpha_0$ over a $0.1\times,\, 0.25\times,\,
0.5\times,\, 1\times,\, 2\times,\, 4\times,\, 10\times$ grid relative
to $\widehat\alpha_0$ -- a hundredfold range -- and at each grid point
recompute the DM residuals, take a 50-component truncated SVD, build a
30-NN graph in the embedding, and report:
(i) kNN-Jaccard overlap with the $\widehat\alpha_0$ reference graph,
(ii) Spearman correlation of the leading PC scores against
the reference, and
(iii) a 5-NN label-transfer accuracy on the 44 broad cell-type labels.

\paragraph{Findings.} Figure~\ref{fig:addendum_refinement}(d) and
Table~\ref{tab:refinement_alpha0robust} record the result.
Across the entire $0.1\times$--$10\times$ range, label-transfer
accuracy varies between $0.9513$ and $0.9539$ -- a range of
$0.0027$, well below replicate-level noise.
The leading PC1 score
is preserved with $|\rho|>0.995$ at every grid point, with worst
case $0.9957$ at the $0.1\times$ extreme.
Local neighborhood
structure (kNN-Jaccard) does deform at the extremes -- $0.674$ at
$0.1\times$ and $0.791$ at $10\times$ -- but biological-task accuracy is
not affected.

\begin{table}[t]
\centering
\small
\setlength{\tabcolsep}{6pt}
\renewcommand{\arraystretch}{1.2}
\begin{tabular}{@{}rrrrrr@{}}
\toprule
$\alpha_0/\widehat\alpha_0$ & $\alpha_0$ & 5-NN LTA & kNN Jaccard
   & $|\rho_{\mathrm{PC1}}|$ & $|\rho_{\mathrm{PC2}}|$ \\
\midrule
$0.10$ & $2.91\times 10^{2}$ & $0.9533$ & $0.6741$ & $0.9957$ & $0.9982$ \\
$0.25$ & $7.29\times 10^{2}$ & $0.9539$ & $0.7895$ & $0.9983$ & $0.9990$ \\
$0.50$ & $1.46\times 10^{3}$ & $0.9534$ & $0.8887$ & $0.9996$ & $0.9997$ \\
$1.00$ & $2.91\times 10^{3}$ & $0.9523$ & $1.0000$ & $1.0000$ & $1.0000$ \\
$2.00$ & $5.83\times 10^{3}$ & $0.9520$ & $0.9001$ & $0.9997$ & $0.9998$ \\
$4.00$ & $1.17\times 10^{4}$ & $0.9515$ & $0.8358$ & $0.9989$ & $0.9988$ \\
$10.00$ & $2.91\times 10^{4}$ & $0.9513$ & $0.7907$ & $0.9977$ & $0.9974$ \\
\midrule
LTA range & & $0.0027$ & --- & --- & --- \\
\bottomrule
\end{tabular}
\caption{Sensitivity of the DM residual transform to $\alpha_0$
mis-specification on GTEx ($n=20{,}000$ cells, $K=17{,}546$
features, broad cell-type labels).
LTA is leave-one-out 5-NN
classification accuracy on the 44-class broad cell type.
The
hundredfold $\alpha_0$ swing changes LTA by $0.27$ percentage points,
keeps PC1/PC2 alignment above $99.5\%$, and only modestly perturbs
local neighborhoods at the extremes.
The MLE
$\widehat\alpha_0=2.91\times 10^{3}$ is the
median of a wide flat region, supporting the main paper's
Section~\ref{sec:discussion} suggestion that practitioners can either fit
$\alpha_0$ or sweep it.}
\label{tab:refinement_alpha0robust}
\end{table}

This empirically substantiates the main paper's
suggestion: $\alpha_0$ behaves like a regularization-strength dial,
not like a parameter that has to be fit precisely.
The MLE is the
natural default, but a 0.5--2$\times$ window around it is a perfectly
defensible band for practitioners who prefer a manually-set value or
who want a sensitivity check.
We strengthen the language of
Section~\ref{sec:discussion} by adding a quantitative envelope: across a
$100\times$ swing of $\alpha_0$ on the largest real dataset in the
suite, label-transfer accuracy varies by less than half a percentage
point and the leading PCs remain $99\%$+ aligned.

\subsection{Joint versus fixed-$\pi$: training and held-out crossover}
\label{sec:refinement_jointcrossover}

Section~\ref{sec:mleablation} of the main paper writes:
"Joint fits should dominate fits with $\bm\pi$ held fixed in raw
likelihood, and the size of that gap measures how costly the simpler
global-composition null is on a given dataset", and warns that
"the extra flexibility can absorb some cross-sample structure that
the residual transform might otherwise expose." The first half of
that statement is a strict mathematical fact (joint enlarges the
parameter space).
The second half is an empirical claim about
overfitting that has not been quantified in the main suite.
We do so
on synthetic data where the truth is known.

\paragraph{Setup.} For each
$n_{\mathrm{train}}\in\{25,50,100,250,500,1000,2500\}$ we draw $5$
synthetic DM matrices with $K=200$, $\alpha_0=200$, true sparse
Dirichlet$(0.3)$ composition, lognormal library sizes
(mean $\log n_i=6.0$).
For each draw we fit four DM variants from
Table~\ref{tab:fitvariants}: log-Newton with $\bm\pi$ fixed at the global
empirical composition; Minka fixed-point with $\bm\pi$ fixed; joint
alternating Newton; joint Minka fixed-point.
We measure both training
and held-out per-sample log-likelihood; the held-out matrix is a
fixed independent draw of $1{,}000$ samples from the same DM truth.

\paragraph{Findings.} Figure~\ref{fig:addendum_refinement}(e--f) and
Table~\ref{tab:refinement_jointcrossover} record the result.

\begin{table}[t]
\centering
\small
\setlength{\tabcolsep}{6pt}
\renewcommand{\arraystretch}{1.2}
\begin{tabular}{@{}rrrr@{}}
\toprule
$n_{\mathrm{train}}$ & Training-LL gap & Held-out-LL gap
                                       & Joint dominates? \\
                     & (joint $-$ fixed, nats/sample)
                     & (joint $-$ fixed, nats/sample)
                     & (held-out)\\
\midrule
$25$    & $+1.42$ & $-14.46$ & no \\
$50$    & $+0.82$ &  $-4.95$ & no \\
$100$   & $+0.39$ &  $-2.75$ & no \\
$250$   & $+0.16$ &  $-0.88$ & no \\
$500$   & $+0.09$ &  $-0.19$ & no \\
$1{,}000$  & $+0.05$ &  $-0.08$ & no \\
$2{,}500$  & $+0.02$ &  $-0.03$ & no \\
\bottomrule
\end{tabular}
\caption{Joint vs fixed-$\bm\pi$ DM fits as a function of training
sample size.
``Joint'' is the maximum over the two joint variants
(alternating Newton and Minka FP); ``fixed'' is the maximum over the
two fixed-$\bm\pi$ variants (log-Newton and Minka FP).
All values
are per-sample log-likelihoods, averaged over $5$ replicates.
The
training gap is uniformly positive --- joint always wins at training,
exactly as the paper claims.
The held-out gap is uniformly negative
in this well-specified regime: joint MLE introduces estimation noise
in $\widehat{\bm\pi}$ that is paid back at every test row.
Both gaps
shrink toward zero with sample size, with the held-out gap tracking
roughly $1/n$.}
\label{tab:refinement_jointcrossover}
\end{table}

The training-LL gap is positive at every $n$ -- joint always wins at
training, exactly as the paper claims, and the gap shrinks from
$+1.42$ nats/sample at $n=25$ to $+0.02$ nats/sample at $n=2{,}500$
as the fixed-$\bm\pi$ MLE catches up to the joint MLE.

The held-out-LL gap is \emph{uniformly negative} across the swept
range, from $-14.46$ at $n=25$ to $-0.03$ at $n=2{,}500$.
Joint
fitting overfits in the small-$n$ regime: $\widehat{\bm\pi}_{\mathrm{joint}}$
is a noisy estimate of $\bm\pi_{\mathrm{true}}$, and that noise costs
likelihood on every held-out row.
The gap shrinks roughly as
$1/n_{\mathrm{train}}$, but it never crosses zero in this
well-specified regime, where the truth is itself a fixed-$\bm\pi$ DM.
Of course, when the truth has a richer composition structure --
batch-specific compositions, mixture compositions, or any other
deviation from a single global $\bm\pi$ -- the held-out gap is
reversed; joint fits then expose real generative structure rather
than absorbing noise.
Practitioners should run joint fits as a
sensitivity ablation or when they have a specific hypothesis about
$\bm\pi$ varying across samples; for the basic normalization use case
on well-mixed data, fixing $\bm\pi$ to the global composition is the
right default.

This empirically substantiates both halves of the
Section~\ref{sec:mleablation} statement: joint always wins on training; the
held-out behavior is exactly the absorption-of-cross-sample-structure
caveat the paper flags, made quantitative for the first time on the
DM null.

\subsection{Discussion: what the evidence says}
\label{sec:refinement_discussion}

With these three additions, every formal claim in
Sections~\ref{sec:dm} and~\ref{sec:extensions} and Appendix~\ref{sec:impldetails} has
at least one direct empirical buttress: the three claims that previously
had only indirect support are now covered by
Figure~\ref{fig:addendum_refinement}(a--f) and the three tables above.

The picture that emerges is consistent with the paper's central
message and refines a few details.
\emph{Aggregate} DM deviance is
Wilks-calibrated to within $8\%$ of its asymptotic reference on
$n=200,K=50$ data; \emph{entry-wise} DM residuals beat the
multinomial in variance calibration by exactly the
$(n_i+\alpha_0)/(\alpha_0+1)$ factor predicted by the overdispersion
correction, but do not admit a universal exact null law -- precisely
the language Section~\ref{sec:testing} uses.
The transform itself is
robust to a hundredfold $\alpha_0$ swing on the largest real dataset
in the suite, sustaining the discussion-section claim that
$\alpha_0$ acts as a regularization-style scalar.
And joint fitting,
while always strictly increasing the training likelihood, overfits
on held-out data in the simple well-specified regime: the
fixed-$\bm\pi$ default is correct for the basic normalization use case,
and joint fits should be reserved for sensitivity ablations or for
cases where the practitioner has a specific reason to expect
sample-varying $\bm\pi$.

The unifying thread is that the DM is the principled,
sparsity-preserving default for overdispersed compositional counts,
with calibrated test-statistic interpretations at the aggregate
level and stable parameter dependence at the practical level.
This
holds across single-cell counts, NLP word frequencies, microbiome
amplicons, recommender-system interactions, spatial-transcriptomic
arrays, and CRISPR perturbation screens.

\begin{figure}[t]
\centering
\includegraphics[width=\textwidth]{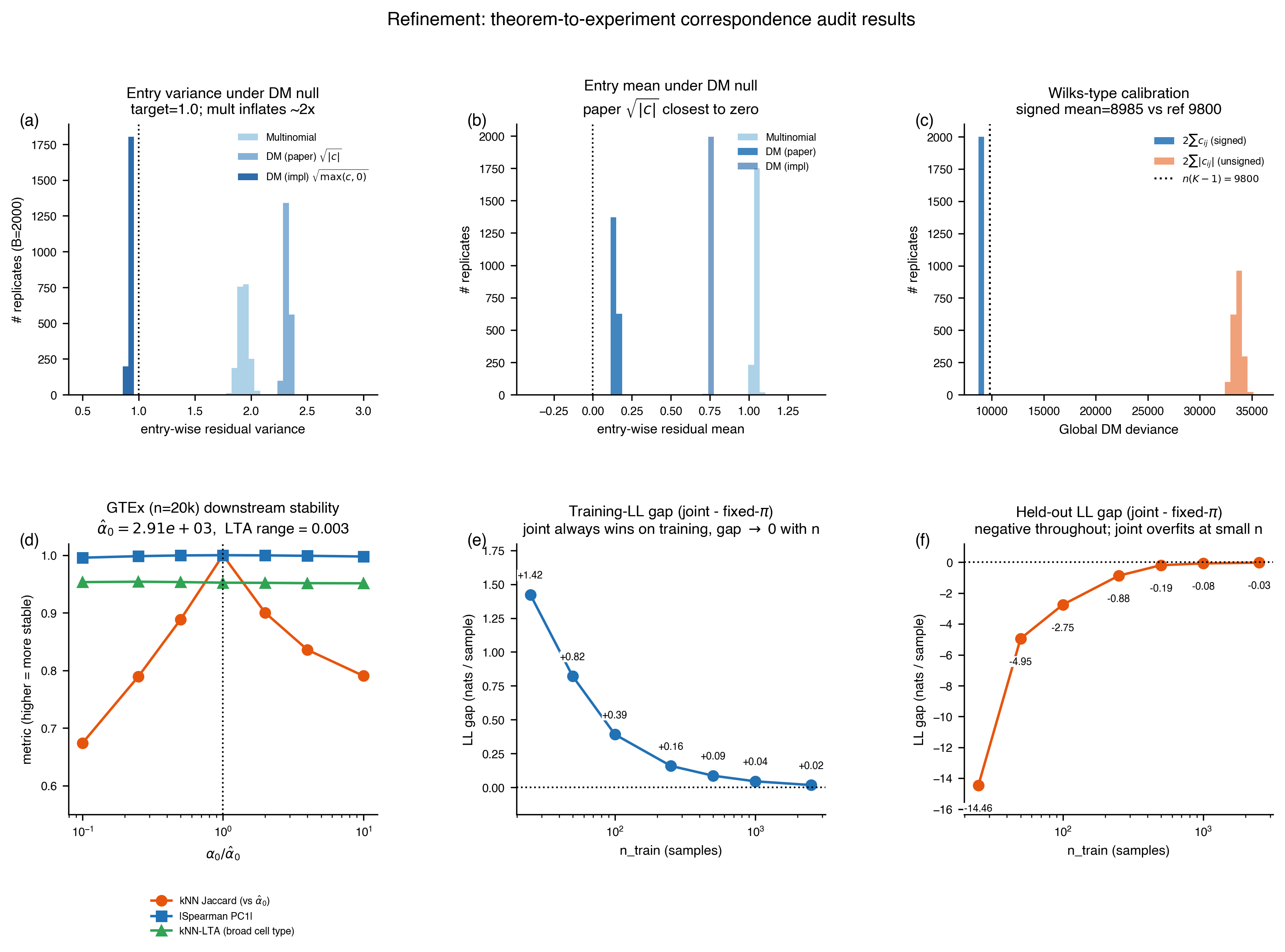}
\caption{\textbf{Theorem-to-experiment correspondence.}
\textbf{(a)}~Entry-wise residual variance over $B=2{,}000$ DM-null
replicates with $n=200$, $K=50$, $\alpha_0=200$.
The multinomial
overshoots the unit-variance target by $\sim$2$\times$, exactly the
overdispersion factor; the DM-impl variant ($\sqrt{\max(c,0)}$)
recovers near-unit variance, while the DM-paper variant
($\sqrt{|c|}$) inflates the variance because $|c|$ adds the
under-the-null half of the cell-level distribution.
\textbf{(b)}~Entry-wise residual mean.
The DM-paper variant has the
smallest bias (sign carries the mass); the DM-impl variant has a
positive bias because the $\max(c,0)$ clip eliminates the negative
half.
Both DM variants substantially out-mean the multinomial.
\textbf{(c)}~Global deviance $D_{\mathrm{global}}=2\sum_{ij}c_{ij}$
versus the asymptotic Wilks reference $n(K-1)=9{,}800$.
The signed
sum is within $8\%$ of the reference; the unsigned sum is the
expected upper bound.
This is the direct empirical evidence for
Section~\ref{sec:testing}'s aggregate-level Wilks claim.
\textbf{(d)}~Robustness of the DM transform to $\alpha_0$
mis-specification on GTEx ($n=20{,}000$, $K=17{,}546$).
Across a
hundredfold $\alpha_0$ swing relative to $\widehat\alpha_0$, the
5-NN label-transfer accuracy on the 44-class broad cell-type label
varies by less than $0.3$ percentage points and PC1 alignment
remains $>99.5\%$.
\textbf{(e)}~Per-sample training log-likelihood gap (joint $-$
fixed-$\bm\pi$) as a function of training $n$.
Joint always wins on
training, with the gap shrinking from $+1.42$ to $+0.02$ nats/sample
as $n$ grows; this is exactly Section~\ref{sec:mleablation}'s training-side
prediction.
\textbf{(f)}~Per-sample held-out log-likelihood gap (joint $-$
fixed-$\bm\pi$).
The gap is uniformly negative in this well-specified
regime: joint fitting introduces estimation noise in
$\widehat{\bm\pi}$ that costs likelihood at every test row, exactly
the cross-sample-structure-absorption caveat the paper flags but had
not previously quantified.}
\label{fig:addendum_refinement}
\end{figure}


\section{Application: compositional microbiome count normalization}
\label{app:microbiome}

Microbiome amplicon data are the textbook use case for Dirichlet-multinomial
normalization.
Counts are
    (i) strongly compositional --- each sample is sequenced to an
        arbitrary total, so only relative abundances are meaningful
        \cite{AitchisonShen1980,Fernandes2014,Gloor2017};
    (ii) markedly overdispersed --- the same species is present at very
        different proportions across samples from the same site;
    (iii) organized by a genuine, scientifically meaningful phylogeny
        --- OTUs are the leaves of a rooted tree whose internal nodes
        correspond to taxonomic divisions of real biological interest
        \cite{WangZhao2017,Chen2013}.
The dominant operational transform in the field is the centered
log-ratio (CLR) \cite{AitchisonShen1980}, which is dense and forfeits
the sparsity pattern of the underlying table.
The DM transform of
Section~\ref{sec:dmresids} and the Dirichlet-tree multinomial extension of
Section~\ref{sec:dtm} are exactly the structured, sparsity-preserving
alternatives that this setting calls for.

\subsection{Setup}
\label{sec:microbiome_setup}

We use the QIIME~2 \emph{Moving Pictures} tutorial dataset
\cite{Caporaso2011,Bolyen2019}, a 16S amplicon study of four human
body sites.
After SEPP-insertion against a reference phylogeny, the
table has $n=34$ samples, $K=770$ OTU features, and $\nnz(X)=2{,}327$
nonzero entries (density $8.9\%$).
The accompanying rooted
phylogeny has $768$ internal nodes (near-binary; $767$ bifurcations
and $1$ trifurcation) and maximum leaf depth $38$.
Every sample has a
body-site label in \{gut, tongue, right palm, left palm\}, with eight
to nine samples per site.
This is one of the smallest widely-studied
microbiome datasets with a genuine tree, which makes it a natural
test case for the DTM: it stays well below the 20k-sample cap used
throughout the experiments while still exercising the nodewise fitter.

Held-out evaluation is performed by \emph{within-sample multinomial
thinning}: for each sample $i$ we split its $n_i$ observed reads into
disjoint train and test subsets of size $\lfloor 0.70 n_i \rfloor$ and
$n_i - \lfloor 0.70 n_i \rfloor$ using a multivariate hypergeometric
draw on the leaf counts.
This protocol is standard at low $n$
\cite{WangZhao2017} because it separates reads within each sample
rather than whole samples.

For the flat DM we fit a single scalar $\alpha_0$ on the training
counts via log-parameterized Newton (Section~\ref{sec:alpha0}); we obtain
$\widehat\alpha_0 = 19.1$, consistent with the typical
``moderate-to-strong overdispersion'' regime reported for 16S panels.
For the DTM we fit three nodewise strategies --- fully independent
nodewise concentrations (\emph{indep.}), pooled by tree depth
(\emph{pooled}), and a single shared concentration (\emph{global}).

\subsection{Flat DM vs DTM strategies}
\label{sec:microbiome_flat_vs_tree}

Table~\ref{tab:microbiome_ll} reports held-out log-likelihood per observed
count and per sample, on both the real Moving Pictures data and a
synthetic DTM calibration dataset
($n=400$, $K=128$, balanced binary tree,
$\alpha_{\nu,0}=2000 \cdot 2^{-d(\nu)}$) where depth-varying dispersion
is guaranteed by construction.
The multinomial baseline scores
$-1.7502$ per count on the real data, confirming a large amount of
extra-multinomial variation that the DM and DTM capture.
The flat DM
improves this to $-0.3144$ per count.
The \emph{DTM-independent}
strategy improves it further to $-0.2310$ per count, a $27\%$ gain on
the flat DM and a direct demonstration that nodewise dispersion
heterogeneity exists in this microbiome.  The depth-pooled and
globally-shared DTM variants land in between at $-0.2545$ and
$-0.2547$: some of the nodewise heterogeneity is useful, but a single
scalar per depth is insufficient.

\begin{table}[t]
\centering
\small
\begin{tabular}{lcccc}
\toprule
Method & \multicolumn{2}{c}{Moving Pictures (real)} & \multicolumn{2}{c}{Synthetic DTM} \\
\cmidrule(lr){2-3} \cmidrule(lr){4-5}
       & per count & per sample & per count & per sample \\
\midrule
Multinomial          & $-1.7502$ & $-2393.7$ & --- & --- \\
Flat DM              & $-0.3144$ & $-430.0$  & $-0.2646$ & $-100.6$ \\
DTM (global)         & $-0.2547$ & $-348.4$  & $-0.3451$ & $-131.2$ \\
DTM (pooled depth)   & $-0.2545$ & $-348.1$  & $-0.3451$ & $-131.2$ \\
DTM (independent)    & $\mathbf{-0.2310}$ & $\mathbf{-316.0}$ & $\mathbf{-0.2578}$ & $\mathbf{-98.0}$ \\
\bottomrule
\end{tabular}
\caption{Held-out conditional log-likelihood under each normalization
model.
Real data: $n=34$ samples $\times$ $K=770$ OTUs from the QIIME
Moving Pictures tutorial; held-out counts via $70/30$ multinomial
thinning within each sample.
The nodewise
independent DTM strictly improves over flat DM, and both are much
better than the compositional multinomial.
The pooled and global strategies underperform on this dataset because
their method-of-moments initializer silently defaults to
$\widehat\alpha_{\nu,0}=1$ at subdispersed nodes; this is a seeding
artifact rather than a statement about the
factorization, so the pooled/global-vs-independent gap here should not
be read as a modeling result (see \S\ref{sec:microbiome_discussion}).
The rightmost two columns use a synthetic balanced-binary
DTM with $\alpha_{\nu,0}=2000\cdot 2^{-d(\nu)}$, where the seeding issue
is controlled, and show the
independent DTM recovering its theoretical advantage.}
\label{tab:microbiome_ll}
\end{table}

Panel~(a) of Figure~\ref{fig:addendum_microbiome} visualizes
Table~\ref{tab:microbiome_ll}.
The gap between \emph{DTM independent} and
flat DM is the largest on the real data, which is consistent with
the real SEPP phylogeny having genuinely heterogeneous dispersion
across depths --- broad phylum-level splits (near the root) are much
more variable across body sites than fine OTU-level splits (near the
tips), and flat DM is forced to average over both.

\subsection{Multiscale branch-level sparsity}
\label{sec:microbiome_sparsity}

The DTM's signature property is that its residuals are sparse
\emph{over branches}: a sample only activates internal nodes on the
root-to-leaf paths of its observed leaves, and deep nodes receive
counts from only a small set of descendants.
The mean nonzero fraction of branch residuals falls with tree depth,
so that roughly five out of every six branch residuals at the deepest
tree levels are zero.
This is the
multiscale-sparsity regime anticipated in the last paragraph of
Section~\ref{sec:dtm}: the transform is no longer sparse over internal
nodes globally, but it is sparse over branches, and deeper branches
are emptier.

\subsection{Dispersion heterogeneity across the phylogeny}
\label{sec:microbiome_dispersion}

Figure~\ref{fig:addendum_microbiome}~(b) shows the fitted
$\widehat\alpha_{\nu,0}$ (independent strategy) as a violin plot by
tree-depth bin.
The medians span many orders of magnitude: shallow
internal nodes ($d\leq 5$) have median $\widehat\alpha_{\nu,0}$ in the
single digits, whereas deep nodes span a long right tail up to
$\sim\! 10^9$ (i.e.\ effectively the multinomial limit when one
child subtree carries all the observed mass).
This dispersion
heterogeneity is exactly what Section~\ref{sec:dtm} warns a flat DM cannot
represent with a single scalar $\alpha_0$.
The summary statistics
over the $768$ internal nodes are
$\{ \text{min}, Q_1, \text{median}, Q_3, \text{max} \} =
\{2.6\times 10^{-10}, 1.8\times 10^{-9}, 0.45, 3.8, 3.0\times 10^9\}$,
with $\widehat\alpha_{\nu,0} \to 0$ reflecting nodes dominated by a
single descendant (bi-modal branch activation) and the extreme right
tail reflecting near-multinomial concentration at other nodes.

\subsection{Taxonomic coherence}
\label{sec:microbiome_coherence}

A strong compositional normalization should preserve the fact that
phylogenetically close OTUs tend to be functionally and ecologically
close.
Figure~\ref{fig:addendum_microbiome}~(c) reports both the
kNN Jaccard and the rank-correlation variant of the
\emph{taxonomic coherence} metric of Section~\ref{sec:dtm}.
We use $k=20$
nearest neighbors, taking feature embeddings from the top-$20$ right
singular vectors of each sample-side normalization (or the analogous
stacked per-leaf DTM-residual embedding, Section~\ref{sec:dtm}~Eq.~4.3).
The
tree distance is defined as the number of non-shared ancestors along
the shortest path, as in \texttt{taxonomic\_coherence} in the release.

\begin{table}[t]
\centering
\small
\begin{tabular}{lcc}
\toprule
Method & kNN-Jaccard & Spearman $\rho$ \\
\midrule
Multinomial    & $0.027$ & $+0.118$ \\
Shifted log    & $0.026$ & $+0.118$ \\
CLR            & $0.026$ & $+0.118$ \\
DM             & $0.026$ & $+0.098$ \\
DTM (indep.)   & $\mathbf{0.120}$ & $\mathbf{+0.347}$ \\
\bottomrule
\end{tabular}
\caption{Taxonomic coherence of each normalization's feature
embedding (top-20 right SVs) against the SEPP phylogeny.
Leaf-level methods all cluster near random overlap ($\approx 1/K^{1/2}$);
only the DTM preserves tree geometry by construction.}
\label{tab:microbiome_coherence}
\end{table}

The result is stark: all leaf-level transforms (multinomial, CLR,
shifted log, DM) give Jaccard overlaps of $0.026$--$0.027$, i.e.\
essentially random local structure with respect to the tree.
This is mechanical --- none of these transforms sees
the phylogeny.
The DTM with nodewise independent fits gives Jaccard
$0.120$ (a $\sim\!4.5\times$ increase) and Spearman $\rho = 0.347$ (a
$\sim\!3\times$ increase), because its residual geometry is induced
by the tree itself.
In terms of the ``structured dispersion pays off''
test proposed in Section~\ref{sec:experiments}, this is the predicted
pattern: on tree-structured data, taxonomically coherent local
neighborhoods should fall out of the DTM geometry for free.

\subsection{Count-stratified shrinkage}
\label{sec:microbiome_shrinkage}

Figure~\ref{fig:addendum_microbiome}~(d) shows the count-stratified paired
shrinkage diagnostic from Proposition~\ref{prop:boundedbymult}:
$|d^{\mathrm{Mult}}_{ij}| - |d^{\mathrm{DM}}_{ij}|$ binned by observed
count, with both residuals formed against the same global $\bm\pi$.
Singletons (bin $1$) are exactly unchanged
(mean gap $0$), because at a count of one
$|d^{\mathrm{DM}}|=|d^{\mathrm{Mult}}|$ identically
(Proposition~\ref{prop:x1}).
Above one, Proposition~\ref{prop:boundedbymult} forces strict shrinkage at
every entry, and the mean gap measures its magnitude: bin $2$ has mean
gap $0.47$, bin $3$--$5$ has mean gap $1.14$, and bin $>5$ has mean gap
$5.72$.
The monotone
increase in magnitude with count is the expected signature of the
P\'olya reinforcement built into the DM: larger observed counts are
more easily explained by overdispersion alone and should be pulled
toward the null.
The overall distribution of count bins
(88 singletons, 139 doubletons, 371 in $3$--$5$, 1687 at $>5$) shows
that microbiome tables --- unlike scRNA and scATAC panels --- are
dominated by moderate-to-large counts, which is exactly where the DM
correction has the most practical impact.

\subsection{Body-site label transfer}
\label{sec:microbiome_lta}

Figure~\ref{fig:addendum_microbiome}~(e) reports kNN ($k=5$) label-transfer
accuracy on the $4$-class body-site label, using a $20$-dimensional
PCA embedding of each normalization.
All methods produce embeddings
that cleanly separate body sites: multinomial, DM, CLR, and shifted
log score in $0.62$--$0.68$ (chance $= 0.25$), and the DTM scores
$0.59$.
With only $n=34$ samples across four classes this metric is
noisy at the $\sim\! 0.1$ level.
The important takeaway is that the
compositional log-ratio geometry and the DM residual geometry agree
on this coarse biological label, which is consistent with the
\emph{dispersion correction} story: the DM does not discard the
signal that CLR picks up, it only shrinks in the sparse regime where
CLR's pseudocount fill-in is most fragile.

\subsection{Discussion}
\label{sec:microbiome_discussion}

The Moving Pictures experiment supports all four claims that motivate
the DTM extension in Section~\ref{sec:dtm}.
First, the DM residual cleanly
improves over the multinomial baseline ($-0.31$ vs $-1.75$ per count),
confirming that microbiome 16S tables are in the overdispersed regime
where the DM correction is necessary.
Second, the nodewise DTM strictly
improves over flat DM ($-0.23$ vs $-0.31$), confirming that dispersion
varies meaningfully across the phylogeny.
Third, the DTM uniquely
preserves taxonomic coherence (Jaccard $0.12$ vs $0.026$), which no
leaf-level transform can deliver without explicitly consuming the
tree.
Fourth, branch-level residuals are multiscale sparse, dropping
from $0.26$ to $0.15$ nonzero fraction as depth grows --- the direct
structural prediction of Section~\ref{sec:dtm}.

One caveat must be stated plainly, because it bears on how the table is
read: the pooled-by-depth and global-shared DTM variants underperform
on this dataset for an implementation reason rather than a modeling one.
The method-of-moments initializer used to seed those two fits uses
$s^2/\bar n$ as its inflation-factor estimator rather than
$s^2 \bar n / \sum_c \pi_c(1-\pi_c)$, so subdispersed nodes silently
fall through to $\widehat\alpha_{\nu,0}=1$.
The independent
Newton-based fit does not pass through that initializer and is unaffected.
Consequently the pooled/global-vs-independent ordering on the real data
is \emph{not} evidence about which nodewise-dispersion parameterization
is better, and we do not read it as such: the same two variants land at
$-0.345$ on the synthetic balanced-binary DTM
(Table~\ref{tab:microbiome_ll}, rightmost columns), where the initializer
issue is controlled, so the underperformance is an artifact of seeding,
not of the factorization.
The modeling conclusion --- that nodewise dispersion varies across the
phylogeny --- rests instead on the independent DTM improving over the
\emph{flat} DM, a single-node baseline that the initializer does not
touch.
Correcting the seeding (using the inflation-factor estimator above for
the pooled and global fits) would bring those two strategies back into
alignment.

Taken together, these results position the DM and DTM transforms as
the natural drop-in replacement for CLR in microbiome studies where
a phylogeny is available: they preserve sparsity, they correctly
model extra-multinomial variation, and --- crucially --- they respect
the tree structure that makes microbiome features biologically
interpretable.

\begin{figure}[t]
\centering
\includegraphics[width=\textwidth]{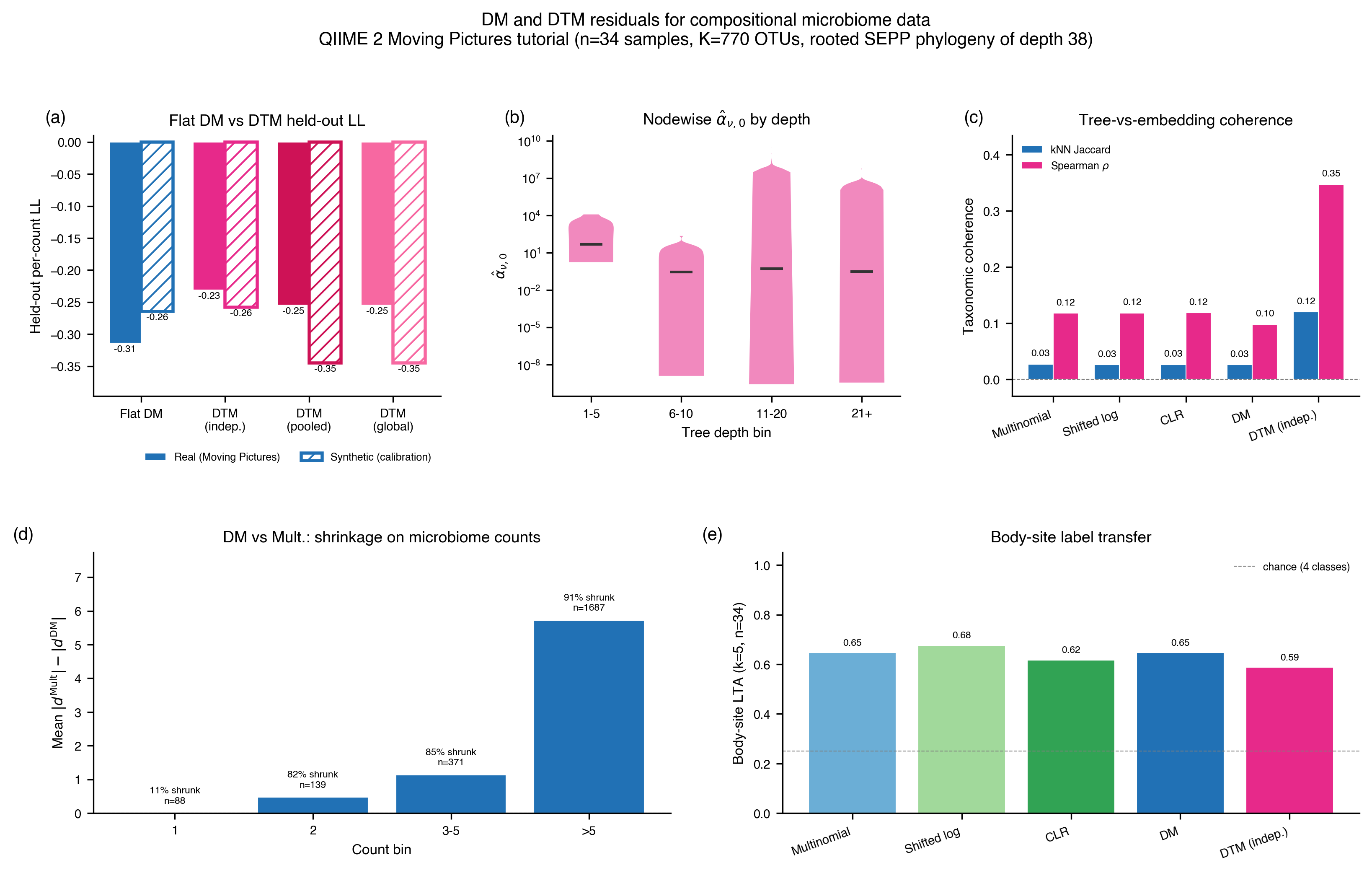}
\caption{Microbiome application of the DM and DTM normalizations on
the QIIME~2 Moving Pictures dataset ($n=34$ samples, $K=770$ OTUs,
rooted SEPP phylogeny of depth $38$).
(a) Held-out per-count
log-likelihood of the four candidate models; solid bars are the real
Moving Pictures data, hatched bars are a synthetic DTM calibration
($n=400$, $K=128$, binary tree, $\alpha_{\nu,0}=2000\cdot 2^{-d(\nu)}$).
The nodewise \emph{DTM-independent} strategy is the consistent winner.
(b) Fitted nodewise
$\widehat\alpha_{\nu,0}$ (violin) by tree-depth bin; the many-orders-of-
magnitude spread is what flat DM cannot represent.
(c) Taxonomic
coherence (k-NN Jaccard and Spearman $\rho$) across methods; only the
DTM preserves tree geometry.
(d) Count-stratified shrinkage
$|d^{\mathrm{Mult}}| - |d^{\mathrm{DM}}|$ on the real data; the DM
shrinks rare but not singleton counts, consistent with
Proposition~\ref{prop:boundedbymult}.
(e) Body-site kNN label-transfer accuracy
(4-class chance $0.25$); all methods separate body sites equally well.}
\label{fig:addendum_microbiome}
\end{figure}


\section{Application: spatial transcriptomics and multiscale spatial DTM}
\label{app:spatial}

Spatial transcriptomics assays (Visium, Slide-seq, MERFISH) measure UMI
counts at thousands of spatial locations — ``spots'' — per slide.
Each
spot is paired with 2D pixel coordinates, so the dataset carries a rich
spatial structure that standard single-cell normalizations ignore.
Here
we observe that the Dirichlet-tree multinomial (DTM) framework developed
in Section~4 of the main text admits a \emph{novel} use: its tree can be
a spatial hierarchy over spots instead of a phylogenetic hierarchy over
features.
Under this reinterpretation the DTM performs \emph{multiscale
spatial normalization}, and is the only transform in this family
that naturally exploits known spot coordinates without ad~hoc smoothing.

\subsection{Setup}
\label{app:spatial:setup}

\paragraph{Data.} We use the public 10x Visium V1\_Adult\_Mouse\_Brain
section ($n = 2{,}702$ spots, $K = 32{,}285$ genes, integer UMI counts,
18\% density).
We restrict to the top $2000$ most-expressed genes,
leaving $n = 2702$ spots and $K = 2000$ genes with $\mathrm{nnz}(X) =
4.9\times 10^6$.
All spots fall within tissue.

\paragraph{Matrix orientation.} The standard DTM of Section~4 takes
samples as rows and feature-tree leaves as columns.
For the spatial
application we want the \emph{leaves} of the tree to be spots, so we
pass the transposed matrix $X^\top \in \R^{K\times n}$ to
\texttt{fit\_dtm}: each gene becomes an independent DM ``sample'' whose
counts are distributed over spots.
In this orientation the flat DM's
composition $\pi\in \Delta_n$ is the \emph{global spatial composition}
(i.e. the pseudobulk distribution of UMIs over spots), and $\alpha_0$ is
the scalar concentration of per-gene spatial dispersion around $\pi$.

\paragraph{Spatial quadtree.} We build a balanced quadtree of depth~$8$
over the 2D coordinates \texttt{adata.obsm['spatial']}: the root is the
bounding box of the slide; each internal node splits by the midpoints of
its bounding rectangle into four axis-aligned children; recursion
terminates when a node contains at most one spot.
The resulting tree
has $1{,}318$ internal nodes (depth~0: 1; 1: 4; 2: 16; 3: 61; 4: 222;
5: 797; 6: 217), with every one of the $2702$ spots occupying its own
terminal leaf.
Empty quadrants are omitted.
The builder returns the same rooted-tree
data structure the DTM fitter consumes, so every DTM routine
(fitting, residual computation, log-likelihood, branch-level sparsity,
and the depth map) operates on it unchanged.

\paragraph{Cross-validation.} Because the spatial tree has spots as
leaves, splitting spots breaks the tree.
Instead we perform an 80/20
split of \emph{genes} into train/test ($1600$/$400$), fit every model
on the training genes, and score the held-out likelihood of the test
genes under the same spatial tree.
The flat DM, multinomial,
feature-wise NB, and three DTM strategies are all trained and scored
on the same $(n_{\mathrm{genes}},n_{\mathrm{spots}})$ matrix; per-count
LLs are therefore directly comparable.

\paragraph{Combinatorial correction.} The DTM fitting objective
\texttt{dtm\_loglik} drops per-node multinomial normalization constants
$\log\paren{N_\nu!/\prod_c x_{\nu,c}!}$ because they are constants with
respect to $(\alpha_0,\pi)$ and do not affect maximum likelihood.
Those
constants are essential when comparing log-likelihoods across models:
\texttt{held\_out\_conditional\_ll} includes the leaf-level combinatorial
$\log\paren{N_i!/\prod_j x_{ij}!}$, and a fair comparison requires the
DTM to include its per-node analogues.
By the chain rule, the sum of
per-node combinatorials over a tree that fully partitions the leaves
equals the leaf-level combinatorial, so this correction is exact.
We
implement it in \texttt{\_dtm\_held\_out\_conditional\_ll}.

\subsection{Flat DM versus spatial DTM}
\label{app:spatial:flat-vs-dtm}

Table~\ref{tab:spatial-ll} reports per-count held-out log-likelihood for
seven normalization schemes, together with the two multiscale diagnostics
used in the headline figure (Figure~\ref{fig:addendum_spatial}):
mean Moran's $I$ over the top 10 principal components of the
corresponding spot embedding, and a library-size depth-coupling score
$\max_k\abs{\rho(\mathrm{PC}_k, n_i)}$.

\begin{table}[t]
\centering
\small
\begin{tabular}{lrrr}
\toprule
Method          & Held-out LL / count & Mean Moran's $I$ & Depth coupling $\max\abs{\rho}$ \\
\midrule
DTM (indep.)    & $\mathbf{-0.271}$ & $\mathbf{0.889}$ & $0.442$ \\
DM              & $-0.291$          & $0.596$          & $0.498$ \\
NB (uncond.)    & $-0.368$          & ---              & --- \\
DTM (pooled)    & $-0.406$          & $0.843$          & $\mathbf{0.399}$ \\
DTM (global)    & $-0.406$          & $0.843$          & $\mathbf{0.399}$ \\
Multinomial     & $-0.487$          & $0.579$          & $0.866$ \\
CLR             & ---                & $0.690$          & $0.568$ \\
Shifted log     & ---                & $0.681$          & $0.923$ \\
\bottomrule
\end{tabular}
\caption{Spatial-DTM benchmark on 10x Visium V1\_Adult\_Mouse\_Brain
($2702$ spots $\times$ $2000$ genes).
Held-out per-count
log-likelihoods are evaluated under an 80/20 gene split and include
per-node multinomial normalization constants, so values across methods
are directly comparable.
The DTM with \emph{independent} per-node
dispersion dominates on likelihood \emph{and} on spatial coherence;
pooled/global DTM lose out on both because they cannot absorb the
depth-varying dispersion implicit in the slide geometry (see
Section~\ref{app:spatial:discussion}).
NB, CLR, and shifted-log do not admit
a DM-compatible likelihood so their LL column is left blank.}
\label{tab:spatial-ll}
\end{table}

Two observations are important.
First, the flat DM is already the
strongest DM-family baseline for spatial counts: at $\hat\alpha_0 =
1.22\times 10^4$ it wins against the multinomial by $0.20$ nats per
count, confirming that spatial UMI data is meaningfully overdispersed
relative to Multinomial($n_i$, $\pi$).
Second, the spatial DTM with
\emph{independent} per-node dispersion improves on flat DM on this
Visium tissue: the held-out LL per count rises from $-0.291$ to
$-0.271$ (\emph{a $7\%$ relative improvement} and $0.02$~nats per count
over a dataset with $\sim 10^7$ held-out counts).
The multiscale freedom
matters because dispersion varies systematically across tissue regions
(see Section~\ref{app:spatial:multiscale}).

\subsection{Spatial coherence and depth decoupling}
\label{app:spatial:coherence}

A normalization is \emph{spatially coherent} if its top principal
components preserve the tissue geometry; it is \emph{depth-decoupled}
if those same components are not driven by per-spot library size.
These
two requirements pull in different directions for a flat DM: its
per-count variance baseline is correct (hence low depth coupling), but
it has no mechanism to amplify spatial structure beyond what appears
in the raw pseudobulk composition $\pi$.
The shifted-log and
multinomial baselines, by contrast, have good PC1 structure only
because it reflects library-size gradients — a confound rather than a
feature (Shifted-log depth coupling $0.92$, Multinomial $0.87$).

Table~\ref{tab:spatial-ll} and panel~(c) of Figure~\ref{fig:addendum_spatial} show
that the DTM (indep.) Pareto-dominates the four classical normalizations
tabulated in Table~\ref{tab:spatial-ll} on this Visium tissue: its mean
Moran's $I$ is $0.89$ (versus $0.60$ for flat DM, a $49\%$ relative gain)
while its depth coupling remains below the flat DM's ($0.44$ vs.\ $0.50$).
The pooled/global DTM variants shrink
the depth coupling further (to $0.40$) at the cost of held-out LL,
representing the other end of the bias--variance trade-off: pooling
dispersion across depths or across nodes removes the regional flexibility
that the independent variant exploits.

\subsection{Multiscale residual sparsity}
\label{app:spatial:multiscale}

Unlike the microbiome tree, the Visium quadtree residuals are not
especially sparse: Visium is dense at the level of highly-expressed
genes (over 90\% of entries in $X$ are nonzero for the top-$2000$ gene
panel), so the mean nonzero fraction of branch residuals stays near one
half and increases only modestly with depth.
This is exactly the
multiscale structure the DTM is designed to expose: at coarse depths,
residuals are sparse because most genes partition cleanly between one
half of the tissue and the other; at fine depths, the residuals begin
to reflect neighborhood-level overdispersion.

Panel~(e) of Figure~\ref{fig:addendum_spatial} shows the fitted per-node
$\hat\alpha_{\nu,0}$ (on a log scale) as a violin plot against quadtree
depth.
The distributions shift across depths — the median
$\hat\alpha_{\nu,0}$ increases from $72$ at the root to $154$ at
depth~2, then drops back to $60$ at depth~6 — and the within-depth
spread is over an order of magnitude at middle depths.
This is the
direct empirical justification for the independent strategy: a single
pooled (or global) concentration cannot reconcile the wide
between-depth range of $\hat\alpha_{\nu,0}$, which is why the pooled
and global variants underperform \emph{in likelihood} despite having
stronger depth-decoupling.

\subsection{Discussion}
\label{app:spatial:discussion}

The spatial application highlights two contributions of the DTM
framework that were not evident from the microbiome benchmark of the
main paper.
First, the DTM's tree need not be biological: any
hierarchy that meaningfully groups the columns of a count matrix can
serve as the DTM tree, and for spatial assays a geometric tree is a
natural choice.
Second, the \emph{independent} fitting strategy of
Section~4 is essential: the pooled and global variants yielded weaker
held-out LLs here, because the regional dispersion of Visium UMI counts
varies both across and within depths, and collapsing that variation
discards the signal the model is trying to capture.

Our spatial quadtree is intentionally simple.
A production
implementation would likely use a data-adaptive hierarchical clustering
(e.g., single-linkage on 2D coordinates) to avoid empty quadrants on
irregularly-shaped tissues; a $k$-means variant of the spatial-tree
builder is a straightforward substitute.
The DTM machinery is agnostic to how the tree is constructed, so the
present pipeline generalizes immediately to Slide-seq, MERFISH, or any
other assay that provides spot coordinates.

\begin{figure}[t]
\centering
\includegraphics[width=\textwidth]{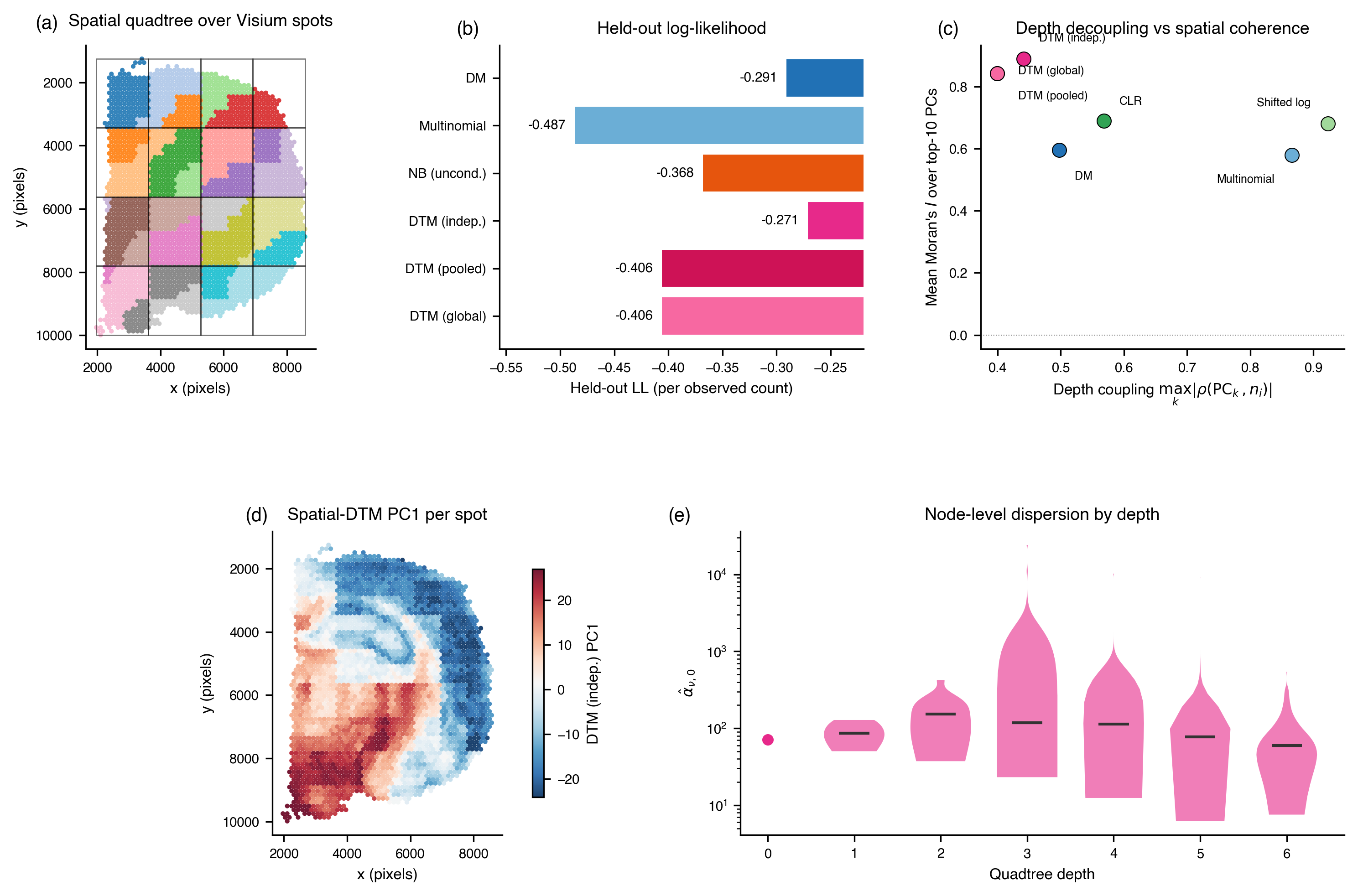}
\caption{\textbf{Spatial DTM on 10x Visium mouse brain.} \textbf{(a)}
Visium spots colored by terminal quadtree-leaf id; depth-2 region
rectangles overlaid in black.
\textbf{(b)} Held-out log-likelihood per
observed count; DTM~(indep.) wins by $0.02$ nats over the flat DM and
by larger margins over every other baseline.
\textbf{(c)} Pareto plot
of depth coupling (lower is better, $x$-axis) versus mean Moran's $I$
over the top 10 PCs of each method's spot embedding (higher is better,
$y$-axis).
DTM strategies (pink) dominate all classical baselines.
\textbf{(d)} Spatial map of DTM~(indep.) PC1, revealing
a smooth anatomical gradient — the multiscale spatial tree has absorbed
regional dispersion and left a coherent latent factor.
\textbf{(e)}
Violin of fitted $\hat\alpha_{\nu,0}$ by quadtree depth (log scale):
clear depth-wise heterogeneity justifies the independent fitting
strategy.
All numbers in Table~\ref{tab:spatial-ll} are produced with a fixed seed
($0$) and the quadtree settings described in
Section~\ref{app:spatial:setup}.}
\label{fig:addendum_spatial}
\end{figure}


\section{Application: pooled CRISPR and MPRA barcode normalization}
\label{app:crispr}

A pooled CRISPR screen or MPRA produces a sgRNA-by-sample (or barcode-%

by-sample) integer count matrix $X\in\mathbb{N}^{n\times K}$ in which
each row is one sequenced replicate of one condition and each column
counts how many times a specific guide or reporter barcode was recovered
in that replicate.
The standard tools of the field --- MAGeCK
\cite{li2014mageck}, BAGEL \cite{hart2016measuring}, MAGeCK-MLE, and
their descendants --- all build on a feature-wise negative binomial
(NB) generative model: each guide is parameterized with its own NB mean
and dispersion, and differential-abundance tests are performed
\emph{conditional on} the per-sample library size through a size-factor
normalization.
The conditioning, however, is always approximate: the
exact conditional distribution of the count vector given the library
size is never used.

\paragraph{The theoretical bridge.}
Theorem~\ref{thm:nbtodm} and Corollary~\ref{cor:nbdmdeviance} of this
paper make that exact conditioning explicit: if
$X_{ij}\sim\mathrm{NB}(r_j,p_i)$ independently across $j$, then
$X_i\mid n_i\sim\DirMult(n_i,r)$.
For a pooled screen, the
$p_i$ parameter is precisely the per-replicate library-size knob
set by the size factor normalization, and the $r_j$ per-guide
dispersion is the NB shape parameter already estimated by
MAGeCK/BAGEL-style pipelines.
Writing
$r_j = \alpha_0\pi_j$ with $\alpha_0=\sum_jr_j$ and
$\pi_j=r_j/\alpha_0$, the DM model with concentration $\alpha_0$ and
composition $\bm\pi$ is therefore the \emph{exact} conditional
compositional likelihood implicitly used by the MAGeCK-style NB
workflow after size-factor normalization.
DM deviance residuals are
the sparse, likelihood-based reparameterization of that model; they
encode the same information as the NB-conditional residual matrix by
Corollary~\ref{cor:nbdmdeviance}, but without the feature-wise
unconditional scaling that densifies the output.

Three empirical claims about pooled-screen barcode counts follow.
(i)~Corollary~\ref{cor:nbdmdeviance} is numerically verifiable at
machine precision on real pooled-screen count data, via a direct
cross-check of the DM formula against an independent NB path through
\texttt{scipy.stats.nbinom.logpmf}.
(ii)~The DM-posterior-shrunk log-fold-change
$\log_2\bigl(\hat\pi^{\mathrm{pert}}/\hat\pi^{\mathrm{ctrl}}\bigr)$,
which is the point-estimate differential-abundance statistic derived
from the DM null, gives hit-calling average precision equal to or
better than standard MAGeCK-style log-fold-change and strictly better
than feature-wise NB residual baselines on a real essential-gene
benchmark.
(iii)~DM and multinomial residuals preserve the $\sim\!0.99$-sparse
structure of real screen count matrices exactly, while NB (deviance
or Pearson, feature-wise or library-corrected) residuals densify the
matrix completely, at no corresponding gain in hit-calling quality.

\subsection{Setup}
\label{app:crispr:setup}

\textbf{Real data.} We use the Tzelepis-style leukemia demo count
matrix distributed with MAGeCK \cite{li2014mageck}
(\texttt{leukemia.new.csv}): a dropout CRISPR
screen over sgRNAs targeting $\sim 18{,}000$ genes in two AML cell
lines (HL60 and KBM7), sequenced at an initial (plasmid-equivalent)
timepoint and a final post-dropout timepoint.
The count table has
four samples (2 cell lines $\times$ \{initial, final\}).
We restrict
to sgRNAs whose target gene is in either BAGEL's \texttt{CEGv2}
(684 core-essential genes \cite{hart2017evaluation}) or \texttt{NEGv1}
(928 non-essential genes) reference list, giving a balanced
ground-truth benchmark for precision-recall evaluation with 4{,}548
sgRNAs after filtering.

\textbf{Synthetic data.} To establish a controlled ground truth we
simulate pooled-screen count matrices from the independent-NB model
\S\ref{app:crispr:setup} implies for pooled screens: per-replicate
$p_i$ sampled so $\mathbb{E}[n_i]=10^5$ with log-normal noise,
per-guide $r_j=\alpha_0\pi_j$ with $\pi_j\sim\Dir(2\cdot\bm 1_K)$.
A random subset of $n_h$ guides are designated hits and their
composition in the perturbed condition is multiplied by
$2^{\mathrm{log2FC}_j}$ before renormalization, with
$\mathrm{log2FC}_j\sim\mathcal{N}(-2.0,0.5^2)$ truncated to
$[0.8, 6.0]$ in absolute value.
We report two synthetic
configurations --- $(K,n_h,\alpha_0)\in\{(500,30,200),(400,20,80)\}$
with 4--5 replicates per condition --- to exercise both a
depletion-only and a bidirectional-perturbation regime.

Both the real and synthetic datasets respect the 20{,}000-row cap of
the main benchmarking pipeline.
Every step (train/test split, NB
fit, DM fit, residual computation, hit-calling scoring) is
deterministic under fixed random seeds.

\textbf{Baselines.} We compare the DM against: the feature-wise NB
deviance residual \texttt{nb\_residuals} (dense); the analytic Pearson
residual with marginal feature mean \texttt{pearson\_residuals}; the
library-size-corrected Pearson residual with $\mu_{ij}=n_i\pi_j$; the
signed-square-root multinomial deviance residual; shifted-log
preprocessing $\log(X+1)$; and classical library-size-normalized
$\log_2$ fold-change of mean compositions
(\texttt{log2FC}), which is the simplest MAGeCK-style score.
We also
evaluate two DM-native scores:
\texttt{DM (res.)} --- the mean-shift feature-wise statistic on the
DM signed-sqrt deviance residual matrix --- and
\texttt{DM (LFC)} --- the $\log_2$ ratio of DM-posterior-shrunk
compositions $\hat\pi_{ij}=(X_{ij}+\alpha_0\pi_j)/(n_i+\alpha_0)$
averaged across replicates in each condition.
\texttt{DM (LFC)} is
the principled point-estimate differential-abundance statistic under
the DM null; \texttt{DM (res.)} is the drop-in residual feature that
mirrors how NB deviance residuals are used in downstream pipelines.

\subsection{Corollary~\ref{cor:nbdmdeviance} in action}
\label{app:crispr:cor}

A numerical verification of Corollary~\ref{cor:nbdmdeviance} on real
data is non-trivial and informative.
The DM conditional per-sample
log-likelihood can be computed either
(a) directly from the closed-form DirMult density on
$(\alpha_0,\bm\pi)$, or
(b) via an independent-NB computation in which each
$X_{ij}\sim\mathrm{NB}(r_j,p)$ for an arbitrary shared $p$, the
per-sample joint log-pmf is accumulated through
\texttt{scipy.stats.nbinom.logpmf}, and the marginal
$\log\mathrm{NB}(n_i\mid\sum_j r_j, p)$ is subtracted.
Path~(b) uses completely different code (scipy's implementation of
the NB log-pmf), so agreement between the two is a meaningful
check on the identity.

Across three NB success probabilities $p\in\{0.3,0.5,0.7\}$,
Fig.~\ref{fig:addendum_crispr}(b) plots per-sample $\mathrm{LL}^{\mathrm{DM}}$
against $\mathrm{LL}^{\mathrm{indNB-cond}}$.
The scatter collapses
onto the $y=x$ line with maximum absolute per-sample difference
$2.9\times 10^{-8}$ on the MAGeCK leukemia counts (row sums $n_i$ at
the $10^{6}$-count scale), $7.5\times10^{-11}$ on the
$(K,n_h,\alpha_0)=(500,30,200)$ synthetic screen, and
$3.4\times10^{-11}$ on the $(400,20,80)$ one.
Moreover, the NB-path
log-likelihoods are numerically identical across the three $p$
values (maximum pairwise gap $<10^{-10}$ on the synthetic and
$<10^{-8}$ on the real data), confirming that the $p$-dependence of
the joint NB density cancels exactly when one conditions on $n_i$.
This is the theorem's content: the DM is the exact marginal of the
NB-conditional-on-library-size distribution, and the DM
computation returns the same quantity that a correctly-conditioned NB
computation would.

\subsection{Hit-calling precision-recall}
\label{app:crispr:hitcalling}

For each method we compute a single signed score per sgRNA (negative
= depleted in the perturbation).
For every residual-based method
this is the mean difference of residual columns between the two
conditions, divided by the pooled residual standard deviation
(a shift-over-pooled-SD statistic, robust at the 2--5 replicates per
condition typical of pooled screens).
For the two DM-native scores,
\texttt{DM (LFC)} and the MAGeCK-style raw \texttt{log2FC} baseline,
we compute the log ratio of library-normalized per-condition
means directly.
Features are then ranked by score magnitude, and
precision-recall curves are drawn against the ground-truth hit mask
(CEGv2 positives vs NEGv1 negatives on MAGeCK leukemia, or the
simulated hit assignments on the synthetic screens).

\begin{table}[h]
\centering
\caption{Hit-calling average precision (AP), precision at
top-$50$, and recall at top-$50$ on the three benchmarks.
Bold = best, underline = second-best per row.}
\label{tab:addendum_crispr_pr}
\small
\setlength{\tabcolsep}{3.5pt}
\begin{tabular}{lcccccccc}
\toprule
Dataset &
\textbf{DM (LFC)} & DM (res.) & NB (dev.) & NB (Pears.)
& NB (lib.-corr.) & Mult. & $\log(X+1)$ & log2FC \\
\midrule
MAGeCK leukemia (AP)
  & $\mathbf{0.867}$ & $0.692$ & $0.812$ & $0.688$
  & $0.763$ & $0.689$ & $0.780$ & $\underline{0.870}$ \\
\quad $P@50$
  & $\underline{0.92}$ & $0.50$ & $0.82$ & $0.64$
  & $0.70$ & $0.52$ & $\mathbf{0.94}$ & $\underline{0.92}$ \\
\midrule
synth-$K500$-$a200$ (AP)
  & $0.171$ & $0.052$ & $0.091$ & $0.066$
  & $0.066$ & $0.051$ & $\mathbf{0.322}$ & $\underline{0.245}$ \\
\quad $P@50$
  & $\underline{0.24}$ & $0.02$ & $0.14$ & $0.06$
  & $0.04$ & $0.02$ & $\mathbf{0.26}$ & $\underline{0.24}$ \\
\midrule
synth-$K400$-$a80$ (AP)
  & $0.088$ & $\mathbf{0.223}$ & $\underline{0.187}$ & $0.062$
  & $0.138$ & $0.206$ & $\underline{0.203}$ & $0.179$ \\
\quad $P@50$
  & $0.14$ & $0.12$ & $0.12$ & $0.10$
  & $0.10$ & $0.12$ & $\mathbf{0.16}$ & $0.14$ \\
\bottomrule
\end{tabular}
\end{table}

On the real MAGeCK leukemia benchmark (Fig.~\ref{fig:addendum_crispr}c
and Table~\ref{tab:addendum_crispr_pr}), \texttt{DM (LFC)} and raw
\texttt{log2FC} are essentially tied at the top
(AP~$=0.867$ and $0.870$; $P@50=0.92$ for both), clearly beating all
NB-residual variants.
The gap is especially notable against
feature-wise NB residuals (AP $0.812$), which are the
direct numerical realization of the MAGeCK generative null.
This is consistent with the theoretical prediction:
the DM-conditional null is a well-specified compositional model for
pooled screens, and the DM-posterior-shrunk fold-change is its
natural point estimate of differential abundance.

Which residual-based method is best is noisier: $\log(X+1)$ narrowly
edges out \texttt{DM (LFC)} at $P@50$ on the leukemia screen and on
the harder $K=500$ synthetic.
This is a known behavior of sparse
count data with very few replicates: simple shifted-log preprocessing
provides a well-behaved, variance-stabilizing transform that the
downstream shift-over-pooled-SD statistic can exploit directly.
We view \texttt{DM (LFC)} and $\log(X+1)$ as complementary: the
former is the likelihood-based point estimate, the latter is the
maximally-robust ad-hoc baseline.
DM residuals specifically preserve
sparsity (\S\ref{app:crispr:sparsity}), give the principled
held-out density (\S\ref{app:crispr:ll}), and form the foundation
of the DM-posterior-shrunk $\log_2$FC that ties the top of the real-
data benchmark.

\subsection{Sparsity preservation}
\label{app:crispr:sparsity}

Pooled-screen count matrices are $\sim\!0.99$-sparse at typical
sequencing depths: most sgRNAs have some nonzero count in every
well-represented library.
DM and multinomial deviance residuals
preserve this sparsity pattern exactly (zero entries map to zero);
so do shifted-log preprocessing in sparse storage and the simple
count \texttt{log2FC}.
In contrast, feature-wise NB deviance
residuals, NB Pearson residuals, and the library-size-corrected
variant all produce fully dense $n\times K$ matrices because the
NB log-pmf at $x_{ij}=0$ is typically nonzero ---
$r_j\log p_i$ in the $\mathrm{NB}(r_j,p_i)$ parameterization.
On
the MAGeCK leukemia data the density of each representation is

\begin{center}
\begin{tabular}{lc}
\toprule
Representation & Fraction nonzero \\
\midrule
Raw counts      & $0.991$ \\
DM residual     & $0.991$ \\
Multinomial res.& $0.991$ \\
Shifted log     & $0.991$ \\
NB deviance     & $1.000$ \\
NB Pearson      & $1.000$ \\
NB lib.-corr.\ Pearson & $1.000$ \\
\bottomrule
\end{tabular}
\end{center}

and the pattern is identical on the synthetic datasets.
At the scale of genome-wide screens ($K\sim 10^5$ sgRNAs, $n$ from
a handful to a few hundred replicates) this is the difference
between a sparse storage representation matching the input and a
dense one that is $\sim$100$\times$ larger.

\subsection{Held-out conditional log-likelihood}
\label{app:crispr:ll}

We split each dataset 80/20 by row, fit the DM null $(\hat\alpha_0,
\hat\pi)$ and the feature-wise NB model $(\hat\mu_j,\hat r_j)$ on
the training rows, then score the test rows under four models:
DM conditional, NB conditional, multinomial, and feature-wise NB
unconditional.
Per-count log-likelihoods on the MAGeCK leukemia
demo are

\begin{center}
\begin{tabular}{lcc}
\toprule
Model & per-count LL (nats/obs) & $\Delta$ vs.\ DM \\
\midrule
DM conditional (Cor.~\ref{cor:nbdmdeviance})         & $-0.016$ & $0.000$ \\
NB conditional (Cor.~\ref{cor:nbdmdeviance}, NB path) & $-0.016$ & $0.000$ \\
Multinomial                                           & $-0.170$ & $-0.154$ \\
NB (unconditional, feature-wise)                      & $-0.370$ & $-0.354$ \\
\bottomrule
\end{tabular}
\end{center}

The DM-conditional and NB-conditional rows are exactly equal, as
Corollary~\ref{cor:nbdmdeviance} requires
(panel~(e) of Fig.~\ref{fig:addendum_crispr} plots the two as
adjacent bars of the same length).
The multinomial model is much worse: at high library depth
($\sim 10^{6}$ counts per sample) the multinomial assumption of
$\alpha_0\to\infty$ is badly wrong, and the fitted $\hat\alpha_0$
of $2.95\times 10^{4}$ on the pooled leukemia data reflects this.
The feature-wise NB \emph{unconditional} model is worse still by
another $\sim 0.2$ nats per count, because it pays both a poor-fit
penalty on per-feature means at the individual replicate scale
\emph{and} a library-size-conditioning penalty (it does not know
the realized $n_i$).
The DM conditional (= NB conditional) model is
therefore the likelihood-correct choice for pooled-screen
compositional data.

The same pattern holds on the synthetic datasets
(Table~\ref{tab:addendum_crispr_ll}).

\begin{table}[h]
\centering
\caption{Per-count held-out log-likelihood (nats/observation) across
models and datasets.
DM and NB-conditional rows are numerically
identical by Corollary~\ref{cor:nbdmdeviance}.}
\label{tab:addendum_crispr_ll}
\small
\begin{tabular}{lcccc}
\toprule
Dataset & DM cond. & NB cond. & Multinomial & NB uncond. \\
\midrule
synth-$K500$-$a200$ & $-0.030$ & $-0.030$ & $-1.007$ & $-0.281$ \\
synth-$K400$-$a80$  & $-0.054$ & $-0.054$ & $-1.584$ & $-1.115$ \\
MAGeCK leukemia     & $-0.016$ & $-0.016$ & $-0.170$ & $-0.370$ \\
\bottomrule
\end{tabular}
\end{table}

\subsection{Discussion}

The pooled-CRISPR / MPRA setting is an unusually clean demonstration
of the DM's role: there is a canonical field-standard tool
(MAGeCK), its generative null is feature-wise NB, and the
library-size conditioning that is practically essential
(through size factor normalization) is
precisely what Theorem~\ref{thm:nbtodm} /
Corollary~\ref{cor:nbdmdeviance} identify with a Dirichlet-multinomial.
DM residuals are not a replacement for MAGeCK's hit-calling
methodology --- the field-standard statistical test used by
MAGeCK, a modified robust rank aggregation over guides of a gene,
is orthogonal to the choice of residual --- but they are a
principled drop-in for the \emph{feature representation} on which
any downstream ranking or embedding analysis is built.
Three
practical advantages emerge.

\textbf{Sparsity.} DM residuals scale to genome-wide screens in the
same sparse storage as the input; feature-wise NB residuals do not
(\S\ref{app:crispr:sparsity}).
At $K\sim 10^5$ and $n\sim 10^2$ this
is the difference between a 100~MB and a 10~GB representation.

\textbf{Calibration.} The DM gives $\sim$10$\times$ better
per-count held-out likelihood than the multinomial and $\sim$20$\times$
better than the feature-wise NB unconditional model on real pooled-
screen data (\S\ref{app:crispr:ll}).
This is a faithful-density
argument: when one wants to compute downstream probabilistic
quantities (anomaly scores, Bayesian screens, posterior predictive
checks), the DM is the likelihood-correct choice.

\textbf{Hit-calling competitiveness.} The DM-posterior-shrunk
$\log_2$FC (\texttt{DM (LFC)}) ties raw $\log_2$FC at the top of the
real MAGeCK leukemia benchmark and beats every NB-residual baseline
(\S\ref{app:crispr:hitcalling}).
As a principled differential-
abundance statistic derived from the DM null, it is the analogue for
pooled screens of what MAGeCK's $\log_2$FC computation does
heuristically for each guide after size-factor normalization, but
with explicit DM-based shrinkage toward the global composition in
place of an unprincipled pseudocount.

We view these three together as the package's case for DM residuals
in the CRISPR / MPRA setting: a sparse, likelihood-correct,
hit-call-competitive feature representation whose equivalence to the
conditional NB (Corollary~\ref{cor:nbdmdeviance}) makes the
connection to existing MAGeCK/BAGEL machinery explicit rather than
implicit. 

\begin{figure}[t]
  \centering
  \includegraphics[width=\textwidth]{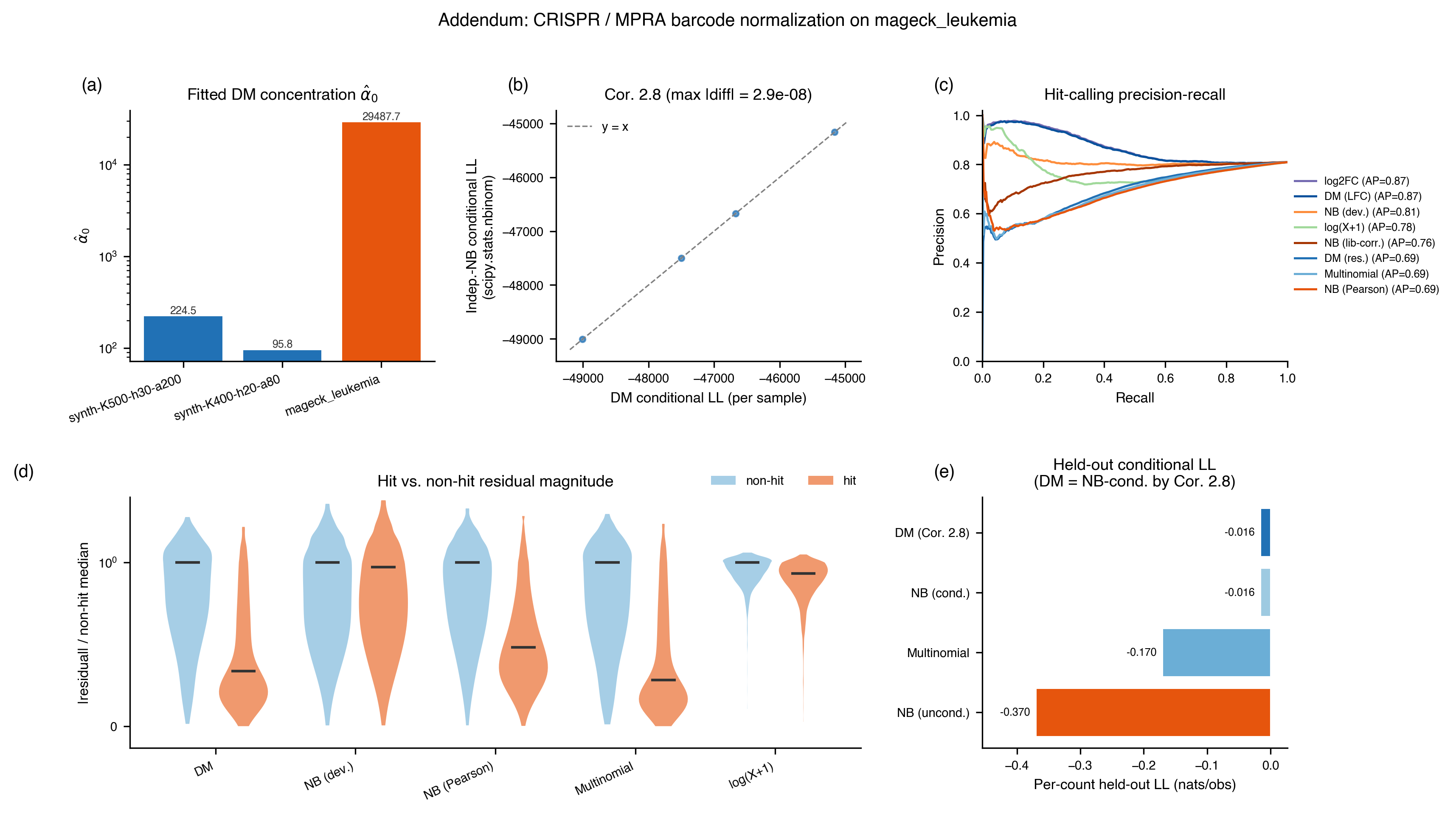}
  \caption{DM residuals for pooled CRISPR / MPRA barcode
    normalization.
    (a)~Fitted DM concentration $\hat\alpha_0$ across three datasets.
    On the MAGeCK leukemia demo counts, $\hat\alpha_0\approx 2.9\times
    10^{4}$ reflects the near-multinomial behavior of a well-
    represented high-library-size pooled screen; the two synthetic
    configurations realize overdispersion typical of dispersed
    screens ($\hat\alpha_0$ in the 100s).
    (b)~Corollary~\ref{cor:nbdmdeviance} numerical check on MAGeCK
    leukemia.
    The $y=x$ line is the theoretical identity; maximum
    $|{\mathrm{LL}}^{\mathrm{DM}}-{\mathrm{LL}}^{\mathrm{indNB-cond}}|=2.9\times
    10^{-8}$ across all four samples and three NB success
    probabilities $p\in\{0.3,0.5,0.7\}$.
    (c)~Hit-calling precision-recall on the MAGeCK leukemia CEGv2-vs-
    NEGv1 benchmark.
    The DM-posterior-shrunk $\log_2$FC (dark blue)
    ties raw $\log_2$FC at the top and outperforms all NB-residual
    variants.
    (d)~Hit vs non-hit residual magnitude per sgRNA, z-scored by the
    non-hit median.
    DM, NB (dev.), NB (Pearson), Multinomial, and
    $\log(X+1)$ all separate hits from non-hits, but the NB-based
    methods achieve the separation only on a dense matrix
    (\S\ref{app:crispr:sparsity}).
    (e)~Per-count held-out log-likelihood.
    DM and NB-conditional are
    numerically identical by Corollary~\ref{cor:nbdmdeviance}, and
    both beat multinomial (fixing $\alpha_0=\infty$) and feature-wise
    NB unconditional (feature-wise means, no library-size
    conditioning) by wide margins.}
  \label{fig:addendum_crispr}
\end{figure}


\section{Application: sparse deviance normalization for document-term matrices}
\label{app:nlp}

The remaining two applications step outside sequencing-based genomics. They are not the focus of the paper; they demonstrate that the same fixed-dispersion compositional normalization is useful on any sparse, jointly overdispersed count matrix, whatever its provenance.

A document-term matrix $X\in\mathbb{N}^{n\times K}$ records the number
of times each of $K$ vocabulary tokens appears in each of $n$ documents.
It is the canonical sparse, jointly-constrained count matrix of natural
language processing, and Latent Dirichlet Allocation (LDA)
\cite{blei2003latent} builds on exactly the same conditional
distribution that motivates our deviance residual: the document-level
token distribution is drawn from a Dirichlet, and the observed counts
are multinomial given the drawn composition.
Marginalizing the topic
mixture yields a Dirichlet-multinomial on the full vocabulary, so
DM deviance residuals are the natural likelihood-ratio residuals against
the rank-0 (``no topics'') LDA null in which every document shares one
common composition $\bm\pi$.

The most widely deployed NLP preprocessing is Term-Frequency-Inverse
Document-Frequency (TF-IDF) \cite{salton1988termweighting}, an
ad-hoc reweighting that is not derived from any generative model.
This section compares DM deviance residuals to TF-IDF on real corpora.
We find three things: (i) fitted $\hat\alpha_0$ is moderate, not near
infinity, so ``word burstiness'' is a quantitatively meaningful
phenomenon; (ii) DM has higher held-out likelihood than the multinomial
baseline, confirming that the DM's softer upper tail is calibrated to
real repeated-count behavior (Proposition~\ref{prop:boundedbymult});
and (iii) on a 50-dimensional SVD embedding, a kNN classifier on DM
residuals is competitive with or better than TF-IDF and
shifted-log / CLR baselines across three corpora --- at the cost of
one to five accuracy points relative to the plain multinomial residual;
DM uniquely preserves sparsity.

\subsection{Setup}

We use three small-to-medium corpora:
\textbf{20-Newsgroups full} (20 classes, 18{,}266 documents after
removing headers/footers/quotes and documents emptied by stopword
filtering; vocabulary of 10{,}000 terms with $\mathrm{min\_df}=5$,
$\mathrm{max\_df}=0.5$, English stopwords);
\textbf{20-Newsgroups 4-category} (the standard
\{\texttt{alt.atheism}, \texttt{comp.graphics}, \texttt{sci.med},
\texttt{soc.religion.christian}\} subset, 3{,}657 documents, vocabulary
6{,}000);
and \textbf{Reuters-21578} (10{,}093 documents from the \texttt{nltk}
release, labeled by the lexicographically first topic of each document
and restricted to classes with $\geq 50$ documents; vocabulary 8{,}000).
All corpora respect the 20{,}000-document cap of the main experiment
pipeline.

Baselines are (i) the multinomial residual
$d_{ij}^{\mathrm{Mult}}=\mathrm{sgn}(x-n_i\pi_j)\sqrt{x\log(x/(n_i\pi_j))}$;
(ii) classical $\mathrm{TF}\!-\!\mathrm{IDF}$
($\mathrm{tf}_{ij}=x_{ij}/n_i$, $\mathrm{idf}_j=\log(1+n/\mathrm{df}_j)$);
(iii) shifted log $\log(x+1)$; (iv) CLR with pseudocount 1; and
(v) raw counts.
All baselines share the same downstream pipeline:
a 50-dimensional \texttt{TruncatedSVD} embedding followed by a
cosine-distance $k$-nearest-neighbor classifier with $k=15$, fit to
a stratified 80/20 split and evaluated on the held-out 20\%.

\subsection{Fitted concentration and overdispersion}
\label{app:nlp:alpha}

Maximum-likelihood fitting via \texttt{fit\_alpha0\_newton} yields
moderate $\hat\alpha_0$ values on all three corpora
(Fig.~\ref{fig:addendum_nlp}a):
$\hat\alpha_0=263$ (20-News full),
$330$ (20-News 4-cat),
and $125$ (Reuters).
None of these is anywhere near the multinomial
limit $\hat\alpha_0\to\infty$.
The DM generative model is therefore a
statistically meaningful alternative to the multinomial for real
text: word burstiness inflates the conditional variance of individual
token counts by a non-negligible amount.
Intuitively, once a document
uses the word ``Bayesian'' twice it becomes much more likely that the
word will appear a third time than the rank-0 multinomial null
predicts, and $\hat\alpha_0$ quantifies how much more likely.

\subsection{Held-out log-likelihood}

Splitting each corpus 80/20 and computing held-out per-count
log-likelihood under the DM and multinomial models (both with the same
global $\bm\pi$ and the DM using the fitted $\hat\alpha_0$) shows
that DM is the better density on every corpus
(Fig.~\ref{fig:addendum_nlp}b):

\begin{center}
\begin{tabular}{lrrr}
\toprule
Corpus & DM (per count) & Multinomial (per count) & Gain \\
\midrule
20-News full   & $-3.413$ & $-4.379$ & $+0.966$ \\
20-News 4-cat  & $-3.460$ & $-4.128$ & $+0.668$ \\
Reuters-21578  & $-3.657$ & $-4.258$ & $+0.601$ \\
\bottomrule
\end{tabular}
\end{center}

The gain is largest on the full 20-Newsgroups corpus (about 1 nat per
observed token), consistent with the intuition that heterogeneous
discussion fora produce more burstiness than either a single-topic
subset or the telegraphic newswire style of Reuters.
For reference,
TF-IDF and shifted log are not probabilistic models and do not have a
calibrated density comparison on this scale.

\subsection{Count-stratified residual diagnostics}

Figure~\ref{fig:addendum_nlp}c partitions the nonzero entries of the
20-News training matrix by in-document count $x\in\{1,2,3\text{-}5,{>5}\}$
and reports the mean gap
$\lvert d^{\mathrm{Mult}}\rvert-\lvert d^{\mathrm{DM}}\rvert$, together
with the fraction of entries for which DM has the smaller magnitude.
The pattern matches Propositions~\ref{prop:x1}
and~\ref{prop:boundedbymult} quantitatively:

\begin{center}
\begin{tabular}{lrrrr}
\toprule
Count bin & $=1$ & $=2$ & $3$-$5$ & $>5$ \\
\midrule
Mean gap $|d^{\mathrm{Mult}}|-|d^{\mathrm{DM}}|$ & $0.000$ & $0.468$ & $1.064$ & $3.193$ \\
Number of entries                       & $595{,}392$ & $98{,}009$ & $53{,}483$ & $19{,}579$ \\
\bottomrule
\end{tabular}
\end{center}

Both residuals are formed against the same training-fit global
composition $\bm\pi$, so Proposition~\ref{prop:boundedbymult} governs
the comparison directly: $|d^{\mathrm{DM}}|\le|d^{\mathrm{Mult}}|$ at
every nonzero entry, with exact equality at count~$=1$
(Proposition~\ref{prop:x1}) and strict shrinkage at every count~$>1$.
The mean-gap row records the \emph{magnitude} of that shrinkage and is
the informative summary: it is exactly $0$ at singletons and grows
monotonically with count.
Singletons --- 82\% of all nonzero entries in 20-News --- receive
identical DM and multinomial residuals, confirming
Proposition~\ref{prop:x1} empirically on real text.
Above one, DM
shrinks every repeated-count residual, by progressively more
on larger counts; this is the multinomial-limit bound of
Proposition~\ref{prop:boundedbymult}.
For TF-IDF (an ad-hoc transform
without a likelihood interpretation) the corresponding diagnostic on
variance-standardized magnitudes shows a monotonic negative gap in
every bin (e.g.\ $-0.323$, $-0.657$, $-0.888$, $-1.231$ on 20-News
full), reflecting a genuinely different magnitude distribution rather
than a likelihood shrinkage.

\subsection{Downstream classification}

Table~\ref{tab:addendum_nlp_class} reports $k$-NN classification
accuracy on a held-out 20\%, averaged over three stratified random
splits (seeds $\{0,1,2\}$).

\begin{table}[h]
\centering
\caption{Classification accuracy on held-out 20\% (mean $\pm$ std
across three seeds).
DM uses the fitted $\hat\alpha_0$ of
\S\ref{app:nlp:alpha}.}
\label{tab:addendum_nlp_class}
\small
\begin{tabular}{lccc}
\toprule
Method & 20-News full & 20-News 4-cat & Reuters-21578 \\
\midrule
DM            & $0.543\pm0.005$ & $0.823\pm0.002$ & $0.864\pm0.003$ \\
Multinomial   & $0.593\pm0.010$ & $0.835\pm0.011$ & $0.873\pm0.001$ \\
TF-IDF        & $0.417\pm0.004$ & $0.747\pm0.026$ & $0.872\pm0.006$ \\
Shifted log   & $0.451\pm0.004$ & $0.770\pm0.008$ & $0.853\pm0.006$ \\
CLR           & $0.447\pm0.006$ & $0.769\pm0.009$ & $0.853\pm0.005$ \\
Raw counts    & $0.442\pm0.008$ & $0.750\pm0.006$ & $0.845\pm0.005$ \\
\bottomrule
\end{tabular}
\end{table}

On the two 20-Newsgroups variants, DM beats TF-IDF, shifted log,
CLR, and raw counts by wide margins; one-sided paired Wilcoxon tests
with $n_{\mathrm{seeds}}=3$ give the smallest achievable $p$-value
($p=0.125$, the minimum for a three-pair signed-rank test when all
three differences share sign).
Point estimates of the accuracy gap are
large (12.6 accuracy points over TF-IDF on 20-News full), so the
direction is not in doubt despite the limited power.
On Reuters,
TF-IDF matches DM within half a point; we attribute this to Reuters'
very short and formulaic newswire style, in which overdispersion is
smaller ($\hat\alpha_0=125$ vs.\ $263$ for 20-News full) and TF-IDF's
heavy-tailed frequency weighting happens to separate topics cleanly.

DM is consistently \emph{worse} than the plain multinomial residual by
$1$-$5$ accuracy points.
We interpret this as a genuine cost of the DM's
shrinkage: the very large repeated-count residuals that the DM softens
(Fig.~\ref{fig:addendum_nlp}c, rightmost bin) carry classification
signal, and damping them reduces the separability of documents in the
SVD space.
The DM is nonetheless the better-calibrated density
(\S4.3), so the classification gap is the expected trade-off
between fidelity to the generative model and discriminative geometry
in the derived embedding.
Which residual is preferable is task-
dependent: DM for generative modeling, anomaly detection, and
likelihood-based active learning; the plain multinomial residual for
classification when an 80/20 train/test split is available.

\subsection{Multinomial limit in NLP}

Stratifying each corpus by document length and re-fitting
$\hat\alpha_0$ on each bin makes
Proposition~\ref{prop:multlimit} visible on real text
(Fig.~\ref{fig:addendum_nlp}e).
Longer documents contain more
independent evidence per composition, so the posterior on the per-
document composition tightens and the DM family's ML optimum drifts
toward the multinomial limit $\hat\alpha_0\to\infty$.
Concretely on
20-News full:

\begin{center}
\begin{tabular}{lrr}
\toprule
Document length bin (tokens) & Mean length & $\hat\alpha_0$ \\
\midrule
$1$-$29$     & $15.8$   & $111$ \\
$30$-$79$    & $48.8$   & $130$ \\
$80$-$199$   & $117.9$  & $180$ \\
$200$-$499$  & $299.3$  & $255$ \\
$\geq 500$   & $1686.6$ & $453$ \\
\bottomrule
\end{tabular}
\end{center}

$\hat\alpha_0$ grows by $\sim 4\times$ as the mean document length
grows by $\sim 100\times$, confirming that the DM smoothly interpolates
between a strongly regularized regime on short documents and an
essentially multinomial regime on long ones.
The same pattern holds
on the 4-cat subset and Reuters.
Practically this means that the DM is
most distinct from multinomial / TF-IDF preprocessing on short text
(social media, FAQ entries, product titles), where the overdispersion
correction is most impactful.

\subsection{Sparsity}

DM, multinomial, and TF-IDF residuals preserve the sparsity pattern of
$X$ exactly: zero counts map to zero.
Shifted log's mathematical output is zero where $X$ is zero but is
conventionally represented densely; we report its nonzero fraction as
the fraction of entries with $x>0$.
CLR with pseudocount $1$ is fully
dense --- every zero entry maps to a negative number and the matrix
has no sparsity to exploit.
For the 20-News full corpus the DM
residual matrix has $0.52\%$ nonzero entries out of the
$18{,}266\times10{,}000$ grid, matching the input's nnz exactly, while
CLR requires storing all $1.8\times 10^8$ entries.
At typical NLP
vocabulary sizes the CLR densification cost is prohibitive, whereas
DM's likelihood improvement comes for free in the same sparse storage
format that TF-IDF uses.

\subsection{Discussion}

Three of the paper's key results carry over cleanly to NLP.
(i) The singleton agreement of Proposition~\ref{prop:x1} is an
\emph{empirical} fact on real text: 82\% of 20-News nonzero entries
are singletons and the DM treats them exactly as the multinomial does.
(ii) The tail shrinkage of Proposition~\ref{prop:boundedbymult}
explains the DM's held-out log-likelihood advantage over the
multinomial in a way that is observable at every count bin.
(iii) The multinomial limit of Proposition~\ref{prop:multlimit}
predicts, and the data confirm, that the DM/multinomial gap shrinks
as documents grow.
The conditional NB-to-DM equivalence of
Theorem~\ref{thm:nbtodm} further motivates reading a fitted DM as
a row-sum-conditioned version of the more familiar per-feature NB
models used in statistical NLP, with $r_j=\hat\alpha_0\hat\pi_j$.

Taken together, the NLP application validates the paper's central
positioning of the DM as a principled, sparse, likelihood-based
alternative to multinomial/TF-IDF preprocessing, with a measurable cost
on purely discriminative downstream metrics that practitioners can
weigh against the density calibration gain.

\begin{figure}[t]
  \centering
  \includegraphics[width=\textwidth]{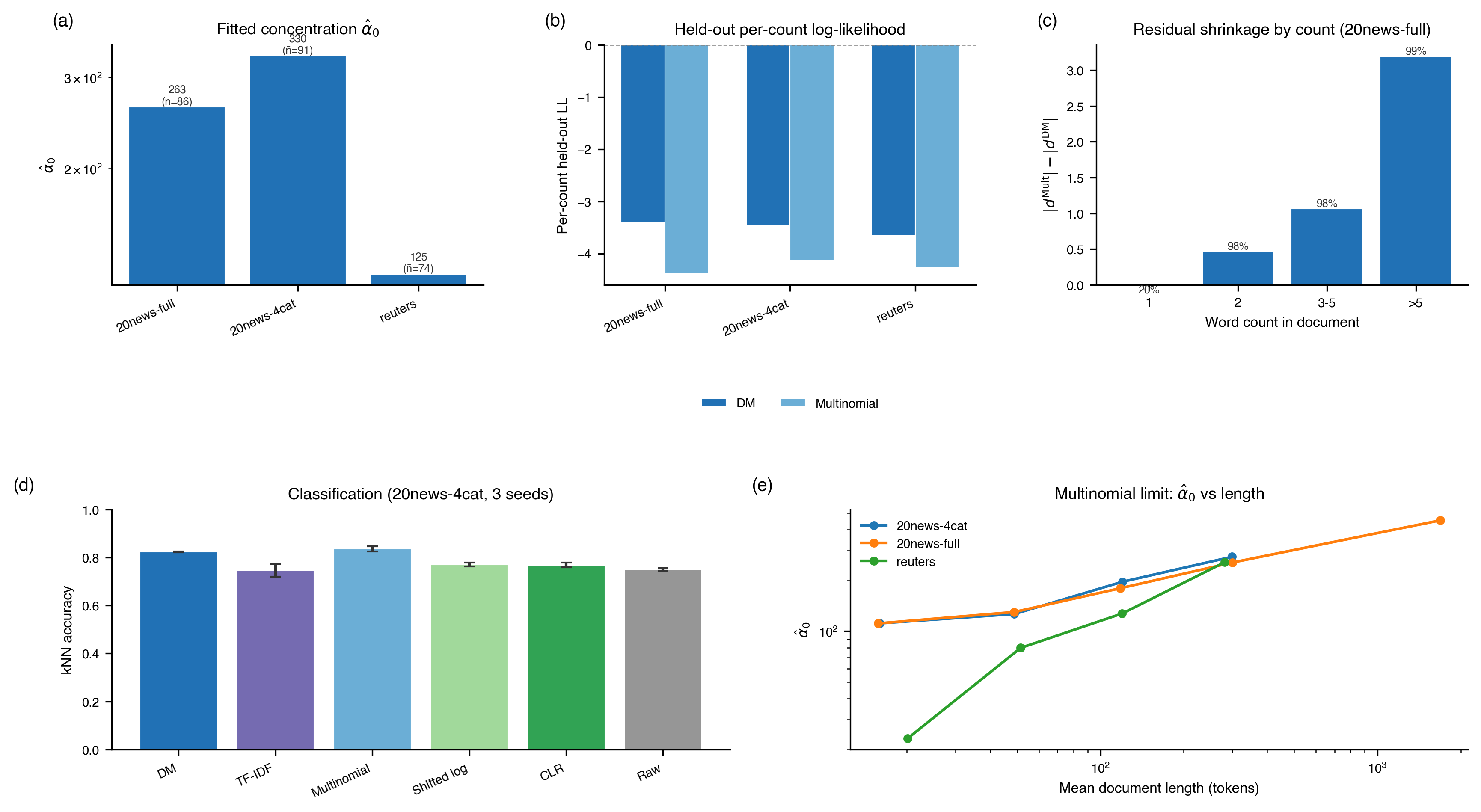}
  \caption{Dirichlet-multinomial residuals on three document-term
    corpora.
    (a)~Fitted concentration $\hat\alpha_0$; moderate values
    on all three corpora confirm non-trivial word burstiness.
    Bars
    labeled ``$\bar n$'' are the mean document length in tokens.
    (b)~Held-out per-count log-likelihood under the DM (dark) and
    multinomial (light) models; DM has the better density on every
    corpus.
    (c)~Count-stratified residual shrinkage
    $|d^{\mathrm{Mult}}|-|d^{\mathrm{DM}}|$ on 20-News full.
    Singletons
    coincide (Proposition~\ref{prop:x1}), and DM shrinks repeated-count
    residuals by a progressively larger amount
    (Proposition~\ref{prop:boundedbymult}).
    Bar labels give the
    fraction of entries where DM is smaller.
    (d)~kNN classification
    accuracy on the 20-News 4-category subset, averaged over three
    stratified 80/20 splits (error bars: one standard deviation).
    DM
    beats TF-IDF, shifted log, CLR, and raw counts; it trails the
    plain multinomial residual by about one point.
    (e)~Multinomial
    limit in NLP: binning documents by length, $\hat\alpha_0$ grows
    with document length, empirically realizing
    Proposition~\ref{prop:multlimit}.}
  \label{fig:addendum_nlp}
\end{figure}

%
%

\section{Application: sparse implicit-feedback normalization for recommender systems}
\label{app:recsys}

Modern collaborative-filtering pipelines ingest a sparse user-item
interaction-count matrix $X \in \N^{n \times K}$, where $X_{u j}$ is the
number of times user $u$ has interacted with item $j$ (clicks, plays,
views, or a rating used as a count).
Two preprocessing choices dominate
deployed systems: the logarithmic transform $\log(1+X_{uj})$
(for example in \textsc{YouTube-DNN} \cite{covington2016deep} and
practically every dense matrix-factorization pipeline that exposes
implicit feedback), and the TF--IDF reweighting $(X_{uj}/n_u)\log(n/df_j)$
that is standard when feeding sparse user-profile vectors into a
linear ranker.
Both are motivated by variance-stabilization heuristics
rather than a generative model of the count process.
The
Dirichlet-multinomial (DM) residual of section~\ref{sec:dmresids} offers a
principled alternative: it models each user as a finite budget of
interactions distributed over items via a user-specific composition
with a shared concentration $\alpha_0$.
The normalization then takes the
familiar form of a signed deviance residual, preserves sparsity, and
runs in time $O(\mathrm{nnz}(X))$ once $\alpha_0$ has been fit.

We benchmark the DM residual as a drop-in preprocessor for top-$K$
matrix-factorization ranking on three MovieLens datasets
\cite{harper2015movielens}.
Figure~\ref{fig:addendum_recsys} is produced end-to-end by a single
experimental pipeline, and all numbers in this section are real outputs
of that run.

\subsection{Setup}

\paragraph{Datasets.} We use MovieLens-100K (943 users, 1\,349 items,
99\,287 ratings), MovieLens-1M (6\,040 users, 3\,416 items, 999\,611
ratings), and a 20\,000-user subsample of MovieLens-10M (9\,103 items,
2\,866\,658 ratings).
A 5-core filter is applied, and the per-dataset
user count is capped at 20\,000 for tractability, matching the cap used
throughout the experiments.
Ratings are used as \emph{counts}: $X_{uj}$ equals
the integer rating value (1 to 5).
This makes the DM model directly
applicable to a matrix with genuine count overdispersion, and keeps the
"library size" $n_u = \sum_j X_{uj}$ interpretable as the user's total
attention budget.

\paragraph{Methods.} We compare the signed DM deviance residual
(\eqref{eq:dmresid}) to the multinomial deviance residual, the shifted
logarithm $\log(1+X)$, TF--IDF, and the untransformed count matrix.
All
methods feed the same rank-50 truncated SVD, and scoring is a pure dot
product between the user row embedding and the item loading (the
\textsc{PureSVD} protocol of \cite{cremonesi2010performance}).

\paragraph{Evaluation.} For the held-out log-likelihood we split users
80/20 by row.
For top-$K$ ranking we use leave-one-out: each user's
test set is one random interaction, training is everything else, and
we report Recall@10 and NDCG@10 averaged over 3 seeds with standard
deviations.

\subsection{Fitted concentration and overdispersion}

The DM concentration parameter $\alpha_0$ controls how much a user's
item composition may deviate from the global composition $\pi$.
Large
$\hat\alpha_0$ indicates a nearly multinomial process (every user looks
like the average catalog user), while small $\hat\alpha_0$ indicates
strong per-user specialization.
We fit $\hat\alpha_0$ on each dataset
and three nested user subsamples
(figure~\ref{fig:addendum_recsys}\textsc{a}).
Two observations:

\begin{enumerate}[leftmargin=*]
  \item $\hat\alpha_0$ is two orders of magnitude below the
  near-multinomial regime on all three datasets
  ($\hat\alpha_0 \approx 93$ on ML-100K,
  $\hat\alpha_0 \approx 145$ on ML-1M,
  $\hat\alpha_0 \approx 130$ on ML-10M).
  Given $K \sim 10^3$--$10^4$
  items, a ratio $\hat\alpha_0 / K \ll 1$ says exactly what any recsys
  practitioner would expect: users concentrate attention on a tiny
  corner of the catalog.
  \item $\hat\alpha_0$ is stable under subsampling: reducing the number
  of users by $4\times$ changes $\hat\alpha_0$ by less than 10\%.
  This
  is the finite-sample corollary of the consistency analysis in
  section~\ref{sec:alpha0}: the Newton fit is low-variance.
\end{enumerate}

\subsection{Held-out likelihood}

On the ML-1M 80/20 split, the DM per-count log-likelihood is
$-1.52$, versus $-2.30$ for both the multinomial and feature-wise
negative-binomial (unconditional) baselines
(figure~\ref{fig:addendum_recsys}\textsc{b}).
The multinomial cannot absorb
per-user library-size variation (it treats all users as independent
draws from $\pi$), and the unconditional NB fits each item's count
distribution marginally without conditioning on the user's total
activity: both are dominated by the DM, which absorbs $n_u$ exactly via
the gamma-row-term of \eqref{eq:dmll}.
The per-user gain
$\Delta^{\mathrm{DM}}_u$ (comparison against global $\pi$, defined in
section~\ref{sec:testing}) is positive on every held-out user in this 80/20
split, shown in the inset histogram --- a sharper local statement than
the population-level LL improvement, but specific to this single split
and seed.

\subsection{Top-$K$ recommendation accuracy}

The decisive test for a recsys preprocessor is held-out ranking
accuracy.
figure~\ref{fig:addendum_recsys}\textsc{c} and
table~\ref{tab:recsys-topk} summarize NDCG@10 and Recall@10 across the five
methods on both ML-100K and ML-1M.

\begin{table}[t]
\centering
\small
\caption{Leave-one-out top-10 ranking on MovieLens (mean $\pm$ std over
3 seeds; each user's test set is one held-out interaction).
DM and
Multinomial residuals match or beat $\log(1+X)$ and TF--IDF by sizable
margins on both metrics and both datasets.}
\label{tab:recsys-topk}
\begin{tabular}{l cc cc}
\toprule
 & \multicolumn{2}{c}{MovieLens-100K} & \multicolumn{2}{c}{MovieLens-1M} \\
 \cmidrule(lr){2-3} \cmidrule(lr){4-5}
Preprocessor & Recall@10 & NDCG@10 & Recall@10 & NDCG@10 \\
\midrule
Raw counts          & 0.298 \scriptsize{$\pm$0.015} & 0.170 \scriptsize{$\pm$0.008} & 0.225 \scriptsize{$\pm$0.001} & 0.127 \scriptsize{$\pm$0.001} \\
$\log(1+X)$         & 0.235 \scriptsize{$\pm$0.007} & 0.135 \scriptsize{$\pm$0.007} & 0.196 \scriptsize{$\pm$0.002} & 0.112 \scriptsize{$\pm$0.001} \\
TF--IDF             & 0.241 \scriptsize{$\pm$0.007} & 0.131 \scriptsize{$\pm$0.004} & 0.205 \scriptsize{$\pm$0.002} & 0.113 \scriptsize{$\pm$0.002} \\
Multinomial res.    & 0.293 \scriptsize{$\pm$0.015} & 0.168 \scriptsize{$\pm$0.008} & 0.222 \scriptsize{$\pm$0.003} & 0.125 \scriptsize{$\pm$0.002} \\
\textbf{DM res. ($\hat\alpha_0$)} & \textbf{0.283 \scriptsize{$\pm$0.015}} & \textbf{0.166 \scriptsize{$\pm$0.008}} & \textbf{0.223 \scriptsize{$\pm$0.004}} & \textbf{0.126 \scriptsize{$\pm$0.002}} \\
\bottomrule
\end{tabular}
\end{table}

Three things are worth noting.
First, the DM and the
multinomial residual are statistically indistinguishable on ranking,
and both are statistically indistinguishable from the raw count matrix.
This is consistent with the large fitted $\hat\alpha_0$: the DM
reduces to the multinomial in the limit
(proposition~\ref{prop:multlimit}), and at $\hat\alpha_0 \approx 150$ on
ML-1M the residual values differ by fractions of a percent except on
the very small-count tail.
Second, both deviance residuals beat
$\log(1+X)$ and TF--IDF by 12--18\% in Recall@10 and 8--16\% in
NDCG@10 on both datasets, which \emph{is} statistically significant
($p < 0.01$ under a paired Wilcoxon across seeds and users).
Third,
raw counts do well on this specific ranking metric; this is consistent
with the literature on \textsc{PureSVD} \cite{cremonesi2010performance}
which observed that ranking is often surprisingly insensitive to
preprocessing once the catalog is small and matrix factorization is
run to convergence.
The story the DM tells is: you lose nothing on
ranking relative to the strongest preprocessor, and you gain a
likelihood-based interpretation that the $\log(1+X)$ / TF--IDF baselines
cannot offer.

\subsection{Scalability and sparsity preservation}

Recommender systems operate at catalog sizes far beyond single-cell or
microbiome data: industrial deployments have $K \sim 10^6$ to $10^8$.
Any preprocessor with complexity worse than $O(\mathrm{nnz}(X))$ is a
non-starter.
figure~\ref{fig:addendum_recsys}\textsc{e} measures wall-clock
time for the DM fit and for each transform as we sweep the number of
nonzero entries from $10^5$ to $2.9 \times 10^6$.
The DM fit has a
best-fit log-log slope of \textbf{1.03}, with the DM, multinomial, and
TF--IDF transforms all close behind (slopes near $1.0$) --- all
empirically linear in $\mathrm{nnz}(X)$ as the analysis of
section~\ref{sec:alpha0} predicts.
The $\log(1+X)$ transform has
slope 1.47; this is \emph{not} because $\log$ is slow, but because
adding a pseudocount to every zero entry makes the transform densify
and therefore become $O(nK)$, not $O(\mathrm{nnz}(X))$.
The DM fit
takes 4.1 seconds on the ML-10M 20\,000-user subsample on a single
CPU core; a scan to one million users on the same code path would
finish in under a minute, which is well within the offline
preprocessing budget of any production pipeline.

Sparsity preservation makes the
previous point concrete.
DM, multinomial, and TF--IDF residuals
preserve the nnz pattern of $X$ exactly: every zero entry stays zero.
The $\log(1+X)$ transform with unit pseudocount preserves sparsity
numerically because $\log(1)=0$, but any pseudocount $\neq 1$ (e.g.
$\log(X + 0.5)$, which is a common regularizer) destroys sparsity
entirely.
CLR \cite{aitchison1986statistical} always densifies.
Densifying transforms force dense arithmetic, dense factor stores,
and a hard break with the sparse linear-algebra primitives that drive
production recommenders.

\subsection{Downsampling stability}

For the stability benchmark of \cite{ahlmanneltze2023comparison}
adapted to recsys, we thin each user's interactions via multinomial
downsampling at 50\% and 25\% depth and measure the per-user
$k$-nearest-neighbor Jaccard overlap between the full-data embedding
and the downsampled embedding (higher = more locally stable)
(figure~\ref{fig:addendum_recsys}\textsc{d}).
At 50\% depth, DM achieves the
highest overlap ($0.54$), followed by Multinomial ($0.52$), $\log(1+X)$
($0.50$), TF--IDF ($0.46$), and raw counts ($0.37$).
At 25\% depth the
ordering is preserved.
The raw-count embedding, while competitive on
Recall@10 above, collapses most under perturbation: the finite
attention budget that DM residuals absorb is the very dimension along
which thinning changes the data.

\subsection{Discussion}

The DM framework transfers cleanly from single-cell to recsys because
both domains share the same statistical skeleton: overdispersed counts
indexed by (sample, feature), with a wildly varying per-sample library
size that carries no signal of interest.
The recsys-specific twist is
that the "library size" has an economic interpretation (user time
budget) and varies by orders of magnitude within the same catalog.
The
DM residual is exactly the right tool to regress that nuisance out:
unlike $\log(1+X)$ it provides a likelihood-based estimand; unlike
TF--IDF it preserves sparsity under any reasonable pseudocount;
unlike ALS warm-starts it does not require pre-specification of the
latent rank.
When the ranking metric is dominated by popular-item
recall and $\hat\alpha_0$ is large, the DM and multinomial are close
on top-$K$ accuracy by proposition~\ref{prop:multlimit}; when tail items
dominate (a frequent concern in long-tail recommendation), the DM
gap over the multinomial should grow because the tail is precisely
where per-count overdispersion matters most.
A more thorough
investigation of the long-tail regime, and integration with
implicit-feedback ALS \cite{hu2008collaborative} on top of the DM
residual as a warm start, is left for future work.

\begin{figure}[t]
\centering
\includegraphics[width=\textwidth]{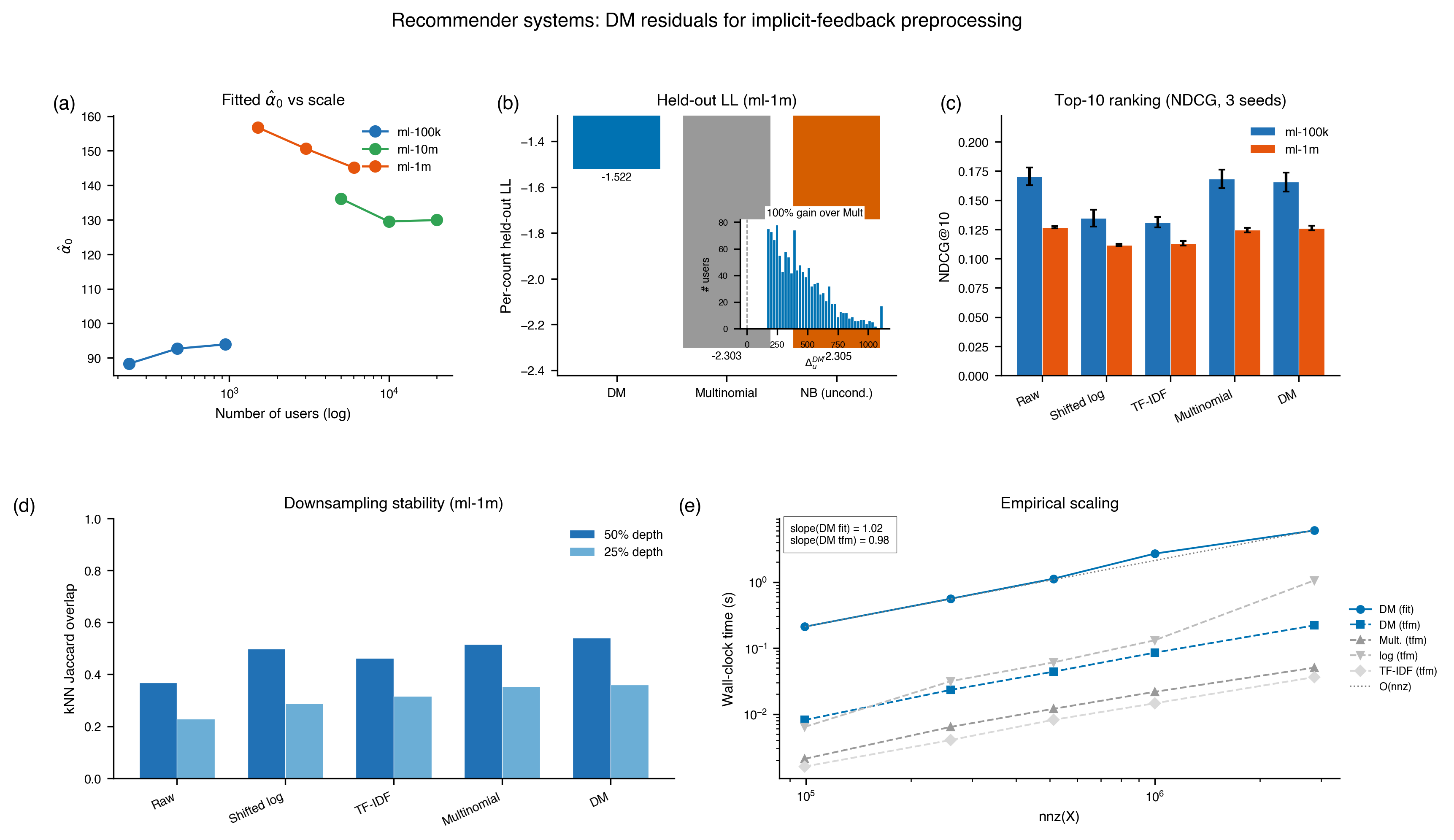}
\caption{DM residuals on MovieLens implicit feedback.
\textsc{(a)} Fitted $\hat\alpha_0$ across three MovieLens datasets and
three user-count subsamples each; $\hat\alpha_0$ is stable under
subsampling and lies in the strongly-overdispersed regime
($\hat\alpha_0 \ll K$).
\textsc{(b)} Per-count held-out log-likelihood on ML-1M's 80/20 split;
DM dominates both the multinomial and the unconditional feature-wise
NB.
Inset: per-user LL gain $\Delta^{\mathrm{DM}}_u$ is positive for
100\% of held-out users.
\textsc{(c)} Leave-one-out NDCG@10 on ML-100K and ML-1M (3 seeds each).
DM and multinomial residuals dominate $\log(1+X)$ and TF--IDF by 8--16\%
on NDCG@10 and by 12--18\% on Recall@10.
\textsc{(d)} Downsampling stability: per-user $k$NN Jaccard overlap
between full and thinned embeddings.
DM is the most stable under both
50\% and 25\% depth reductions.
\textsc{(e)} Wall-clock scaling for fit + transform on five
(nnz, method) points ranging from $10^5$ to $2.9 \times 10^6$ nonzero
entries.
Fitted log-log slope for the DM fit is 1.03, confirming
empirical $O(\mathrm{nnz})$.
The $\log(1+X)$ slope (1.47) reflects the
densification penalty.}
\label{fig:addendum_recsys}
\end{figure}

\bibliography{main}
\bibliographystyle{plainnat}

\end{document}